\documentclass[11pt]{article}
\usepackage{fullpage}
\usepackage{amsthm}
\usepackage{url}
\usepackage{enumerate,paralist}
\usepackage{xspace}
\usepackage{amsmath,amssymb,amsfonts}
\usepackage{mathrsfs}
\usepackage{stmaryrd}
\usepackage{color}
\usepackage{epsfig}

\newtheorem{definition}{Definition}\newtheorem{theorem}{Theorem}\newtheorem{lemma}{Lemma}\newtheorem{proposition}{Proposition}

\newcommand{\email}[1]{\url{#1}}

\let\Paragraph=\paragraph \renewcommand{\paragraph}[1]{\Paragraph{{#1}.}}

\newcommand{\supp}{\operatorname{\mathsf{Supp}}}
\newcommand*{\Disc}{\mathsf{Disc}}
\newcommand*{\cross}{\times}
\newcommand*{\cons}{{:}}
\newcommand*{\restrict}[2]{\lfloor{#1}\rfloor_{#2}}
\newcommand*{\diff}[2]{\Delta({#1},{#2})}
\newcommand*{\last}{\mathsf{last}}

\newcommand*{\set}[2]{\{\,{#1}~|~{#2}\,\}}
\newcommand*{\st}{\,\mathrm{s.t.}\,}
\newcommand*{\range}[1]{\operatorname{\mathsf{range}({#1})}}
\newcommand*{\supseqeq}{\sqsupseteq}
\newcommand*{\supseq}{\sqsupset}
\newcommand*{\e}[1]{\exp({#1})}
\newcommand*{\lift}[1]{\mathcal{L}({#1})}
\newcommand*{\emptybag}{\mbox{$\{\!\!\{\}\!\!\}$}}
\newcommand*{\bag}[1]{\mathsf{bag}({#1})}
\newcommand*{\bagunion}{\uplus}
\newcommand*{\bigbagunion}{\biguplus}

\newcommand*{\argmax}{\mbox{$\mathrm{arg}\max$}}
\newcommand*{\model}[1]{\llbracket {#1} \rrbracket}
\newcommand*{\mathsc}[1]{\text{\textsc{{#1}}}}
\newcommand{\rel}{\mathtt{rel}}
\newcommand*{\rels}{\mathtt{Rels}}
\newcommand*{\nil}{\mathtt{nil}}

\newcommand*{\exparamM}{M_{\mathsf{ex1}}}
\newcommand*{\exparamL}{L_{\mathsf{ex1}}}
\newcommand*{\excsM}{M_{\mathsf{ex2}}}

\newcommand*{\m}{m}
\newcommand*{\M}{M}
\newcommand*{\K}{K}

\DeclareMathOperator*{\transop}{\rightarrow}
\newcommand*{\trans}{\transop\limits}
\DeclareMathOperator*{\transopone}{\rightarrow_1}
\newcommand*{\transone}{\transopone\limits}
\DeclareMathOperator*{\transoptwo}{\rightarrow_2}
\newcommand*{\transtwo}{\transoptwo\limits}
\DeclareMathOperator*{\transopthree}{\rightarrow_3}
\newcommand*{\transthree}{\transopthree\limits}

\DeclareMathOperator*{\wtransop}{\Rightarrow}
\newcommand*{\wtrans}{\wtransop\limits}

\title{Formal Verification of Differential Privacy for Interactive Systems\thanks{This work was partially supported by the U.S. Army Research Office contract on Perpetually Available and Secure Information Systems (DAAD19-02-1-0389) to Carnegie Mellon CyLab, the NSF Science and Technology Center TRUST, the NSF CyberTrust grant ``Privacy, Compliance and Information Risk in Complex Organizational Processes,'' and the AFOSR MURI ``Collaborative Policies and Assured Information Sharing.''}}

\author{Michael Carl Tschantz\\
Computer Science Department\\
Carnegie Mellon University\\
5000 Forbes Avenue\\
Pittsburgh, PA 15213\\
\email{mtschant@cs.cmu.edu}
\and Dilsun Kaynar\\
CyLab\\
Carnegie Mellon University\\
5000 Forbes Avenue\\
Pittsburgh, PA 15213\\
\email{dilsunk@cmu.edu}
\and Anupam Datta\\
CyLab\\
Carnegie Mellon University\\
5000 Forbes Avenue\\
Pittsburgh, PA 15213\\
\email{danupam@cmu.edu}}

\begin{document}

\maketitle

\begin{abstract}
Differential privacy is a promising approach to privacy preserving data analysis with a well-developed theory for functions.  Despite recent work on implementing systems that aim to provide differential privacy, the problem of formally verifying that these systems have differential privacy has not been adequately addressed.  This paper presents the first results towards automated verification of source code for differentially private interactive systems.  We develop a formal probabilistic automaton model of differential privacy for systems by adapting prior work on differential privacy for functions.  The main technical result of the paper is a sound proof technique based on a form of probabilistic bisimulation relation for proving that a system modeled as a probabilistic automaton satisfies differential privacy.  The novelty lies in the way we track quantitative privacy leakage bounds using a relation family instead of a single relation.  We illustrate the proof technique on a representative automaton motivated by PINQ, an implemented system that is intended to provide differential privacy.  To make our proof technique easier to apply to realistic systems, we prove a form of refinement theorem and apply it to show that a refinement of the abstract PINQ automaton also satisfies our differential privacy definition.  Finally, we begin the process of automating our proof technique by providing an algorithm for mechanically checking a restricted class of relations from the proof technique.
\end{abstract}

\section{Introduction}
\label{sec:intro}

\newcommand{\cut}[1]{}

\paragraph{Differential Privacy}
Differential privacy is a promising approach to privacy-preserving data analysis (see~\cite{d08survey,Dwork10} for surveys).  This work is motivated by statistical data sets that contain personal information about a large number of individuals (e.g., census or health data).  In such a scenario, a trusted party collects personal information from a representative sample with the goal of releasing statistics about the underlying population while simultaneously protecting the privacy of the individuals.  
In an interactive setting, an untrusted data examiner poses queries that the trusted party evaluates over the data set and appropriately modifies to protect privacy before sending the result to the examiner.  Differential privacy formalizes this operation in terms of a probabilistic \emph{sanitization function} that takes the data set as input.  
Differential privacy requires that the probability of producing an output should not change much irrespective of whether information about any particular individual is in the data set or not.  The amount of change is measured in terms of a \emph{privacy leakage bound}---a non-negative real number $\epsilon$, where a smaller $\epsilon$ indicates a higher level of privacy.
The insight here is that since only a limited amount of additional privacy risk is incurred by joining a data set, individuals may decide to join the data set if there are societal benefits from doing so (e.g., aiding cancer research).  A consequence and strength of the definition is that the privacy guarantee holds irrespective of the auxiliary information and computational power available to an adversary.  
Previous work on algorithms for sanitization functions and the analysis of these algorithms in light of the trade-offs between privacy and utility (answering useful queries accurately without compromising privacy) has provided firm foundations for differential privacy (e.g.~\cite{dmns06calibrating,d06differential,mt07mechanism,nrs07smooth,blr08,d08survey,d09differential,grs09universally,Dwork10,DNPR10}).  

In a different direction, these sanitization algorithms are being implemented for inclusion in data management systems.  For example, \textsc{pinq} resembles a \textsc{sql} database, but instead of providing the actual answer to \textsc{sql} queries, it provides the output of a differentially private sanitization function operating on the actual answer~\cite{m09privacy}.  Another such system, \textsc{airavat}, manages distributed data and performs MapReduce computations in a cloud computing environment while using differential privacy as a basis for declassifying data in a mandatory access control framework~\cite{rrsksw10airavat}.  Both of these are interactive systems that use sanitization functions as a component: they interact with both the providers of sensitive data and untrusted data examiners, store the data, and perform computations on the data some of which apply sanitization functions. Even if we assume that these systems correctly implement the sanitization functions to give differential privacy, this is not sufficient to conclude that the guarantees of differential privacy apply to the system as a whole. For the differential privacy guarantee of functions to scale to the whole of the implemented system, the system must properly handle the sensitive data and never provide channels through which untrusted examiners can infer information about it without first sanitizing it to the degree dictated by the privacy error bound.

\paragraph{Formal Methods for Differential Privacy}
We work toward reconciling formal analysis techniques with the growing body of work on abstract frameworks or implemented systems that use differential privacy as a building block.  While prior work in the area has provided a type system for proving that a non-interactive program is a differentially private sanitization function~\cite{rp10distance}, we know of no formal methods for proving that an interactive system using such functions has differential privacy.  Applying formal methods to interactive systems ensures that these systems properly manage their data bases and interactions with untrusted users.

Formal verification that an interactive system provides privacy requires that the system be modeled in such a way that the correspondence between the system and model is evident and the model includes all relevant behavior of the system. 
Once formal verification is done on the model, one can assert with the confidence afforded by formal proofs that the system as implemented and modeled preserves privacy in addition to knowing that the algorithms implemented by the system preserve privacy.
For formal verification to scale to large programs with complex models, the creation of the model and the verification of its privacy must be mechanized, preferably in a compositional manner.  

To this end, we present an automaton model for which the correspondence between the automaton and the implementation of a system is so plainly evident that the automaton could be automatically extracted from source code as is done with model checking~\cite{cgp00model}.  
For this model, we introduce a form of compositional reasoning that allows us to separate the proof that a function gives differential privacy from the proof that the system correctly uses that function.
Furthermore, we present a proof technique for such models that is amenable to mechanization and an algorithm that can be used to check that the proof technique is correctly applied to a model.  

Our effort can be likened to those efforts in the security community that involve the development of formal models for cryptographic protocols and the accompanying verification methods~\cite{st07,backes07reactive,CCKLLPS08}.  These works use stylized proofs with multiple levels of abstraction and compositionality to enable scaling mechanical checking of these proofs to the size of realistic systems.  Making these proofs shorter or more readable for humans than their informal counterparts is not a goal.

\paragraph{Contributions}
We work with a special class of probabilistic I/O automata that allow us to model interactive systems in terms of states and probabilistic transitions between states. These automata provide us with the needed expressive power for modeling how data is stored in an internal state of an implementation, and how it is updated through computations, some of which apply differentially private sanitization functions on data. In Section~\ref{sec:autom}, we present this probabilistic automaton model and our differential privacy definition for probabilistic automata, which we call \emph{differential noninterference} due to the similarities it has with the information flow property \emph{noninterference}~\cite{gm82security}.  Indeed, when applied to interactive systems, both differential privacy and noninterference privacy aim at restricting information leakage about sensitive data by requiring that the system produces similar outputs for inputs that differ only in sensitive data.  However, differential privacy allows for the degree of similarity to decrease as the inputs diverge, making it a more flexible requirement.

As formal methods can only scale to large systems with compositional reasoning, in Section~\ref{sec:compo}, we examine the ability to perform compositional reasoning with our formal model.  We show that correctness proof of sanitization functions may be separated from the correctness proof of the system that uses them.

Our main technical contribution, presented in Section~\ref{sec:unwinding}, is a \emph{proof technique} for establishing that a system has differential noninterference. Our technique allows the global property of differential noninterference to be proved from local information about transitions between states. This proof technique was inspired by the \emph{unwinding} proof technique originally developed for proving that a system has noninterference~\cite{gm84unwinding}.

Our unwinding technique is also similar to bisimulation-based proof techniques as both uses a notion of ``similarity'' of states with respect to their observable behavior. 
Unlike traditional bisimulation relations for probabilistic automata, the unwinding relation is defined over the states of a single automaton with the intention of establishing the similarity of two states where one is obtainable from the other by the input of an additional data point.  Moreover, the notion of similarity is approximate, which is in keeping with the definition of differential privacy.  An unwinding proof involves finding a relation family indexed by the set of possible values of the privacy leakage bound $\epsilon$, rather than a single relation.  This departure from traditional probabilistic bisimulations is needed to track the maximum privacy leakage tolerable from a given state in the execution. 
We prove the soundness of our proof technique in Theorem~\ref{thm:unwinding}, which roughly states that the existence of appropriate $\epsilon$-unwinding families for an automaton $\M$ implies that $\M$ has $\epsilon$-differential noninterference.  

As in other formal proof techniques of this nature, the real creativity in doing the proofs with our technique goes into defining the unwinding family.  Unsurprisingly, the rest consists of repeated, routine applications of basic arguments showing that the defined relation between states is preserved by transitions of the system.   In Section~\ref{sec:checking-algo}, this quality enables us to develop an algorithm to check whether a given relation family is an unwinding family, thereby automating proofs for differential noninterference modulo the definition of the relations.  
We prove that the algorithm soundly runs in polynomial time: it will only return true if the automaton has $\epsilon$-differential noninterference (Theorems~\ref{thm:checking-algo-cover-sound} and~\ref{thm:checking-algo-cover-time}).

To motivate our work, we start by presenting a system similar to \textsc{pinq}.  We refer to the example system throughout our paper as we model it in our formalism and use our unwinding technique and algorithm to verify that it has differential noninterference.  As \textsc{pinq} may be configured to use any set of sanitization functions, we present an automaton $\exparamM$ that is parametric in the sanitization functions that it uses.  We show two methods for proving differential noninterference for any correct instantiation of $\exparamM$ with differentially private sanitization functions: by using the composition method presented in Section~\ref{sec:compo}, and by using our unwinding verification algorithm.  This second method illustrates the applicability of our algorithm in proving differential noninterference for interesting automata. 

Along the way, we find interactions between a bounded memory model and differential privacy of interest beyond formal verification.  In particular, we find the inability to store an unbounded number of data points results in doubling the privacy leakage.  

We finish with Section~\ref{sec:related-work} covering related work and Section~\ref{sec:future-work} presenting future work and conclusions.


\section{Background and Motivation}
\label{sec:diffp-systems}

\subsection{Differential Privacy}

Differential privacy formalizes the idea that a private process should not reveal too much information about a single person.  A \emph{data point} represents all the information collected about an individual (or other entity that must be protected).  A multiset (bag) of data points forms a \emph{data set}.  A \emph{sanitization function} $\kappa$ processes the data set and returns a result to the untrusted data examiner that should probabilistically not change whether or not a single data point is in the data set.
Dwork~\cite{d06differential} states differential privacy as follows:
\begin{definition}[Differential Privacy]
\label{def:Dwork-diff-priv}
A randomized function $\kappa$ has $\epsilon$-differential privacy iff for all data sets $B_1$ and $B_2$ differing on at most one element, and for all $S \subseteq \range{\kappa}$,
\[ \Pr[\kappa(B_1) \in S] \leq \e{\epsilon} * \Pr[\kappa(B_2) \in S]\]
\end{definition}
Formally, multisets $B_1$ and $B_2$ differ on at most one element iff either $B_1 = B_2$ or there exists $d$ such that $B_1 \cup \{d\} = B_2$ or $B_2 \cup \{d\} = B_1$.  Note that the above definition is well-defined only if $\range{\kappa}$ is countable. 

Differential privacy has many pleasing properties.  For example, if $B_1$ and $B_2$ differ by $n$ datapoints instead of just one, then the probabilities of $\kappa(B_1)$ and $\kappa(B_2)$ being in a set $S$ will be within a factor of $\e{n*\epsilon}$ of one another~\cite[Corollary 4]{mt07mechanism}.  Furthermore, a function that sequentially applies $n$ functions each with $\epsilon$-differential privacy and provides all of their outputs is an $(n*\epsilon)$-differentially private function~\cite[Corollary 5]{mt07mechanism}.

\paragraph{Privacy Mechanisms}
As shown in the original work on differential privacy, given a statistic $f$ that can be computed of the data sets $B_i$, one can construct a sanitization function $\kappa_f$ from $f$ by having $\kappa_f$ add noise to the value of $f(B_i)$ where the noise is drawn from a Laplace distribution~\cite{dmns06calibrating}.  This is an example of a \emph{privacy mechanism}, a scheme for converting a statistic into a sanitization function with differential privacy.

Systems in practice would implement a sanitization function such as $\kappa_f$ as a program.
As actual computers have only a bounded amount of memory, the program computing $\kappa_f$ must only use a bounded amount of memory.
However, many sanitization functions proposed in the differential privacy literature, including all sanitization functions constructed using the Laplace privacy mechanism, use randomly drawn real numbers, which requires an uncountably infinite number of states.   While such functions can be approximated using a finite number of states (e.g., by using floating point numbers), it is unclear whether the proofs that these functions have differential privacy carry over to their approximations.

As we are interested in formally proving that finite systems provide differential privacy, we limit ourselves to privacy mechanisms that operate over only a finite number of values.  One such mechanism is the \emph{Truncated Geometric Mechanism} of Ghosh et al.~\cite{grs09universally}, which uses noise drawn from a bounded, discrete version of the Laplace distribution.  As we are interested in applying formal methods to systems using such mechanisms, we provide an implementation of this mechanism that runs in expected constant time and proofs about it in Appendix~\ref{app:bdm}.

\subsection{Motivating Example System}
\label{sec:mot-ex}

To further motivate and illustrate our work, we provide an example of an interactive system that uses sanitization functions.  Throughout the remainder of this paper, we apply the various formal methods we develop to prove that it preserves privacy.  The system manages data points entered by data providers and processes requests of data examiners for information by receiving queries and answering them after sanitizing the answer computed over the data set.  The system must apply the sanitization functions to the data set and interact with the data examiner in a manner that does not compromise privacy.

Possible source code for one such system is shown in Figure~\ref{fig:code-window}.
To be concrete, suppose that the data points are integers and the system handles only two queries.  The first produces the output of the sanitization function \mathsc{count}, which provides the number of data points currently in the data base.  The second produces the output of \mathsc{sum}, which provides their sum.  In both cases, the sanitization functions use the Truncated Geometric Mechanism to preserve privacy~\cite{grs09universally}.  (Appendix~\ref{app:bdm-count-sum} provides source code for \mathsc{count} and \mathsc{sum}.)

\begin{figure}
\begin{verbatim}
01 dPts:= emptyArray(t);
02 numPts := emptyArray(t);
03 for(j:=0; j<t; j++)
04   dPts[j]:= emptyArray(maxPts);
05   numPts[j] := 0;
06 curSlot:=0;
07 while(1)
08   y:=input();
09   if(datapoint(y))
10     if(numPts[curSlot]<maxPts)
11       dPts[curSlot][numPts[curSlot]]:=y;
12       numPts[curSlot]++;
13   else
14     k:=get_sanitization_funct(y);
15     res:=k.compute(dPts);
16     print(res);
17     curSlot:=(curSlot + 1) mod t;
18     delete dPts[curSlot];
19     dPts[curSlot] := emptyArray(maxPts);
20     numPts[curSlot] := 0
\end{verbatim}
\caption{Program that tracks data point usage to ensure differential noninterference}
\label{fig:code-window}
\end{figure}
Intuitively, the program in Figure~\ref{fig:code-window} keeps an array of $t$ arrays of data points and a variable \verb|curSlot|, whose value indicates a (current) slot in the array.  
If the input is a data point, that data point is added to the array indexed by  \verb|curSlot| unless that array is full, in which case the data point is ignored.

If the input is a query, then the query requested by the input is computed on the union of all the data points collected from all the arrays. 
Line~\verb|15| uses either the implementation of \mathsc{count} or \mathsc{sum} to compute the system's response to the query \verb|y| where Line~\verb|14| selects the correct function.  
Furthermore, the index \verb|curSlot| to one of these arrays is cyclically shifted and the array to which it now points is replaced with an empty array.  Since there are only $t$ slots, this means that each array will only last for $t$ queries before being deleted.  (If $t = 0$, we take the program to have an array \verb|dPts| of length $0$, in which case it never stores any data points.)  Since each query has $\epsilon$-differential privacy, this ensures that each data point will only be involved in $t*\epsilon$ worth of queries.

\paragraph{Verification}
The goal of our work is to formally verify that systems like this one preserve the privacy of their users.  In addition to showing that the sanitization functions \textsc{count} and \textsc{sum} have differential privacy (a subject of previous work~\cite{grs09universally}), we study how the system leaks information about the data points in ways other than through the outputs from these functions.  Indeed, one might expect from the sequential result for differential privacy discussed above~\cite[Corollary 5]{mt07mechanism}, that the system would provide $(t*\epsilon)$-differential privacy.  However, due to how the system manages data points, it actually only provides $(2t*\epsilon)$-differential privacy as we show later.

Had our goal only been to formally verify the implementations of the sanitization functions \textsc{count} and \textsc{sum}, it would suffice to use a simple formal model such as that of probabilistic finite-state automata with no interaction and use a suitable algorithmic technique to verify differential privacy, which research on Markov chains provides.   (We provide further details in Section~\ref{sec:compo-algo}.)

However, to verify differential privacy for interactive systems that use privacy mechanism as a building block as the above system does, we need a more expressive formal model that models the interaction of the data examiner with the system and the addition of data points to the system over time.  The next section provides such a model.

\section{Modeling Interaction for Formal Verification}

In this section, we present the basics of the formal framework we use in modeling interactive systems and show how we can model the example system of Section~\ref{sec:mot-ex} using this formalism.  Specifically, in Sections~\ref{sec:autom} and~\ref{sec:mut-diff-priv}, we introduce a special class of probabilistic I/O automata and present our definition of differential privacy for this class of probabilistic I/O automata.  In Section~\ref{sec:ex-autom} we model the program of Figure~\ref{fig:code-window} as a probabilistic I/O automaton. 

\subsection{Automata}\label{sec:autom}

We use a simplified version of probabilistic I/O automata (cf.~\cite{lsv07observing}).  We define an automaton in terms of a probabilistic labeled transition system (\textsc{plts}).
\begin{definition} 
A probabilistic labeled transition system (\textsc{plts}) is a tuple $L = \langle S,I,O,\trans\rangle$ where $S$ is a countable set of states; $I$ and $O$ are countable and pairwise disjoint sets of actions, referred to as \emph{input} and \emph{output} actions respectively; and $\trans \subseteq S \times (I \cup O) \times \Disc(S)$ represents the possible transitions where $\Disc(S)$ is the set of discrete probability measures over $S$.  
\end{definition}

We use $A$ for $I \cup O$. We partition the input set $I$ into $D$, the set of data points, and $Q$, the set of queries.
We also partition the output set $O$ into $R$, the set of responses to the data examiner's queries and $H$, the set of outputs that are hidden from (not observable to) the data examiner. Note that $H$ includes outputs to the data provider.  We let $E$ range over all actions to which the examiner has direct access: $E = Q \cup R$.
When only one automaton is under consideration, we denote a transition $\langle s,a,\mu\rangle \in \trans$ by $s \trans^{a} \mu$.  

Henceforth, we require that \textsc{plts}s satisfy the following conditions:
\begin{itemize}
\item \emph{Transition determinism:} For every state $s \in S$ and action $a \in A$, there is at most one $\mu \in \Disc(S)$ such that $s \trans^{a} \mu$. 

\item \emph{Output determinism:} For every state $s \in S$, output $o \in O$, action $a \in A$, and $\mu \in \Disc(S)$, if $s \trans^o \mu$ and $s \trans^a \mu'$, then $a = o$ and $\mu' = \mu$.

\item \emph{Quasi-input enabling:} For every state $s \in S$, inputs $i_1$ and $i_2$ in $I$, and $\mu_1 \in \Disc(S)$, if $s \trans^{i_1} \mu_1$, then there exists $\mu_2$ such that $s \trans^{i_2} \mu_2$.
\end{itemize}
Output determinism and quasi-input enabling means that the state space may be partitioned into two parts: states that accept all of the inputs and states that produce exactly one output.  We require that each output producing state produces only one output since the choice of output should be made by the \textsc{plts} to avoid nondeterminism that might be resolved in a way that leaks information about the data set.  
Owing to transition determinism, we will often write $s \trans^a \mu$ without explicitly quantifying $\mu$.

We define an extended transition relation $\wtrans$ that describes how a \textsc{plts} may perform a sequence of actions where some of the output actions are hidden from the data examiner.  In particular, the hidden outputs in $H$ 
model unobservable internal actions irrelevant to privacy.
To define $\wtrans$, let a state that produces an output from $H$ be called \emph{$H$-enabled} and one that does not be called \emph{$H$-disabled}.  By output determinism, $H$-enabled states may only transition under an action in $H$ and, thus, cannot have transitions on actions from $R \cup Q \cup D$.  To skip over such states and focus on $H$-disabled states, which are more interesting from a verification point of view, we define $\wtrans$ to show to which $H$-disabled states the system may transition while performing any finite number of hidden actions.  
We define $s \wtrans^a \nu$ so that $\nu(s')$ is the probability of reaching the $H$-disabled state $s'$ from the state $s$ where $a$ is the action performed from state $s$.
Note that $\nu$ is not a distribution over the set $S$ of states since the automaton might execute an infinite sequence of $H$-enabled states never reaching an $H$-disabled state.  We let $\nu$ be a distribution over $S_\bot = S \cup \{\bot\}$ where $\bot \notin S$ represents nontermination and $\nu(\bot) = 1 - \sum_{s \in S} \nu(s)$.  Note that for no $a$, $\mu$, or $\nu$ does $\bot \trans^a \mu$ or $\bot \wtrans^a \nu$.

A \textsc{plts} $L$ combined with a state $s$ defines a probabilistic I/O automaton $\langle L, s\rangle$.  This state is thought of as the initial state of the automaton or the current state of the \textsc{plts}.  
We define a \emph{trace} to be a sequence of actions from $A^* \cup A^\omega$.  
Given such an automaton $\M$, we define $\model{\M}$ to be a function from input sequences to the random variable over traces that describes how the automaton $\M$ behaves under the inputs $\vec{i}$.  We let $\restrict{\model{\M}(\vec{i})}{E}$ denote the random variable over sequences of actions observable to the data examiner obtained by projecting only the actions in $E$ from the trace returned by random variable $\model{\M}(\vec{i})$.

To deal with nontermination, we note that the examiner can only observe finite prefixes of any nonterminating trace. 
When the examiner sees the finite prefixes of a trace, he must consider all traces of the system with the observed prefix as possible.  
(The set of these traces has been called a \emph{cone} --- see e.g.~\cite{lsv07observing}.)  Since the examiner may only see actions in $E$, these sets are in one-to-one correspondence with $E^*$.  
Thus, the examiner observing some event is not modeled as the probability of the system producing a trace in some set, but rather with the probability of a system producing a prefix of trace in some set.  That is, rather than using $\Pr[\restrict{\model{\M}(\vec{i})}{E} \in S]$ for $S \subseteq E^* \cup E^\omega$, we need $\Pr[\restrict{\model{\M}(\vec{i})}{E} \supseqeq S]$ for $S \subseteq E^*$ where $\supseqeq$ is the super-sequence-equal operator raised to work over sets in the following manner: $\vec{e} \supseqeq S$ iff there exists $\vec{e}' \in S$ such that $\vec{e} \supseqeq \vec{e}'$ where $\vec{e} \in E^* \cup E^\omega$ and $S \subseteq E^*$.

In Appendix~\ref{app:autom-details}, we formalize these concepts and show how to calculate these probabilities from the transitions of the automaton.

\subsection{Differential Noninterference}
\label{sec:mut-diff-priv}

Often the data set of a differentially private system is loaded over time and may change between queries.  Such changes in the data set are not explicitly modeled by the definition of differential privacy, but one could conceive of modeling such changes by having data points be time-indexed sequences of data.  Nevertheless, for formal verification, we require an explicit model of data set mutation.  Thus, we present a version of differential privacy defined in terms of the behavior of an automaton that accepts both queries and data points over time.

\begin{definition}[Differential Noninterference]
\label{def:strong-diff-priv}
An automaton $\M$ has $\epsilon$-differential noninterference if for all input sequences $\vec{i}_1$ and $\vec{i}_2$ in $I^*$ differing on at most one data point, and for all $S \subseteq E^*$,
\[ \Pr[\restrict{\model{\M}(\vec{i}_1)}{E} \supseqeq S] \leq \e{\epsilon} * \Pr[\restrict{\model{\M}(\vec{i}_2)}{E} \supseqeq S]\]
where we say two input sequences differ by one data point if one of the sequences may be constructed from the other by inserting a single data point anywhere in it. 
\end{definition}

By restricting the traces of $\M$ to only those elements of $E = Q \cup R$, we limit traces to only those actions accessible to the untrusted data examiner.  The definition requires that any subset of such traces be almost equally probable under the input sequences $\vec{i}_1$ and $\vec{i}_2$, which differ by at most one data point. Note that like the original form of differential privacy, we do not model the adversary explicitly but rather consider the behavior of the automaton over all possible input sequences the adversary could supply.

In Appendix~\ref{app:mut-diff-proofs}, we give full definitions for sequence differencing and prove results showing that our adaptation of differential privacy preserves pleasing properties of the original.  One such property is a composition result (Proposition~\ref{prp:strong-n-composition}):
the privacy leakage bound for a system whose inputs differ on at most $n$ data points is $n * \epsilon$ where $\epsilon$ is the leakage bound for the system if its inputs differ on one data point.

\subsection{Example: Automaton Model for Program of Figure~\ref{fig:code-window}}
\label{sec:ex-autom}

To eventually prove that the program of Figure~\ref{fig:code-window} has $(2t*\epsilon)$-differential noninterference, we first give an automaton model of the program, called $\exparamM(\K)$. Note that the model we give here is parametric in the set of sanitization functions; 
it applies not only to the program of Figure~\ref{fig:code-window}, which assumes $\K = \{\textsc{count},\textsc{sum}\}$ but to any other instance of the same program that uses a possibly different set of sanitization functions (modeled by the parameter $\K$). 
We define below the state space $S$ and transition relation $\trans$, which determine $\exparamL(\K) = \langle S, I, O, \trans\rangle$ for every set $\K$ of sanitization functions.  Using an initial state $s_0$, we get the automaton $\exparamM(\K) = \langle \exparamL(\K), s_0\rangle$.

\paragraph{States} Each state of the automaton can be viewed as a particular valuation of the variables in the program allowed by its type.  We model the array \verb|dPts| as a $t$-tuple of multisets of data points.  We model \verb|numPts| as a $t$-tuple of integers ranging from $0$ to $v$ where $v$ is the value held by the constant \verb|maxPts|.  We model the index \verb|curSlot| as an integer $c$ ranging from $0$ to $t-1$, which selects one of the multisets of the $t$-tuple.  The variable \verb|y| stores the most recent input. The variable \verb|res| keeps track of which output from $O$ is about to be produced and the sanitization function is stored in \verb|k|.
The state must also keep track of a program counter $pc$, which ranges over the program line numbers from $01$ to $20$.
Thus, the set of states $S$ is $\{01,\ldots,20\} \cross (\bag{D})^t \cross \{0,\ldots,v\}^t \cross \{0,\ldots,t-1\} \cross I \cross O \cross K$ where $\bag{D}$ is the set of all multisets with elements from $D$ and $K$ is the set of sanitization functions.

\paragraph{Actions} We model the \verb|input| command in the source code with the input action set $I$ of our automaton: for each possible value that \verb|input| can return there is an input action in $I$ corresponding to that value.  Inputs in the code can be either queries or data points, which is modeled by the partition of the set $I$ into the sets $Q$ for queries and $D$ for data points.  We model the \verb|print| command in the source code with the observable outputs $R$ (responses) of our automaton.  For each possible value that can be printed we have an output action in $R$.  We model all other commands by internal (hidden) actions.

\paragraph{Transitions} We list below only those transitions that are interesting for our purposes.  That is, transitions on actions from the sets $I$ and $R$, and transitions on hidden actions that represent internal computation such as choosing of an appropriate sanitization function for a given query and computation of the result using that function.  We use the symbol $\tau$ for hidden actions.
We also use \emph{Dirac} distributions: let $\mathsf{Dirac}(s)$ be the distribution such that $\Pr[\mathsf{Dirac}(s){=}s] = 1$ and $\Pr[\mathsf{Dirac}(s){=}s'] = 0$ for all $s' \neq s$.  Given a query $q$ in $Q$, we let $\kappa_q$ be the sanitization function that answers that query.  Some key transitions are:
\newcommand*{\Item}[1][\mbox{}]{\item[\mbox{}\ \ {#1}]}
\begin{description}
\Item[Input] $\;\;\;\langle 08, \vec{B}, \vec{n}, c, y, r, k\rangle \trans^i \mathsf{Dirac}(\langle 09, \vec{B}, \vec{n}, c, i, r, k\rangle)$

\Item[Choose Function] $\;\;\;\langle 14, \vec{B}, \vec{n}, c, y, r, k\rangle \trans^\tau \mathsf{Dirac}(\langle 15, \vec{B}, \vec{n}, c, y, r, \kappa_y\rangle)$

\Item[Compute Function] $\;\;\;\langle 15, \langle B_0, \ldots, B_{t-1}\rangle, \vec{n}\rangle, c, y, r, k\rangle \trans^\tau \mu$
 where 
\[ \mu(\langle 16, \langle B_0, \ldots, B_{t-1}\rangle, \vec{n}, c, y, r', k\rangle) = \Pr[k(\bigbagunion_{\ell=0}^{t-1} B_{\ell}) = r'] \]
 using $\bagunion$ for multiset union and $\mu(s') = 0$ for states not of that form, and 

\Item[Output Result] $\;\;\;\langle 16, \vec{B}, \vec{n}, c, y, r, k\rangle \trans^r \mathsf{Dirac}(\langle 17, \vec{B}, \vec{n}, c, y, r, k\rangle)$
\end{description}

The third transition above is a probabilistic transition that represents the internal computation of a sanitization function $k$ on the union of the multisets $B_0, \ldots, B_{t-1}$.  The effect of the transition is to update the value of the $pc$ from $15$ to $16$ and to update the result to be output from $r$ to a new value $r'$ such that the probability of ending up in state $\langle 16, \langle B_1, \ldots, B_t\rangle, c, n, y, r', k\rangle$ as a result of the transition is $\Pr[k(\bigbagunion_{\ell=1}^{t} B_{\ell}) = r']$.

From these transitions, we can calculate the extended transitions for each of the three types of $H$-disabled states:
\begin{description}
\Item[Drop] $\;\;\;\langle 08, \vec{B}, \vec{n}, c, y, r, k\rangle \wtrans^d \mathsf{Dirac}(\langle 08, \vec{B}, \vec{n}, c, d, r, k\rangle)$ when $n_c$ of $\vec{n}$ is $v$;

\Item[Add] $\;\;\;\langle 08, \vec{B}, \vec{n}, c, y, r, k\rangle \wtrans^d \mathsf{Dirac}(\langle 08, \vec{B}', \vec{n}', c, d, r, k\rangle)$ when $n_c$ of $\vec{n}$ is less than $v$ and $\vec{B}'$ and $\vec{n}'$ are such that $B'_{c} = B_{c} \bagunion \{d\}$, $n'_{c} = n_{c}+1$, and for all $c' \neq c$,  $B'_{c'} = B_{c'}$ and $n'_{c'} = n_{c'}$;
 
\Item[Answer Query] $\;\;\;\langle 08, \langle B_0, \ldots, B_{t-1}\rangle, \vec{n}, c, y, r, k\rangle \wtrans^q \nu$ 
where 
\[ \nu(\langle 16, \langle B_0, \ldots, B_{t-1}\rangle, \vec{n}, c, q, r', \kappa_q\rangle) = \Pr[k(\bigbagunion_{\ell=0}^{t-1} B_{\ell}) = r'] \]
 and $\nu(s') = 0$ for states not of that form; and 

\Item[Delete Old Data] $\;\;\;\langle 16, \vec{B}, \vec{n}, c, y, r, k\rangle \wtrans^r \mathsf{Dirac}(\langle 08, \vec{B}', \vec{n}, c, y, r, k\rangle)$\\
where we have $B'_{c+1 \mod t} = \emptybag$, $n'_{c+1 \mod t} = 0$, and for all $c'' \neq c+1 \mod t$, $B'_{c''} = B_{c''}$ and $n'_{c''} = n_{c''}$ using $\emptybag$ for the empty multiset.
\end{description}

The third extended transition above represents a sequence of transitions that starts with the input of a query $q$.  The input of the query is followed by transitions on hidden actions that model the computation of the answer to the query where some of these hidden steps are probabilistic.
The resulting state has the property that $\kappa_q$ has been chosen as the sanitization function and that $pc = 16$, which implies that the resulting state is $H$-disabled and the automaton is ready to perform an observable output by outputing the answer to the query. 

The state space $S$ and transition relation $\trans$ determines the \textsc{plts} $\exparamL(\K) = \langle S, I, O, \trans\rangle$ for every set $\K$ of differentially private functions.  Using the initial state $s_0 = \langle 1, \emptybag^t, 0^t, 1, y_0, r_0, k_0\rangle$, we get the automaton $\exparamM(\K) = \langle \exparamL(\K), s_0\rangle$.  (The initial values $y_0$, $r_0$, $k_0$ do not matter since they will be replaced before being used.)

\paragraph{Verification of Differential Privacy and Bounded Memory}
The remainder of this paper develops the proof techniques needed to formally verify that models such as the one shown above has differential noninterference.  In particular, in the next section, we describe a composition result that allows to separately consider whether the sanitization functions in $K$ have differential privacy and whether $\exparamM$ properly uses them.  In Section~\ref{sec:unwinding}, we present a proof technique using \emph{unwinding families} for showing that for all sets $K$ of sanitization functions with $\epsilon$-differential privacy, the automaton $\exparamM(K)$ has $(2t*\epsilon)$-differential noninterference.  Lastly, Section~\ref{sec:checking-algo}, provides a proof-checking algorithm that ensures our unwinding technique is properly used.  These methods together allow for a compositional and mechanically checked formal proof of differential noninterference.
 
Given that the system modeled above uses $\epsilon$-differentially private functions $t$ times, one might be surprised that we prove that it has $(2t*\epsilon)$-differential noninterference rather than $(t*\epsilon)$-differential noninterference.  This extra leakage comes from dealing with the bounded memory of actual computers.  In particular, each array in \verb|dPts| is limited to a length of \verb|maxPts|.  The program keeps track of the current number of data points stored in each slot with the array \verb|numPts|.  If the current slot has reached \verb|maxPts| data points, the program drops any incoming data points until \verb|curSlot| advances.  

This dropping of data points introduces extra privacy leakage.  A single data point can have two effects: it is both included in calculations and can cause the system to drop future data points and exclude from calculations.  Thus, the system has only $(2t*\epsilon)$-differential noninterference.  In many scenarios, the possibility of running out of memory for storing data points is unrealistic.  If the number of data points can never reach the memory bound, then under this assumption, one can show that system has $(t*\epsilon)$-differential noninterference.

It may be tempting to use a linked list for each slot and keep track of how many total data points are stored in all the slots combined.  Then, the program could drop data points only when all the memory is exhausted instead of just the current slot's allocation.  However, this change would allow a single data point stored in one slot to affect which data points are dropped from other slots in the future.  Thus, a single data point may have an unbounded effect on future computation preventing such a program from satisfying differential noninterference for any privacy bound.

\section{Decomposing Verification}
\label{sec:compo}

Recall the example system presented in Section~\ref{sec:mot-ex}. The source code in Figure~\ref{fig:code-window} is written parametrically in the set of sanitization functions (Lines \texttt{14} and \texttt{15}).
The model $\exparamM(\K)$ of the system given in Section~\ref{sec:ex-autom} is parametrized over the set of sanitization functions $\K$ where the computation of a sanitization function from $\K$ is idealized as a single transition in the transition system of  $\exparamM(\K)$.
We will call such models in which computation of functions are abstracted as a single step \emph{idealized models}.  In reality, any function in the set $\K$ would be implemented by a subroutine that can be modeled by an automaton and an \emph{implementation model} could be obtained from an idealized model by replacing each idealized transition for a sanitization function with its corresponding subroutine automaton.  

In this section, we first provide an algorithm for checking that such subroutine automata modeling sanitization functions have differential privacy.  Second, we show how to use the proof that a subroutine has differential privacy to simplify the task of proving
that an interactive system using that function has differential noninterference.  That is, we show how we support compositional reasoning by separating the verification of a sanitization function from the verification of a system that uses the function.

\subsection{Mechanized Verification of Differential Privacy}
\label{sec:compo-algo}

Previous work has provided a method of formally verifying that a sanitization function has differential privacy~\cite{rp10distance}.  Their method operates over a special language to enable type-checking.  Below we provide an alternative using automata to model the function.

In particular, we model a subroutine implementing a sanitization function $k$ operating on the database $B$ using an I/O automaton $M_{k,B}$.  As $k$ performs no I/O, the model $M_{k,B}$ has an empty set of inputs and only one output $h$, a hidden action.  The initial state of $M_{k,B}$ represents the start of the computation $k$ operating on the argument $B$.  For each output $r$ in the range of $k$, $M_{k,B}$ has a terminal state $\xi(r)$ with no outgoing transitions corresponding to returning the value $r$.  Since $k$ is a function, $s_0 \wtrans^h \nu$ must be a distribution over these terminal states with $\nu(\bot) = 0$ and $\nu(s) = 0$ for all states not corresponding to an output.

A function $k$ has $\epsilon$-differential privacy only if $M_{k,B_1}$ and $M_{k,B_2}$ induces sufficiently close distributions over related terminal states for all data bases $B_1$ and $B_2$ differing by at most one data point.  In particular, for all $r$ in the range of $k$, $\nu_1(\xi_1(r)) \leq \e{\epsilon} * \nu_2(\xi_2(r))$ where $\nu_i$ is the distribution over terminal states induced by the automaton $M_{k,B_i}$ and $\xi_i$ is the mapping from the range of $k$ to terminal states for $M_{k,B_i}$.  

Thus, mechanically checking if a function $k$ has differential privacy reduces to constructing the appropriate models $M_{k,B_i}$, computing the distributions $\nu_i$ for each of them, and comparing them as needed.  As we are only concerned with systems that can actually be implemented, only a finite number of models and comparisons are needed.  The construction of the models may be done using known techniques from model checking (see, e.g.,~\cite{cgp00model}).   The most complex step is computing the distributions $\nu_i$.  

Fortunately, each of these automaton models $M_{k,B_i}$ may be converted to an \emph{absorbing Markov chain}, a model of random behavior leading to one of a fixed set of \emph{absorbing states} each representing a different outcome.  Under this conversion, the probability of the Markov chain leading to a particular absorbing state corresponds to the distribution $\nu_i$ over terminal states of $M_{k,B_i}$.
This conversion starts with finding the set $S'$ of all $H$-disabled states reachable from $s_0$ by using hidden actions.  For this task, we may view the transition system as a directed graph $G$ where the nodes are states.  If $s \trans^h \mu$ and $\mu(s') > 0$ for some hidden action $h$, then we add an edge from $s$ to $s'$ labeled with $\mu(s')$ to $G$.  (Recall that $s$ will never transition under more than one such hidden action due to the transition-determinism axiom.)  A depth-first search may then find those states reachable from $s$ in $G$.
Second, we remove all states from $G$ that are not reachable from $s$.  
Third, we convert all states in $G$ that are reachable from $s$ that do not reach any $H$-disabled states to a single state $s_\bot$, which we treat as an $H$-disabled state.  We can do this with a reachability analysis for each state to every $H$-disabled state.
Forth, we add a self-loop labeled with probability $1$ from every $H$-disabled state (including $s_\bot$) to itself.  The resulting graph corresponds to an absorbing Markov chain where the $H$-disabled states (including $s_\bot$) are the absorbing states.

To compute the absorbing probabilities of the $H$-disabled states, we use the standard method as presented in~\cite{gs97introduction}.  First, we represent the chain using a transition matrix $\mathbf{P}$ in \emph{canonical form}.  That is, we renumber the states so that the non-absorbing, or \emph{transient}, states come first in $\mathbf{P}$.  In our case, these are the $H$-enabled states.
Let $t$ be the number of transient states and $r$ be the number of absorbing states.  We may view $\mathbf{P}$ as having the following form: 
\[ \mathbf{P} = \left[ \begin{array}{c|c} \mathbf{Q} & \mathbf{R}\\ \hline \mathbf{0} & \mathbf{I} \end{array} \right] \]
where 
$\mathbf{Q}$ is a $t$-by-$t$ matrix,
$\mathbf{R}$ is a non-zero $t$-by-$r$ matrix,
$\mathbf{I}$ is a $r$-by-$r$ identity matrix, and
$\mathbf{0}$ is a $r$-by-$t$ zero matrix.
Here, $\mathbf{Q}$, $\mathbf{R}$, and $\mathbf{I}$ capture, the probabilities for, respectively, moving from a transient state to a transient state, moving from a transient state to an absorbing state, and moving from an absorbing state to an absorbing state.
Second, from $\mathbf{P}$, we compute \emph{fundamental matrix} $\mathbf{N} = (\mathbf{I}-\mathbf{Q})^{-1}$.  Third, we compute $\mathbf{A} = \mathbf{NR}$.  The entry $a_{ij}$ of $\mathbf{A}$ is the probability of the chain ending in (being absorbed by) the state numbered $j$ when started in the state $i$.  Thus, we may set $\nu(s') = a_{ij}$ where $i$ is the number of the initial state and $j$ is the number of the state $s'$.
We refer the reader to \cite{gs97introduction} for the correctness of this algorithm for computing the absorbing probabilities.  

\paragraph{Algorithm $\mathtt{closure}(M,s,a)$} The above algorithm may be generalized to compute $\nu$ for a state $s$ and an action $a$ where $s \wtrans^a \nu$.  The generalization replaces initial state with $s$ and constructs the terminal absorbing states from the $H$-disabled states reachable from $s$.
Let $\mathtt{closure}(M,s,a)$ denote the generalized algorithm used this way to compute $\nu$ such that $s \wtrans^a \nu$.  

As for the runtime of $\mathtt{closure}$, note that the first step of constructing of the graph $G$ runs in $O(|S|)$ where $S$ is the state space of $M_{k,B_i}$.  Converting $G$ to use $s_\bot$ takes $O(|S|^2)$.
Every other step of the conversion process runs in $O(|S|)$.
The matrix operations used to compute the matrix $\mathbf{A}$ can all be done in $O(|S|^3)$ as $t \leq |S|$ and $r \leq |S|$.  Thus, it runs in $O(|S|^3)$ time. 
Using $\mathtt{closure}$ for computing each $\nu_i$, we may check if $k$ has differential privacy in $O(m+|\mathcal{B}|*|D|*|S|^3)$ where $m$ is the time required to compute all the models and $\mathcal{B}$ is the set of all databases $B_i$.

\subsection{Implementation and Composition}
\label{sec:refinement}

The ability to verify that a subroutine provides differential privacy aids the verification that a system using that subroutine has differential noninterference.  In particular, this section shows that the verification of differential noninterference may assume that the subroutine provides a differentially private distribution over return values in a single idealized transition, without modeling the internal transitions of the subroutine. Doing the verification based on such an idealized model is more manageable than doing it based on a model that includes the details about the implementation of the subroutine.  

\paragraph{Implementing a Transition with an Automaton}
We now define what it means in our model for a single step transition on a hidden action to be implemented by an automaton with a series of hidden transitions. We base our notion of implementation on hidden transitions since it is sufficiently general for our purposes --- we do not concern ourselves with the general question of preserving all kinds of observable behavior through implementation but rather the more restricted question of preserving the resulting distribution over computed values. 

A single internal transition of an automaton $M_1$ may result in a distribution over next states that corresponds to the distribution over terminal states induced by many internal transitions in another automaton $M_2$.  To formalize this, let $s^\dagger$ be a state of the automaton $M_1$ such that $s^\dagger \transone^{h^{\dagger}} \mu^\dagger$ for some hidden action $h^{\dagger}$ of $M_1$.  Let $\iota$ be an injection from $\supp(\mu^\dagger)$ to the state space of some other automaton $M_2$ such that every state in the image of $\iota$ is disabled for every action (i.e., they are terminal states).  We say that the automaton \emph{$M_2$ implements the transition of $s^\dagger$ under $\iota$} if for all $s \in \supp(\mu^\dagger)$, $\mu^\dagger(s) = \sum_{\vec{h} \in H_2^+} M_2([\,])(\vec{h},\iota(s))$ where $H_2$ is the hidden action set of $M_2$ and $H_2^+$ is the set of non-empty finite sequences using elements from $H_2$.  That is, $M_2$ implements the transition of $s^\dagger$ under $\iota$ if the distribution over the terminal states that $M_2$ reaches is isomorphic to $\mu^\dagger$ under $\iota$.

\paragraph{Subroutine Composition}
Subroutine composition can be viewed as replacing a single step transition in an idealized model with its automaton implementation where such repeated replacements can be used to derive an implementation model from the idealized model. 

Let $M_1[s^\dagger, M_2, \iota]$ denote the automaton that results from replacing an internal transition from the state $s^\dagger$ of $M_1$ with the subroutine $M_2$ with the injection $\iota$ providing how to return from the subroutine.  Formally, given $M_1 = \langle\langle S_1, Q_1\uplus D_1, R_1 \uplus H_1, \transone\rangle, s^0_1\rangle$, $M_2 =  \langle\langle S_2, \emptyset, H_2, \transtwo\rangle, s^0_2\rangle$, $s^\dagger \in S_1$ such that $s^\dagger \transone^{h^{\dagger}} \mu^\dagger$ for some hidden action $h^\dagger \in H_1$ and $\mu^\dagger$ where $s^{\dagger}$ is the unique state that enables $h^{\dagger}$, $s^{\dagger} \notin \supp(\mu^{\dagger})$, and $\iota : \supp(\mu^\dagger) \to S_2$ such that every state in its image is disabled for all actions, let $M_1[s^\dagger, M_2, \iota]$ denote the automaton $M_3 = \langle\langle S_1 \uplus S_2, I_1, R_1 \uplus H_1 \uplus H_2 \uplus \{h^{\ddagger}\}, \transthree\rangle, s^0_1\rangle$ where $\uplus$ is disjoint union and $\transthree$ is defined as follows:
\begin{itemize}
\item $s_1 \transthree^a \mu$ if $s_1 \in S_1$, $s_1 \neq s^\dagger$, and $s_1 \transone^a \mu$;
\item $s_2 \transthree^a \mu$ if $s_2 \in S_2$ and $s_2 \transtwo^a \mu$;
\item $s^{\dagger} \transthree^{h^{\ddagger}} \mathsf{Dirac}(s^0_2)$; and
\item $\iota(s_1) \transthree^{h^{\ddagger}} \mathsf{Dirac}(s_1)$ for all $s_1 \in \supp(\mu^\dagger)$.
\end{itemize}

The special hidden action $h^{\ddagger}$ in the definition of $M_3$ above is used to mark the entry and exits points of the subroutine represented by $M_2$. This extra action is used to correctly ``hook up'' $M_2$ with $M_1$ to obtain $M_3$.

The lemma below states that if some internal transition of an automaton $M_1$ (for example, a step corresponding to calling a sanitization function in a differentially noninterference system) is replaced by an automaton $M_2$ (for example, multiple steps corresponding to a subroutine implementing the sanitization function), then the observable behavior of the resulting automaton is identical to that of $M_1$.

\begin{theorem}[Subroutine Composition]\label{thm:hidden-transition-replacement}
For all automata $M_1$ and $M_2$, states $s^\dagger$, and injections $\iota$ such that $M_2$ implements the transition of $s^\dagger$ under $\iota$, for all $\vec{i}$ in $I^*$, and $\vec{e}$ in $E^*$,
\[ \Pr[\, \restrict{M_1(\vec{i})}{E} \supseqeq \vec{e}\,] = \Pr[\, \restrict{M_1[s^\dagger,M_2,\iota](\vec{i})}{E} \supseqeq \vec{e}\,] \]
\end{theorem}
In Appendix~\ref{app:hidden-transition-replacement}, we prove this by way of two lemmas.

A corollary is that if an idealized model has differential noninterference then a implementation model formed by replacing its internal transitions with subroutine automata also has differential noninterference.

\subsection{Example: Decomposing Verification}
\label{sec:compo-ex}

Suppose that $\excsM$ is the automaton obtained from $\exparamM(\{\textsc{count},\textsc{sum}\})$ by replacing the transitions that represent the computations of the functions \textsc{count} and \textsc{sum} with subroutine automata $\M_{\textsc{count},B_i}$ and $\M_{\textsc{sum},B_i}$.  That is, $\excsM$ is the code shown in Figure~\ref{fig:code-window} with the implementations of \textsc{count} and \textsc{sum} in-lined.  We may apply the composition theorem repeatedly for each replacement of a single transition in $\exparamM(\{\textsc{count},\textsc{sum}\})$ with a subroutine automaton in $\excsM$.  Such repeated compositions reduces the problem of verifying differential noninterference for $\excsM$ to two smaller problems:  First, we must show that the automata $\M_{\textsc{count},B_i}$ and $\M_{\textsc{sum},B_i}$ implement with a series of internal transitions the transitions corresponding to the functions \textsc{count} and \textsc{sum} found in $\exparamM(\{\textsc{count},\textsc{sum}\})$ as described in our formal definition of \emph{implementation}.  Second, we must show that the idealized model $\exparamM(\{\textsc{count},\textsc{sum}\})$ has the differential noninterference.  

The first problem can be solved using $\mathtt{closure}$, which establishes that the automaton correctly implement $\textsc{count}$ and $\textsc{sum}$.  As $\textsc{count}$ and $\textsc{sum}$ has differential privacy (proofs provided in Appendix~\ref{app:bdm}), we may conclude that these subroutine automata have differentially private distributions over their terminal states.\footnote{We may also mechanically prove that these subroutine automata have differential privacy using other formal methods such a type system~\cite{rp10distance}.}
The next two sections deal with solving the second problem.

While \textsc{count} and \textsc{sum} are simple sanitization functions, the above approach generalizes to more complex sanitization functions: As long as the function can be modeled as a series of internal transitions that ends in states corresponding to its return values, our approach will apply.  While most of the algorithms previously published use unbounded state spaces, we believe our approach can handle bounded versions of them.


\section{Unwinding Proof Technique}
\label{sec:unwinding}

We desire a technique for drawing conclusions about the global behavior (executions) of the system from local aspects (states, actions, and transitions) of the model.  Faced with a similar situation, Goguen and Meseguer introduced unwinding relations to simplify proving that a system has noninterference~\cite{gm84unwinding}.  We present a similar technique for proving that a system has differential noninterference.  In particular we state what it means for a relation family to be an unwinding family and prove Theorem~\ref{thm:unwinding}, which roughly states that the existence of an unwinding family for a given automaton implies that it satisfies differential noninterference.  Our unwinding notion is probabilistic and approximate, which is in keeping with the notion of differential privacy.  The novelty lies in the way we keep track of the privacy leakage bound, which evolves as the system evolves where the evolution is constrained by the differential privacy definition.


\subsection{Definition and Soundness}

Formulating a notion of unwinding relation that is sound for showing
differential noninterference is more complicated than existing notions
for showing noninterference because we must deal with probabilities
and we must keep track of the privacy leakage bound $\epsilon$.  To
deal with probabilities and approximation, we adapt the notion of
\emph{approximate lifting} from previous work on approximate
probabilistic simulation relations in the context of cryptographic
protocols~\cite{st07}.  However, such work does not deal with tracking
a leakage bound (see Section~\ref{sec:related-work} for additional
details).  Thus, we introduce a \emph{family} of unwinding relations
indexed by various amounts of privacy leakage.  Each unwinding
relation in the family is a relation on the state space of the
automaton. The unwinding relation indexed by the leakage amount
$\epsilon$ relates states that exhibit approximately the same trace
distributions in the sense of $\epsilon$-differential noninterference.

To deal with probabilities in a concise and modular way, we first define an approximate lifting operation that takes a relation over sets and produces a relation over distributions on those sets.  The degree of approximation is governed by a parameter $\delta$.  
\begin{definition}[$\delta$-Approximate Lifting]
\label{def:approx-lifting} 
Let $\mathsf{R}$ be a relation between a set $X$ and a set $Y$.  The \emph{$\delta$-approximate lifting} of $\mathsf{R}$ denoted by $\lift{\mathsf{R},\delta}$ is the relation between $\Disc(X)$ and $\Disc(Y)$ such that for all $\nu_1$ in $\Disc(X)$ and $\nu_2$ in $\Disc(Y)$, $\nu_1 \mathrel{\lift{\mathsf{R},\delta}} \nu_2$ 
if and only if there exists a bijection $\beta: \supp(\nu_1) \to \supp(\nu_2)$ such that for all $x$ in $\supp(\nu_1)$, $x \mathrel{\mathsf{R}} \beta(x)$ and $|\ln \nu_1(x) - \ln \nu_2(\beta(x))| \leq \delta$.
\end{definition}
The requirement for $\beta$ to be from the support set of $\nu_1$ to the support set of $\nu_2$ ensures that if a state is assigned a non-zero probability in $\nu_1$ then it is not possible for a related state to be assigned a zero probability in $\nu_2$ and vice versa---there is one to one correspondence between the states with non-zero and identical probabilities in the two distributions.
The form of $\delta$ involves natural logarithms because the privacy leakage bound in the differential privacy definition appears in the exponent.

Next we define our unwinding technique, which is illustrated in Figure~\ref{fig:unwinding-cover}.  Intuitively, since we want the behavior of the automaton to change only by a factor of $\epsilon$ on receiving a single data point, we want the transitions under a data point from a state $s$ to lead to states $s'$ that are only a factor of $\epsilon$ different from $s$.  \emph{Covering} (Definition~\ref{def:cover}) formalizes this by requiring that state $s$ is related to each such state $s'$ by a relation $\mathcal{R}^{\epsilon}$ that is part of an \emph{$\epsilon$-unwinding family} (Definition~\ref{def:unwinding}).  

In more detail, an $\epsilon$-unwinding family starts with a privacy leakage budget of $\epsilon$, which decreases over time to a current balance of $\epsilon'$.  Related states $s_1$ and $s_2$ are required to only make transitions under the same actions.  The distributions $\nu_1$ and $\nu_2$ that result from these transitions followed by any number of transitions under hidden outputs may differ only by a factor of $\delta$.  This difference is subtracted from the current balance $\epsilon'$ to get a new current balance.  Once the balance reaches zero, the resulting distributions must be equivalent.  As the balance started at $\epsilon$, only a total of $\epsilon$ privacy can be leaked, a point proved in Lemma~\ref{thm:near-unwinding}.   
\begin{figure*}
\begin{center}
\mbox{}\!\!\!\!\!\!\!\!\!\!\!\!\!\!\!\!\input{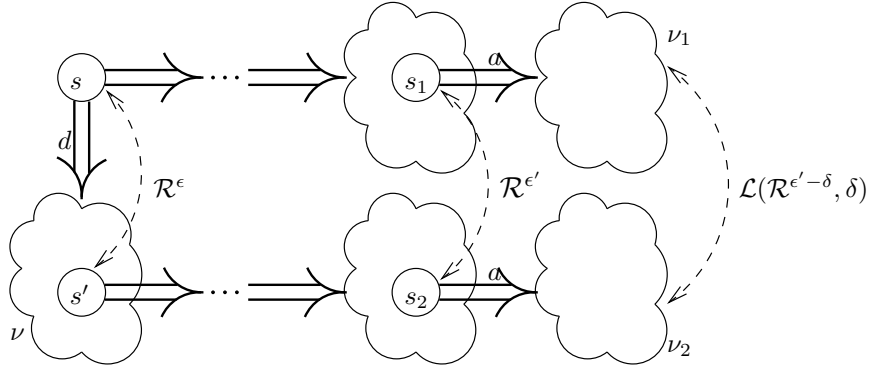}
\caption{Unwinding Family and Covering: The left side shows the requirements for a covering.  The right side shows the requirements placed on an unwinding family.  The solid arrows denote the extended transition relation $\wtrans$ and clouds depict probability distributions such as $\nu$ where $s' \in \supp(\nu)$.}
\label{fig:unwinding-cover}
\end{center}
\end{figure*}

\begin{definition}[$\epsilon$-Unwinding Family]
\label{def:unwinding}
For a non-negative real number $\epsilon$, a family indexed by the set $[0,\epsilon]$ of relations $\mathcal{R}^{\cdot}$ over the $H$-disabled states of a \textsc{plts} $L$ is an \emph{$\epsilon$-unwinding family} for $L$ if for all $\epsilon'$ in $[0,\epsilon]$, for all $x_1$ and $x_2$ in $S_\bot$ such that $x_1 \mathrel{\mathcal{R}^{\epsilon'}} x_2$, for all $a$ in $I \cup R$, there exists $\nu_1$ such that $x_1 \wtrans^a \nu_1$ iff there exists $\nu_2$ such that $x_2 \wtrans^a \nu_2$, and when they do exist, there exists a real number $\delta$ in $[0,\epsilon']$ such that $\nu_1 \mathrel{\lift{\mathcal{R}^{\epsilon' - \delta},\delta}} \nu_2$.
\end{definition}

\begin{lemma}\label{thm:near-unwinding}
For all $\epsilon$-unwinding families $\mathcal{R}^\cdot$, all $\epsilon'$ in $[0,\epsilon]$, all $x_1$ and $x_2$ in $S_\bot$ such that $x_1 \mathrel{\mathcal{R}^{\epsilon'}} x_2$, all $\vec{i}$ in $I^*$, and all $\vec{e}$ in $E^*$, both \\ $\Pr[\, \restrict{\model{\langle L, x_1\rangle}(\vec{i})}{E} \supseqeq \vec{e}\,] \leq \e{\epsilon'} \Pr[\, \restrict{\model{\langle L, x_2\rangle}(\vec{i})}{E} \supseqeq \vec{e}\,]$ and \\ $\Pr[\, \restrict{\model{\langle L, x_2\rangle}(\vec{i})}{E} \supseqeq \vec{e}\,] \leq \e{\epsilon'} \Pr[\, \restrict{\model{\langle L, x_1\rangle}(\vec{i})}{E} \supseqeq \vec{e} \,]$.
\end{lemma}

The above lemma shows that two states related by an $\epsilon$-unwinding family, given the same input sequence, produce distributions that only deviate by a factor $\epsilon$.  
Thus, to maintain $\epsilon$-differential noninterference, we desire that a state $s$ should upon receiving a single data point $d$ transition to a state $s'$ that can be put into an $\epsilon$-unwinding family with $s$.  We formalize this intuition with the next definition and confirm it with the following theorem.

\begin{definition}[Covers]\label{def:cover}
We say that an $\epsilon$-unwinding family $\mathcal{R}^\cdot$ for a \textsc{plts} $L$ \emph{covers} a state $s$ and data point $d$ of $L$ if $s \wtrans^d \nu$ implies that $\nu(\bot) = 0$ and for all $s' \in \supp(\nu)$, $s \mathrel{\mathcal{R}^{\epsilon}} s'$.
\end{definition}

\begin{theorem}\label{thm:unwinding}
For an automaton $\M = \langle L, s_0\rangle$, if for all $H$-disabled states $s$ reachable from $s_0$ and all data points $d$, there exists a $\epsilon$-unwinding family that covers $s$ and $d$, then $\model{\M}$ has $\epsilon$-differential noninterference. 
\end{theorem}
Appendix~\ref{app:unwinding-soundness} holds the proofs of 
Lemma~\ref{thm:near-unwinding} and Theorem~\ref{thm:unwinding}.
We prove Lemma~\ref{thm:near-unwinding} by induction over the structure of $\vec{a}$.  The interesting cases arise when $\vec{a}$ is of the form  $i\cons\vec{a}'$ for $i \in I$ or $o\cons\vec{a}'$ for $o \in O$, which require similar reasoning.  Suppose that $\vec{a} = i\cons\vec{a}'$ and $s_1 \wtrans^i \nu_1$ for some $i \in I$.  By the unwinding relation, we know that there exists a transition $s_2 \wtrans^i \nu_2$ such that $\nu_1$ and $\nu_2$ are in keeping with the privacy leakage bound imposed by the unwinding relation.  Then for states $s'_1 \in \supp(\nu_1)$, and $s'_2 \in \supp(\nu_2)$, we apply the inductive hypothesis for $\vec{a}'$ to obtain the result.  

To prove Theorem~\ref{thm:unwinding}, we use Proposition~\ref{prp:strong-by-trace} and show for all $\vec{i}_1$, $\vec{i}_2$, and $\vec{e}$ where $\Delta(\vec{i}_1, \vec{i}_2) = 1$ that $\Pr[\restrict{\model{\langle L, s\rangle}(\vec{i}_1)}{E} {\supseqeq} \vec{e}] \leq  \e{\epsilon} \Pr[\restrict{\model{\langle L, s\rangle}(\vec{i}_2)}{E} {\supseqeq} \vec{e}]$.  We use proof by induction over $\vec{i}_1$, $\vec{i}_2$, and $\vec{e}$.  When we reach the point where $\vec{i}_1$ and $\vec{i}_2$ differ by a data point $d$, we apply Lemma~\ref{thm:near-unwinding} knowing that an $\epsilon$-unwinding family exists for the current state $s$ and $d$.

\subsection{Example: Applying the Proof Technique}
\label{sec:unwinding-ex}

We now return to the parametric automaton model $\exparamM(\K)$ of Section~\ref{sec:ex-autom}.  We show that for any $\K$, every state $s$ and data point $d$ of $\exparamL(\K)$ 
can be \emph{covered} by a $(2t*\epsilon)$-unwinding family $\mathcal{R}^{\cdot}_{s,d}$ in the sense of Definition~\ref{def:cover}.  Differential noninterference will follow from Theorem~\ref{thm:unwinding}.  

For the $(2t*\epsilon)$-unwinding family $\mathcal{R}^{\cdot}_{s,d}$, we construct for each $j$ in $[0,t]$ the unwinding relation $\mathcal{R}^{2j*\epsilon}_{s,d}$.  To construct these unwinding relations, we first introduce some notation. 

For a state $s = \langle pc, \vec{B}, \vec{n}, c, y, r, k\rangle$ and $d \in D$, $\mathsf{add}(s,c',d)$ adds $d$ to the slot $c'$ of the state $s$.  Formally, 
\[ \mathsf{add}(s,c',d) =  \langle pc, \vec{B}', \vec{n}, c, y, r, k\rangle \]
where $\vec{B}' = \vec{B}$ and $\vec{n}' = \vec{n}$ when $n_c = v$ and, otherwise, $B'_{c} = B_{c} \bagunion \{d\}$, $n'_{c} = n_{c}+1$, and for all $c' \neq c$,  $B'_{c'} = B_{c'}$ and $n'_{c'} = n_{c'}$.

The function $\mathsf{swap}$ replaces one data point with another. Formally,

\[ \mathsf{swap}(s,c',d,d') =  \langle pc, \vec{B}', \vec{n}, c, y, r, k\rangle \] where $B'_{c'} = B_{c'} - \{d'\} \bagunion \{d\}$ and $B'_{c''} = B_{c''}$ for all $c'' \neq c'$.   

For $j$ such that $0 \leq j \leq t$, let $S_1^{j}$ to be the set of all states $s_1$ such that $s_1$ is reachable from $s$ using $t-j$ queries and any number of data points. Intuitively, this means that from $s_1$ one can pose $j$ more queries until the privacy budget runs out on the data point that is input into the system in state $s$. We define the relations as follows:

\begin{itemize}
\item For $j > 0$, let $\mathcal{R}^{2j*\epsilon}_{s,d}$ to be such that for all $s_1 \in S_1^{j}$, $s_1 \mathrel{\mathcal{R}^{2j*\epsilon}_{s,d}} \mathsf{add}(s_1,c,d)$ and for all $d'$, $s_1 \mathrel{\mathcal{R}^{2j*\epsilon}_{s,d}} \mathsf{swap}(s_1,c,d,d')$ where $s = \langle pc, \vec{B}, \vec{n}, c, y, r, k\rangle$.  That is, $\mathcal{R}^{2j*\epsilon}_{s,d}$ relates a state to the states it could have become had it received $d$ as input when the \verb|curSlot| was $c$, the value \verb|curSlot| had in state $s$.  

\item For $j = 0$, $\mathcal{R}^{2j*\epsilon}_{s,d}$ is as above for states with a PC of $16$ and is equality for those with a PC of $08$.
\end{itemize}

\begin{lemma}\label{lem:near-system-unwinding}
For all sets $\K$ of functions such that each function in $\K$ has $\epsilon$-differential privacy, for all states $s$ and for all data points $d$, $\mathcal{R}^{\cdot}_{s,d}$ is a $(2t*\epsilon)$-unwinding family for the automaton $\exparamM(\K)$.
\end{lemma}

Appendix~\ref{app:example-unwinding-proofs} holds the proof. 
The proof uses a case analysis over the different types of actions $a$ that might be received by two related states.  The most interesting case is when $a$ is a query and $j=1$.  In this case, $s_1 \mathrel{\mathcal{R}^{\epsilon'}_{s,d}} s_2$ implies that $s_1$ is in $S_1^{t-1}$ with $s_1$ and $s_2$ reached in $t-1$ queries.  For a $2t*\epsilon$ privacy leakage bound, this corresponds to the last time $d$ may be used in answering a query.  This requirement is met since for $s_1$ and $s_2$ to be reached with $t-1$ queries, by the construction of $\exparamM(\K)$, \verb|curSlot| in both states must be $t-1$ slots away from the slot that holds $d$.  Thus, after answering the next query the slot \verb|curSlot|, whose value is always mod $t$, will point to the slot that holds $d$ and that slot will be rewritten removing $d$.  

Since $\mathcal{R}^{2j*\epsilon}_{s,d}$ covers $s$ and $d$ for all states $s$ and data points $d$ of the automaton $\exparamM(\K)$, Lemma~\ref{lem:near-system-unwinding} and Theorem~\ref{thm:unwinding} implies that the automaton has $(2t*\epsilon)$-differential noninterference.

\begin{theorem}\label{thm:system-unwinding}
For all set of functions $\K$ such that each function in $\K$ has $\epsilon$-differential privacy, $\exparamM(\K)$ has $(2t*\epsilon)$-differential noninterference.
\end{theorem}

As $\textsc{count}$ and $\textsc{sum}$ are $\epsilon$-differentially private functions, this implies that $\exparamM(\{\textsc{count},\textsc{sum}\})$ has $(2t*\epsilon)$-differential noninterference.  Furthermore, as explained in Section~\ref{sec:compo-ex}, subroutine composition shows that $\excsM$, a system with $\textsc{count}$ and $\textsc{sum}$ implemented as subroutines instead of atomic transitions, has $(2t*\epsilon)$-differential noninterference.  Thus, we have proved that our example has $(2t*\epsilon)$-differential noninterference.  In the next section we turn to mechanically verifying differential noninterference.

\section{Mechanizing Verification of Unwinding}
\label{sec:checking-algo}

We provide an algorithm that soundly checks if a given family of relations is an unwinding family for a given automaton. While our algorithm does not generate the unwinding family, it automates the process of showing that a candidate family satisfies all the conditions for being an unwinding family (Definition~\ref{def:unwinding}). 
By repeatedly applying our algorithm to a collection of relation families, we can algorithmically check that the covering condition of Theorem~\ref{thm:unwinding} holds and that automaton has differential noninterference. The process of verifying an unwinding relation family manually is typically tedious and sometimes error-prone.  The existence of a mechanized verifier hence adds practical value to the proof technique presented in the previous section and justifies its use in favor of ad hoc proof methods.

\subsection{Algorithm}
Our algorithm $\mathtt{isUnwindFam}$ takes as input a labeled transition system of finite size, an array of relations over the system's states, a value $\delta$, and a natural number $t$.  
The array $\rel$ may only represent relation families $\mathcal{R}^{\cdot}$ over the interval $[0,t*\delta]$ of a restricted form.  
$\mathcal{R}^\cdot$ must be such that $\mathcal{R}^{j\delta} = \mathcal{R}^{k\delta}$ for all $j$ and $k$ such that $\lfloor j \rfloor = \lfloor k\rfloor$.  That is, it must be possible to break the index set of $\mathcal{R}^{\cdot}$ into $t$ intervals of size $\delta$ such that the relations in that interval are the same and one point corresponding to $\mathcal{R}^{t*\delta}$.

\newcommand{\s}{\ \ \ \ } 
\begin{figure}
\mbox{}$\mathtt{isUnwindFam}(\langle S, I, O, T\rangle, \rel, \delta, t)$\\
\mbox{}\s convert all hidden actions of $\langle\langle S, D, Q, R, T\rangle, s_0\rangle$ to be the same one\\
\mbox{}\s if($|\rel| \neq t+1$),\\
\mbox{}\s\s return $\mathtt{false}$\\
\mbox{}\s for all $i$ in $[0,t]$,\\
\mbox{}\s\s for all $\langle x_1, x_2\rangle \in \rel[i]$,\\
\mbox{}\s\s\s for all $a \in I \cup O$,\\
\mbox{}\s\s\s\s if ($T[x_1][a] = \nil$ xor $T[x_2][a] = \nil$),\\
\mbox{}\s\s\s\s\s then return $\mathtt{false}$\\
\mbox{}\s\s\s\s if ($T[x_1][a] \neq \nil$ and $T[x_2][a] \neq \nil$),\\
\mbox{}\s\s\s\s\s $\nu_1 = \mathtt{closure}(\langle S, I, O, T\rangle, x_1, a)$\\
\mbox{}\s\s\s\s\s $\nu_2 = \mathtt{closure}(\langle S, I, O, T\rangle, x_2, a)$\\
\mbox{}\s\s\s\s\s if(not $\mathtt{isInLiftedRelation}(S_\bot, \rel[i], 0, \nu_1, \nu_2)$)\\
\mbox{}\s\s\s\s\s\s if($i = 0$),\\
\mbox{}\s\s\s\s\s\s\s return $\mathtt{false}$\\
\mbox{}\s\s\s\s\s\s if(not $\mathtt{isInLiftedRelation}(S_\bot, \rel[i-1], \delta, \nu_1, \nu_2)$)\\
\mbox{}\s\s\s\s\s\s\s return $\mathtt{false}$\\
\mbox{}\s return $\mathtt{true}$
\caption{Algorithm for checking relation families.}\label{fig:checking-algo-uw}
\end{figure}
The algorithm is shown in Figure~\ref{fig:checking-algo-uw}.  It represents the transition relation $\trans$ as an array $T$ with $|S_\bot|$ rows and $|A|$ columns where $T[\bot][a] = \nil$ for all $a$.  The array either stores a distribution over next states or $\nil$ to indicate that the state cannot transition under that action.

The first step of the algorithm converts all the hidden actions to be the same one since $\mathtt{closure}$ presumes just one hidden action.  The function $\mathtt{closure}$, defined in Section~\ref{sec:compo-algo}, computes the distribution over states that results from the system exhibiting the observable behavior $a$ from a state $x_i$ and computing until reaching an $H$-disabled state.

The distributions resulting from $\mathtt{closure}$ are compared with the provided family $\rel$ using the function $\mathtt{isInLiftedRelation}$ to determine whether they obey the requirements of a unwinding family.
$\mathtt{isInLiftedRelation}(\mathsf{R}, \delta, \nu_1, \nu_2)$ checks if the two distributions $\nu_1$ and $\nu_2$ are related by the $\delta$-approximate lifting of $\mathsf{R}$.  This function operates in $O(|S|^{2.5})$ time by reducing the problem to the decision problem of if a perfect matching exists for a bipartite graph, which can be solved in $O(|S|^{2.5})$ using the Hopcroft-Karp algorithm~\cite{hk73n52}.  The reduction constructs a bipartite graph such that each vertex in the left part of the graph corresponds to a state in the support of $\nu_1$, and each in the right part to a state in the support of $\nu_2$.  Edges connect those states $x_1$ in the left part to those $x_2$ in the right part such that $x_1 \mathrel{\mathsf{R}} x_2$ and $|\ln \nu_1(x_1) - \ln \nu_2(x_2)| \leq \delta$.  A matching of graph that includes every vertex (i.e., a \emph{perfect} matching) exists iff there is a bijection showing that $\nu_1 \mathrel{\lift{\mathsf{R},\delta}} \nu_2$.  Appendix~\ref{app:lift-algo} formally presents the algorithm and proves this result.

The following lemmas state, respectively, the soundness and the runtime complexity of the algorithm.  Appendix~\ref{app:checking-algo-proofs} contains the proofs for this section.

\begin{lemma}[Soundness]\label{lem:checking-algo-uw-sound}
If the algorithm $\mathtt{isUnwindFam}(L,\rel,\delta,t)$ returns true, then $\rel$ corresponds to relation family that is $(t*\delta)$-unwinding family for $L$.
\end{lemma}

\begin{lemma}[Runtime Complexity]\label{lem:checking-algo-uw-time}
The algorithm $\mathtt{isUnwindFam}$ runs in $O(t*|A|*|S|^4)$ time.
\end{lemma}

We use $\mathtt{isUnwindFam}$ to construct an algorithm $\mathtt{isAllCovered}$ that checks a collection of relation families to conclude if they prove that an automaton has differential privacy (using Theorem ~\ref{thm:unwinding}). In particular, the algorithm takes as input an automaton, an array $\rels$ of relation families, $\delta$, and the natural number $t$.  For all states $s$ that are reachable from the start state of the automaton and data points $d$, the algorithm uses $\mathtt{isUnwindFam}$ to check whether $\rels[s][d]$ corresponds to a $(t*\delta)$-unwinding family that covers $s$ and $d$. 
The algorithm is shown in Figure~\ref{fig:checking-algo-cover}.
\begin{figure}
\mbox{}$\mathtt{isAllCovered}(\langle\langle S, I, O, T\rangle, s_0\rangle, \rels, \delta, t)$\\
\mbox{}\s $\mathtt{reachableStates := computeReachableStates(}\langle S, D, Q, R, T\rangle\mathtt{,} s_0\mathtt{)}$\\
\mbox{}\s for all $s$ in $\mathtt{reachableStates}$,\\
\mbox{}\s\s for all $d \in D$,\\
\mbox{}\s\s\s if($T[s][d] \neq \nil$),\\
\mbox{}\s\s\s\s $\nu = \mathtt{closure}(\langle S, I, O, T\rangle, s, d)$\\
\mbox{}\s\s\s\s if($\nu(\bot) \neq 0$ or $|\rels[s][d]| \neq t+1$),\\
\mbox{}\s\s\s\s\s return $\mathtt{false}$\\
\mbox{}\s\s\s\s for all $s' \in S$,\\
\mbox{}\s\s\s\s\s if($\nu(s') > 0$ and $\langle s, s'\rangle \notin \rels[s][d]$),\\
\mbox{}\s\s\s\s\s\s return $\mathtt{false}$\\
\mbox{}\s\s\s\s if(not $\mathtt{isUnwindFam}(\langle S, I, O, T\rangle, \rels[s][d], \delta)$\\
\mbox{}\s\s\s\s\s return $\mathtt{false}$\\
\mbox{}\s return $\mathtt{true}$
\caption{Algorithm for checking for differential privacy.}\label{fig:checking-algo-cover}
\end{figure}

The following theorems state the soundness and the runtime complexity of the procedure for checking whether all reachable states are covered by a given collection of relation families. 

\begin{theorem}[Soundness]\label{thm:checking-algo-cover-sound}
If $\mathtt{isAllCovered}(M, \rels, \delta, t)$ returns true, then $M$ has $(t*\delta)$-differential noninterference.
\end{theorem}

\begin{theorem}[Runtime Complexity]\label{thm:checking-algo-cover-time}
The algorithm $\mathtt{isAllCovered}$ runs in $O(t*|D|*|A|*|S|^5)$ time.
\end{theorem}

While sound, the algorithm is not complete even for this restricted class of unwinding relations it accepts as input.  The algorithm (soundly) rejects any family if it has a relation that relates two states that transition to distributions over next states that differ by more than $\delta$.  That is, it requires that the automaton never leaks more than a $\delta$ worth of private information in a single step.  Furthermore, it pessimistically presumes that every leakage of private information is a whole $\delta$s worth.  

Nevertheless, we believe the algorithm is still of interest.  In the next section, we show that it is powerful enough to prove that our example system, which is similar to \textsc{pinq}, has differential noninterference.  
While this system only has two very simple sanitization functions, $\mathsc{count}$ and $\mathsc{sum}$, our algorithm will work for more complex sanitization functions provided they can be computed with a finite number of states.

\subsection{Example: Using the Algorithm}

To use our algorithm, we must first model the above program as an automaton $\excsM$ with the subroutines \textsc{count} and \textsc{sum} in-lined as explained in Section~\ref{sec:compo-ex}.  Then, we must construct $\rels$, which stores all the needed $(2t*\epsilon)$-unwinding families in the correct format.  Such families exist since whenever $\excsM$ leaks privacy, it leaks no more than $2*\epsilon$ in a single step, and, thus, we can use $2*\epsilon$ for $\delta$.  These families are instances of the parametric families shown in Section~\ref{sec:unwinding-ex}.
The reader can confirm that these families may be expressed in the needed format for $\rels$.  

Indeed, as the body of the sanitization functions consists entirely of $H$-enabled states, only the distributions over return values matter to our algorithm in that they influence the computation of $\mathtt{closure}$ and nothing more.  Thus, the general families further shows that our algorithm can verify any modification of $\excsM$ that substitutes a different set of $\epsilon$-differentially private functions for $\{\mathsc{count}, \mathsc{sum}\}$ provided that those functions can be implemented using a bounded number of states as we would expect from the discussion of composition in Section~\ref{sec:compo-ex}. 

\section{Related Work}
\label{sec:related-work}

\paragraph{Formal Verification of Differential Privacy}
The most closely related work to ours is a programming language with a linear type system for proving that well-typed programs in the language have differential privacy~\cite{rp10distance}.  Later work applies their type system to detecting network attacks in a private manner~\cite{rawhps10differential}.  
The usual trade-offs between a program analysis technique designed to work over standard programming languages and a custom type system for a specialized language apply: the type system makes explicit in the source code why the program has differential privacy and type checking scales well, but the programmer must use a special-purpose programming language and annotate the code as the type system requires.  Additionally, their programming language lacks I/O commands for creating interactive systems whereas our proof technique is for automata modeling interactive systems.

\paragraph{Other Differential Privacy Definitions} 
The definition of differential privacy may be seen as largely a simplification of the previously defined notion of $\epsilon$-\emph{indistinguishability}~\cite{dmns06calibrating}, which explicitly models interaction between a private system and the data examiner as in our definition of differential noninterference. Our definition, however, is cast in the framework of probabilistic automata rather than Turing machines. This supports having structured models that are capable of highlighting issues arising from the bounded memory of actual computers.   Furthermore, we deal with non-termination using prefixes allowing us to leverage previous work on formal methods for automata (e.g.,~\cite{lsv07observing}).  

Differential privacy is a very active research field giving rise to new definitions and techniques at a fast pace~\cite{Dwork10,DNPR10}. For example, \emph{pan-privacy} is a notion of differential privacy that gives differential privacy against adversaries that can observe the internal state of a system, in addition to outputs~\cite{MPRV09}. \emph{Computational differential privacy} gives certain differential privacy guarantees against computationally bounded adversaries. Our definition of differential noninterference and the formal proof technique was developed from the definition of Dwork~\cite{d06differential}.  We think that our choice of probabilistic automata as a model would prove useful in extending the work of this paper to these new definitions as well. For example, algorithms such as stream-processing algorithms that have been subject to research from pan-privacy point of view can be naturally modeled using probabilistic automata. Similarly, probabilistic automata-based models have successfully been used in the formal analysis of cryptographic protocols against computationally bounded adversaries~\cite{st07,backes07reactive,CCKLLPS08}. 

\paragraph{Information-Flow Properties}
Differential noninterference has some similarities with information flow properties such as noninterference~\cite{gm82security}.  
The literature contains several works on the use of transition systems, observational equivalences, and various notions of bisimulation relations to define information flow properties.
 To name a few, Focardi and Gorrieri have developed a classification of noninterference-like properties in the unifying framework of a process algebra in a non-probabilistic setting~\cite{FG01classification}.  Sabelfeld and Sands~\cite{SS00}, and Smith~\cite{Smith03} have used probabilistic bisimulation in defining probabilistic noninterference for multi-threaded programs, which they enforce using type systems.
Probabilistic noninterference is regarded by many to be too strong in practice since it requires the probabilities of traces of the system observable by low-level users to be identical for any pair of high-level inputs (data points in our setting)~\cite{g91toward,g92toward}.
As noninterference is often too strong of a requirement, weaker probabilistic versions have been proposed that allow for some information leakage~\cite{phw04approximate,bp02computation}.
Di Pierro, Hankin, and Wiklicky introduced \emph{approximate noninterference}~\cite{phw04approximate}, and Backes and Pfitzmann introduced \emph{computational probabilistic noninterference}~\cite{bp02computation}, both of which allow for some information leakage.  
However, unlike differential noninterference, they do not allow the system behavior to diverge as the difference between the high-level inputs (data points) increases.  This divergence, which is allowed by our differential noninterference definition (Proposition~\ref{prp:strong-n-composition} in Appendix~\ref{app:mut-diff-proofs}), is needed to release meaningful statistics and gain utility from the data set as discussed in detail in Section~\ref{sec:intro}.
 
Quantitative information flow analysis attempts to determine how much information a program provides an adversary about a 
sensitive input or class of inputs.  Clark, Hunt, and Malacaria present a formal model of programs for quantifying information flows 
and a static analysis that provides lower and upper bounds on the amount of information that flows~\cite{chm07static}.  
They measure information flow as the mutual information between the high-level inputs and low-level outputs given that the 
adversary has control over the low-level inputs.  Malacaria extends this work to handle loops~\cite{m07assessing}, and Chen 
and Malacaria to multi-threaded programs~\cite{cm07quantitiative}.  McCamant and Ernst~\cite{me07simulation}, and Newsome and 
Song~\cite{ns08influence} provide dynamic analyses for quantitative information flow using the mutual information formalization.  
There is also recent work on efficient computation of information leakage in the information theoretic-sense using a probabilistic automaton model~\cite{APRS10}. 
All of the above approaches assume that the adversary's beliefs are aligned with the actual distribution producing the sensitive input(s) 
and that adversary has no additional background knowledge.  
Clarkson, Myers, and Schneider instead propose a formulation using the beliefs of the adversary~\cite{cms05belief}.  However, such a formulation may be difficult to apply in  practice because the surveyor may not know the adversary's beliefs.  
An advantage of differential privacy is that no assumptions are needed about the adversary's auxiliary information, computational power,
or beliefs.

\paragraph{Proof Techniques for Transition Systems}
Simulation and bisimulation provide a systematic proof technique for showing implementation and equivalence relationships between two automata~\cite{CCSbook,LV95,LS95} 
and are related to unwinding (see e.g.,~\cite{bfpr03bisimulation}).  Most similar to our unwinding technique, Segala and Turrini have studied approximate simulation relations in the context of cryptographic protocols~\cite{st07}.  Their work differs from ours by using asymptotic approximations and only executions of polynomial length in terms of a security parameter.  
Their work allows certain transitions of the protocol to not have a matching transition in the specification.
This models the capability of the adversary to compromise correctness.  
A protocol is deemed correct if the leakage accumulated at the end of a polynomial length execution is exponentially small in some security parameter.  Our unwinding technique, on the other hand, requires that there always be an approximately matching transition, uses an exact error bound, and considers executions of any length.  
However, the probabilities of those transitions are only within some exponential multiplicative factor of one another. 
Thus, neither approach subsumes the other.
Furthermore, our relations are over states whereas theirs is over prefixes of executions. 

Much work has been done on decision algorithms for probabilistic simulation and bisimulation~\cite{bhk04,BEM99a,pls00weak,cs02decision}. Particularly relevant are the works of Baier and Hermans~\cite{bhk04}, and Cattani and Segala~\cite{cs02decision} on decision algorithms for weak bisimulations. Since our unwinding relations keep track of an error bound in the form of indices in a relation family, the methods of these papers to generate relations do not readily apply to our setting. We limit ourselves to checking if a given relation family is an unwinding family rather than generating one. Extending these prior works to our setting remains as future work. 

Finding refinement methods that preserve information flow properties has been investigated by several authors~\cite{Mantel01,Jurjens01,HPS01,ACZ06}. In most of those works refinement is used in the sense of reducing various flavors of nondeterminism in an abstract system. For example, Mantel focuses on a range of information flow properties and unwinding conditions as local conditions that imply these properties~\cite{Mantel01}. He then presents some operators that refine a given transition system such that these conditions are preserved in the system refined by the given operators. We have a more restricted goal in this paper, namely, to pin down the conditions under which  an abstract internal transition can be replaced by a sequence of internal transitions in a way that will preserve differential noninterference. This is sufficient for our purposes because such transition replacements are the sources of different abstraction levels that typically arise in the analysis of systems we consider in this paper.

\section{Future Work}
\label{sec:future-work}

The results of this paper represent progress towards developing a basis for the formal verification of differential privacy for systems, but leave open several interesting directions that we plan to explore in future work.
We hope to create a decision procedure for our proof technique by extending prior work on decision procedures for probabilistic bisimulations~\cite{bhk04,BEM99a,pls00weak,cs02decision} to make them produce a family of relations rather than a single one. We also plan to extend the theory to model and reason about higher level systems, such as computer systems of hospitals and other distributed systems~\cite{rrsksw10airavat} that allow interactions of the system with data providers and with data analysts, while protecting the privacy of the data stored and manipulated by the system. For example, {\sc airavat} allows computations over data distributed in a cloud, and combines mandatory access control with 
differential privacy where differential privacy is used to facilitate declassification governed by the privacy error bound set by a data provider. Our techniques can currently apply to the verification of differential privacy property of the {\sc airavat} system using a whole-system model. We are interested in 
exploring the computational model of {\sc airavat} further to understand the interplay between the fine-grained access control mechanisms and the differential privacy mechanisms in stating the end-to-end information-flow guarantee of 
{\sc airavat}. Moreover, we wish to extend compositionality aspects of our framework so that we can decompose the reasoning about such properties, and exploit our proof technique for differential noninterference for parts of the proof. Finally, while the current paper uses manually constructed automata models of systems, we plan to develop techniques to extract such models from source code of software systems such as \textsc{pinq}~\cite{m09privacy} and {\sc airavat}~\cite{rrsksw10airavat}.

\paragraph{Acknowledgments}
We thank Jeremiah Blocki and Michael Dinitz for helping us understand infinity.

\bibliographystyle{alpha}
\bibliography{diff-priv}

\newcommand{\etalchar}[1]{$^{#1}$}
\begin{thebibliography}{RAW{\etalchar{+}}10}

\bibitem[APvRS10]{APRS10}
Miguel~E. Andres, Catuscia Palamidessi, Peter van Rossum, and Geoffrey Smith.
\newblock Computing the leakage of information-hiding systems.
\newblock In {\em Proceedings of Sixteenth International Conference on Tools
  and Algorithms for the Construction and Analysis of Systems (TACAS)}, volume
  6015 of {\em LNCS}, pages 373--389. Springer, 2010.

\bibitem[AvZ06]{ACZ06}
Rajeev Alur, Pavol \v{C}ern\'{y}, and Steve Zdancewic.
\newblock Preserving secrecy under refinement.
\newblock In {\em Proceedings of 33rd International Colloquium on Automata,
  Languages and Programming (ICALP)}, pages 107--118, 2006.

\bibitem[BEMC00]{BEM99a}
C.~Baier, B.~Engelen, and M.~Majster-Cederbaum.
\newblock Deciding bisimilarity and similarity for probabilistic processes.
\newblock {\em Journal of Computer and System Sciences}, 60:187--231, 2000.

\bibitem[BFPR03]{bfpr03bisimulation}
Annalisa Bossi, Riccardo Focardi, Carla Piazza, and Sabina Rossi.
\newblock Bisimulation and unwinding for verifying possibilistic security
  properties.
\newblock In {\em VMCAI 2003: Proceedings of the 4th International Conference
  on Verification, Model Checking, and Abstract Interpretation}, pages
  223--237, London, UK, 2003. Springer-Verlag.

\bibitem[BHK04]{bhk04}
C.~Baier, H.~Hermanns, and J.-P. Katoen.
\newblock Probabilistic weak simulation is decidable in polynomial time.
\newblock {\em Information Processing Letters}, 89(3):123--152, 2004.

\bibitem[BLR08]{blr08}
Avrim Blum, Katrina Ligett, and Aaron Roth.
\newblock A learning theory approach to non-interactive database privacy.
\newblock In {\em STOC '08: Proceedings of the 40th annual ACM symposium on
  Theory of computing}, pages 609--618, New York, NY, USA, 2008. ACM.

\bibitem[BP02]{bp02computation}
Michael Backes and Birgit Pfitzmann.
\newblock Computational probabilistic non-interference.
\newblock In {\em ESORICS '02: Proceedings of the 7th European Symposium on
  Research in Computer Security}, pages 1--23, London, UK, 2002.
  Springer-Verlag.

\bibitem[BPW07]{backes07reactive}
Michael Backes, Birgit Pfitzmann, and Michael Waidner.
\newblock The reactive simulatability framework for asynchronous systems.
\newblock {\em Information and Computation}, 2007.
\newblock Preprint on IACR ePrint 2004/082.

\bibitem[CCK{\etalchar{+}}08]{CCKLLPS08}
R.~Canetti, L.~Cheung, D.~Kaynar, M.~Liskov, N.~Lynch, O.~Pereira, and
  R.~Segala.
\newblock Time-bounded task-pioas: A framework for analyzing security
  protocols.
\newblock {\em Journal of Discrete Event Dynamic Systems}, 18(1):111--159,
  2008.
\newblock Short version appeared In 20th Symposium on Distributed Computing
  (DISC), 2006.

\bibitem[CGP00]{cgp00model}
Edmund~M. Clarke, Orna Grumberg, and Doron~A. Peled.
\newblock {\em Model Checking}.
\newblock {MIT} Press, 2000.

\bibitem[CHM07]{chm07static}
David Clark, Sebastian Hunt, and Pasquale Malacaria.
\newblock A static analysis for quantifying information flow in a simple
  imperative language.
\newblock {\em Journal of Computer Security}, 15:321--371, 2007.

\bibitem[CM07]{cm07quantitiative}
Han Chen and Pasquale Malacaria.
\newblock Quantitative analysis of leakage for multi-threaded programs.
\newblock In {\em PLAS '07: Proceedings of the 2007 workshop on Programming
  languages and analysis for security}, pages 31--40, New York, NY, USA, 2007.
  ACM.

\bibitem[CMS05]{cms05belief}
Michael~R. Clarkson, Andrew~C. Myers, and Fred~B. Schneider.
\newblock Belief in information flow.
\newblock In {\em CSFW '05: Proceedings of the 18th IEEE workshop on Computer
  Security Foundations}, pages 31--45, Washington, DC, USA, 2005. IEEE Computer
  Society.

\bibitem[CS02]{cs02decision}
Stefano Cattani and Roberto Segala.
\newblock Decision algorithms for probabilistic bisimulation.
\newblock In Lubos Brim, Petr Jancar, Mojm\'{\i}r Kret\'{\i}nsk\'{y}, and
  Anton\'{\i}n Kucera, editors, {\em CONCUR '02: Proceedings of the 13th
  International Conference on Concurrency Theory}, volume 2421 of {\em {LNCS}},
  pages 371--385. Springer, 2002.

\bibitem[DMNS06]{dmns06calibrating}
Cynthia Dwork, Frank Mcsherry, Kobbi Nissim, and Adam Smith.
\newblock Calibrating noise to sensitivity in private data analysis.
\newblock In {\em Theory of Cryptography Conference}, volume 3876 of {\em
  Lecture Notes in Computer Science}, pages 265--284. Springer, 2006.

\bibitem[DNPR10]{DNPR10}
Cynthia Dwork, Moni Naor, Toniann Pitassi, and Guy~N. Rothblum.
\newblock Differential privacy under continual observation.
\newblock In {\em In Proceedings of the 42nd {ACM} Syposium on the Theory of
  Computing (STOC)}, 2010.

\bibitem[Dwo06]{d06differential}
Cynthia Dwork.
\newblock Differential privacy.
\newblock In {\em 33rd International Colloquium on Automata, Languages and
  Programming (ICALP 2006)}, volume~2, pages 1--12, 2006.

\bibitem[Dwo08]{d08survey}
Cynthia Dwork.
\newblock {\em Theory and Applications of Models of Computation}, volume 4978,
  chapter Differential Privacy: A Survey of Results, pages 1--19.
\newblock Springer, 2008.

\bibitem[Dwo09]{d09differential}
Cynthia Dwork.
\newblock The differential privacy frontier (extended abstract).
\newblock In {\em 6th Theory of Cryptography Conference}, volume 5444 of {\em
  Lecture Notes in Computer Science}, pages 496--502. Springer, 2009.

\bibitem[Dwo10]{Dwork10}
C.~Dwork.
\newblock Differential privacy in new settings.
\newblock In {\em Proceedings of Symposium on Discrete Algorithms (SODA)}.
  SIAM, 2010.

\bibitem[FG01]{FG01classification}
Riccardo Focardi and Roberto Gorrieri.
\newblock Classification of security properties ({P}art {I}: {I}nformation
  flow), 2001.

\bibitem[GM82]{gm82security}
J.~A. Goguen and J.~Meseguer.
\newblock Security policies and security models.
\newblock In {\em IEEE Symposium on Security and Privacy}, page~11. IEEE, 1982.

\bibitem[GM84]{gm84unwinding}
Joseph~A. Goguen and Jose Meseguer.
\newblock Unwinding and inference control.
\newblock In {\em Proc. of IEEE Symp. on Security and Privacy}, pages 75--86,
  Los Alamitos, CA, USA, 1984. IEEE Computer Society.

\bibitem[Gra91]{g91toward}
James~W. Gray, III.
\newblock Toward a mathematical foundation for information flow security.
\newblock In {\em IEEE Symposium on Security and Privacy}, pages 21--35, 1991.

\bibitem[Gra92]{g92toward}
James~W. Gray, III.
\newblock Toward a mathematical foundation for information.
\newblock {\em Journal of Computer Security}, 1(3-4):255--294, 1992.

\bibitem[GRS09]{grs09universally}
Arpita Ghosh, Tim Roughgarden, and Mukund Sundararajan.
\newblock Universally utility-maximizing privacy mechanisms.
\newblock In {\em STOC '09: Proceedings of the 41st annual ACM symposium on
  Theory of computing}, pages 351--360, New York, NY, USA, 2009. ACM.

\bibitem[GS97]{gs97introduction}
Charles~M. Grinstead and J.~Laurie Snell.
\newblock {\em Introduction to Probability}.
\newblock American Mathematical Society, second revised edition edition, 1997.
\newblock Available at
  \url{http://www.dartmouth.edu/~chance/teaching_aids/books_articles/probabili%
ty_book/book.html}.

\bibitem[HK73]{hk73n52}
John~E. Hopcroft and Richard~M. Karp.
\newblock An $n^{5/2}$ algorithm for maximum matchings in bipartite graphs.
\newblock {\em SIAM Journal on Computing}, 2(4):225--231, 1973.

\bibitem[HPS01]{HPS01}
Maritta Heisel, Andreas Pfitzmann, and Thomas Santen.
\newblock Confidentiality-preserving refinement.
\newblock In {\em Proceedings of 14th IEEE Computer Security Foundations
  Workshop}, pages 295--305. IEEE Press, 2001.

\bibitem[J\"01]{Jurjens01}
Jan J\"{u}rjens.
\newblock Secrecy-preserving refinement.
\newblock In {\em Proceedings of {FME}: Formal Merhods for Increasing Software
  Productivity}, volume 2021 of {\em LNCS}, pages 135--152. Springer-Verlag,
  2001.

\bibitem[LSV07]{lsv07observing}
N.~Lynch, R.~Segala, and F.~Vaandrager.
\newblock Observing branching structure through probabilistic contexts.
\newblock {\em SIAM Journal on Computing}, 37(4):977--1013, 2007.

\bibitem[LV95]{LV95}
Nancy Lynch and Frits Vaandrager.
\newblock Forward and backward simulations {P}art {I}: Untimed systems.
\newblock {\em Inf. Comput.}, 121(2):214--233, 1995.

\bibitem[Mal07]{m07assessing}
Pasquale Malacaria.
\newblock Assessing security threats of looping constructs.
\newblock In {\em POPL '07: Proceedings of the 34th annual ACM SIGPLAN-SIGACT
  symposium on Principles of programming languages}, pages 225--235, New York,
  NY, USA, 2007. ACM.

\bibitem[Man01]{Mantel01}
H.~Mantel.
\newblock Preserving information flow properties under refinement.
\newblock In {\em Proceedings of the {IEEE} Symposium on Security and Privacy}.
  IEEE Press, 2001.

\bibitem[McS09]{m09privacy}
Frank McSherry.
\newblock Privacy integrated queries: An extensible platform for
  privacy-preserving data analysis.
\newblock In {\em SIGMOD '09: Proceedings of the 2009 ACM SIGMOD international
  conference on Management of data}, New York, NY, USA, 2009. ACM.

\bibitem[ME07]{me07simulation}
Stephen McCamant and Michael~D. Ernst.
\newblock A simulation-based proof technique for dynamic information flow.
\newblock In {\em PLAS '07: Proceedings of the 2007 workshop on Programming
  languages and analysis for security}, pages 41--46, New York, NY, USA, 2007.
  ACM.

\bibitem[Mil89]{CCSbook}
Robin Milner.
\newblock {\em Communication and Concurrency}.
\newblock Prentice Hall, 1989.

\bibitem[MPRV09]{MPRV09}
Ilya Mironov, Omkant Pandey, Omer Reingold, and Salil Vadhan.
\newblock Computational differential privacy.
\newblock In {\em Advances in Cryptology -- CRYPTO 2009}, 2009.

\bibitem[MT07]{mt07mechanism}
Frank McSherry and Kunal Talwar.
\newblock Mechanism design via differential privacy.
\newblock In {\em FOCS '07: Proceedings of the 48th Annual IEEE Symposium on
  Foundations of Computer Science}, pages 94--103, Washington, DC, USA, 2007.
  IEEE Computer Society.

\bibitem[NRS07]{nrs07smooth}
Kobbi Nissim, Sofya Raskhodnikova, and Adam Smith.
\newblock Smooth sensitivity and sampling in private data analysis.
\newblock In {\em STOC '07: Proceedings of the thirty-ninth annual ACM
  symposium on Theory of computing}, pages 75--84, New York, NY, USA, 2007.
  ACM.

\bibitem[NS08]{ns08influence}
James Newsome and Dawn Song.
\newblock Influence: A quantitative approach for data integrity.
\newblock Technical Report CMU-CyLab-08-005, CyLab, Carnegie Mellon University,
  2008.

\bibitem[PHW04]{phw04approximate}
Alessandra~Di Pierro, Chris Hankin, and Herbert Wiklicky.
\newblock Approximate non-interference.
\newblock {\em J. Comput. Secur.}, 12(1):37--81, 2004.

\bibitem[PLS00]{pls00weak}
Anna Philippou, Insup Lee, and Oleg Sokolsky.
\newblock Weak bisimulation for probabilistic systems.
\newblock In {\em CONCUR '00: Proceedings of the 11th International Conference
  on Concurrency Theory}, volume 1877 of {\em Lecture Notes in Computer
  Science}, pages 334--349, London, UK, 2000. Springer.

\bibitem[RAW{\etalchar{+}}10]{rawhps10differential}
Jason Reed, Adam~J. Aviv, Daniel Wagner, Andreas Haeberlen, Benjamin~C. Pierce,
  and Jonathan~M. Smith.
\newblock Differential privacy for collaborative security.
\newblock In {\em European Workshop on System Security (EUROSEC)}, April 2010.

\bibitem[RP10]{rp10distance}
Jason Reed and Benjamin~C. Pierce.
\newblock Distance makes the types grow stronger: {A} calculus for differential
  privacy.
\newblock In {\em {ACM} {SIGPLAN} {I}nternational {C}onference on {F}unctional
  {P}rogramming ({ICFP})}, September 2010.

\bibitem[RRS{\etalchar{+}}10]{rrsksw10airavat}
Indrajit Roy, Hany~E. Ramadan, Srinath~T.V. Setty, Ann Kilzer, Vitaly
  Shmatikov, and Emmett Witchel.
\newblock Airavat: Security and privacy for {M}ap{R}educe.
\newblock In {\em Proceedings of the 7th Usenix Symposium on Networked Systems
  Design and Implementation (NSDI)}, 2010.

\bibitem[SL95]{LS95}
Roberto Segala and Nancy Lynch.
\newblock Probabilistic simulations for probabilistic processes.
\newblock {\em Nordic Journal of Computing}, 2(2), 1995.

\bibitem[Smi03]{Smith03}
Geoffrey Smith.
\newblock Probabilistic noninterference through weak probabilistic
  bisimulation.
\newblock In {\em Proceedings of the 16th IEEE Computer Security Foundations
  Workshop}, pages 3--13, Pacific Grove, California, 2003.

\bibitem[SS00]{SS00}
Andrei Sabelfeld and David Sands.
\newblock Probabilistic non-interefence for multi-threaded programs.
\newblock In {\em Proceedings of the 13th IEEE Computer Security Foundations
  Workshop}, Cambridge, England, July 2000. IEEE Computer Society Press.

\bibitem[ST07]{st07}
Roberto Segala and Andrea Turrini.
\newblock Approximated computationally bounded simulation relations for
  probabilistic automata.
\newblock In {\em Proceedings of the 20th IEEE Computer Security Foundations
  Symposium}, pages 140--156, Venice, Italy, 2007.

\end{thebibliography}

\appendix\newcommand{\appsection}[1]{\section{#1}}\newcommand{\appsubsection}[1]{\subsection{#1}}

\appsection{The Truncated Geometric Mechanism}
\label{app:bdm}

\appsubsection{The Mechanism}
The Truncated Geometric Mechanism of Ghosh et al.~\cite{grs09universally} is an adaptation of the Laplace mechanism made to produce outputs over only a bounded range of discrete values.  The Laplace mechanism works by computing the exact result of some statistic $f$ and then adding noise drawn from a Laplace distribution.  The amount of noise depends upon both the privacy parameter $\epsilon$ and the \emph{sensitivity} of $f$.  The sensitivity of $f$ is the amount the value that $f$ computes can change by adding or removing a single data point from the data set.  Formally, the sensitivity of $f$, denoted $\delta(f)$, is maximum value that $|f(B_1) - f(B_2)|$ can take on where $B_1$ and $B_2$ ranges over all pairs of data sets differing by one data point.  Using $\kappa^{\mathsf{LM}}_{f,\epsilon}$ to denote the Laplace mechanism applied to the statistic $f$, we have that $\kappa^{\mathsf{LM}}_{f,\epsilon} (B) = f(B) + \mathsf{Lap}(\delta(f)/\epsilon)$ where $\mathsf{Lap}(b)$ is a random variable producing noise according to the Laplace distribution centered at zero with variance $2b^2$.

To make the Laplace distribution discrete, start by noting that informally the Laplace distribution is two exponential distributions back to back.  That is, $\Pr[\mathsf{Lap}(b) {=} x] = \Pr[\mathsf{Exponential}(1/b) {=} |x|]$ where $\mathsf{Exponential}(\lambda)$ is the exponential distribution with the p.d.f.\ of $\lambda\e{-\lambda x}$ at $x$ for $x \leq 0$ and $0$ otherwise.  Since the discrete version of the exponential distribution is a geometric distribution, one can use two geometric distributions back to back to create a ``discrete'' Laplace distribution.
Formally, $\Pr[\mathsf{Exponential}(\lambda) {=} x] = \Pr[\mathsf{Geo}(\e{-\lambda}) {=} \lfloor x\rfloor]$ where $\Pr[\mathsf{Geo}(p){=}k] = p^k(1-p)$ (i.e., $p$ is the ``failure probability'').  
Using $\mathsf{DL}$ to denote this distribution, we have that $ \Pr[\mathsf{DL}(p){=}n] = p^{|n|}\frac{1-p}{1+p}$.

Next, one must bound the mechanism to produce only results between the minimal and maximum numbers that the computer can represent.  For simplicity we assume that the minimum is $-m$ where $m$ is the maximum.  Thus, we need that the result of adding noise $f(B) + N$ is such that $-m \leq f(B) + N \leq m$ where $N$ is random variable generating noise.  This implies that $-m - f(B) \leq N \leq m - f(B)$ requiring that $N$ depends upon both $m$ and $f(B)$ in addition to $\epsilon$ and $\delta(f)$.

At this point, it may be tempting to simply take the discrete Laplace distribution $\mathsf{DL}$ and condition on the noise being between $-m - f(B)$ and $m - f(B)$.  This will produce a bounded distribution such that the probability of producing two adjacent outputs are within a multiplicative factor of one another.  However, since the condition involves the value of $f(B)$, the distributions resulting from two adjacent data sets may differ.  In general, they need not be within a multiplicative factor of one another.

Fixing this problem requires adding extra weight to the probability of producing the extreme results $-m$ and $m$ for $f(B) + N$.  Intuitively, this extra weight account for the tails being cut off.  Formally, it comes from a system of equations constraining the relationship between each pair of distributions $N(m,f(B_1),\e{-\epsilon/\delta(f)})$ and $N(m,f(B_2),\e{-\epsilon/\delta(f)})$ where $B_1$ and $B_2$ differ by one data point.  Formally,
\begin{align*}
\Pr[N(m,t,p){=}n] &= 
\begin{cases}
p^{|n|} * \frac{1}{1+p} & |t + n| = m\\
p^{|n|} * \frac{1-p}{1+p} & -m < t + n <  m\\
0 & \text{ otherwise}
\end{cases}
\end{align*}

$N$ produces noise for $\kappa_{f,\epsilon}$, a differentially private mechanism for the statistic $f$:
\[ \kappa_{f,\epsilon}(B) = f(B) + N(m,f(B),\e{-\epsilon/\delta(f)}) \]

\begin{proposition}[Differential Privacy]
\label{prp:bdm-diff-priv}
For all integers $m > 0$, for all functions $f$ from data sets to $\{-m, \ldots, m\}$, the function $\kappa_{f,\epsilon}$ has $\epsilon$-differential privacy.
\end{proposition}
\begin{proof}
By a lemma similar to Proposition~\ref{prp:strong-by-trace}, since $\kappa_{f,\epsilon}$ is discrete, it gives $\epsilon$-differential privacy iff for all data sets $B_1$ and $B_2$ differing on at most one element, and for all $r \in \range{\kappa_{f,\epsilon}}$,
\[ \Pr[\kappa_{f,\epsilon}(B_1) = r] \leq \e{\epsilon} * \Pr[\kappa_{f,\epsilon}(B_2) = r] \]

Note
\begin{align*}
\Pr[\kappa_{f,\epsilon}&(B) = r]\\
&= \Pr[f(B) + N(m,f(B),\e{-\epsilon/\delta(f)}) = r]\\
&= \Pr[N(m,f(B),\e{-\epsilon/\delta(f)}) = r-f(B)]\\
&=\begin{cases}
\e{-\epsilon/\delta(f)}^{|r-f(B)|} * \frac{1}{1+\e{-\epsilon/\delta(f)}} & |f(B) + (r-f(B))| = m\\
\e{-\epsilon/\delta(f)}^{|r-f(B)|} * \frac{1-\e{-\epsilon/\delta(f)}}{1+\e{-\epsilon/\delta(f)}} & -m < f(B) + (r-f(B)) <  m\\
0 & \text{ otherwise}
\end{cases}\\
&=\begin{cases}
\e{-|r-f(B)|\epsilon/\delta(f)} * \frac{1}{1+\e{-\epsilon/\delta(f)}} & |r| = m\\
\e{-|r-f(B)|\epsilon/\delta(f)} * \frac{1-\e{\epsilon/\delta(f)}}{1+\e{-\epsilon/\delta(f)}} & -m < r <  m\\
0 & \text{ otherwise}
\end{cases}
\end{align*}
Thus, if $r > m$ or $r < -m$, $\Pr[\kappa_{f,\epsilon}(B_1) = r] = 0 \leq 0 = \e{-\epsilon} * \Pr[\kappa_{f,\epsilon}(B_2) = r]$.  Otherwise, since the normalization factor, which depends on whether $|r| = m$ or not, is the same on each side of the inequality $\Pr[\kappa_{f,\epsilon}(B_1) = r] \leq \e{\epsilon} * \Pr[\kappa_{f,\epsilon}(B_2) = r]$, the inequality holds iff 
\[ e{-|r-f(B_1)|\epsilon/\delta(f)} \leq \e{\epsilon} \e{-|r-f(B_2)|\epsilon/\delta(f)} \]
 
Since $B_1$ and $B_2$ only differ by at most one data point, we know that $|f(B_1) - f(B_2)| \leq \delta(f)$.

Case: $f(B_2)\leq f(B_1)$.  In this case, $f(B_1) - f(B_2) \leq \delta(f)$.
Let $f(B_1) - f(B_2) = \partial$ so that $\e{-|r-f(B_1)|\epsilon/\delta(f)} = \e{-|r-(f(B_2)+\partial)|\epsilon/\delta(f)}$.

\begin{itemize}
\item Subcase: $|r-f(B_2)| \leq |r-(f(B_2)+\partial)|$.  In this case, $\e{-|r-(f(B_2)+\partial)|\epsilon/\delta(f)} \leq \e{-|r-f(B_2)|\epsilon/\delta(f)}$.  Thus, 
\[ \e{-|r-f(B_1)|\epsilon/\delta(f)} \leq \e{\epsilon} \e{-|r-f(B_2)|\epsilon/\delta(f)} \]
 since $\epsilon \leq 0$.

\item Subcase: $|r-f(B_2)| \geq |r-(f(B_2)+\partial)|$.  Let $\partial' = |r-f(B_2)| - |r-(f(B_2)+\partial)|$.  Since $\partial' \leq \partial$,
\begin{align*}
\e{-|r-f(B_1)|\epsilon/\delta(f)} 
&= \e{-(|r-(f(B_2)|-\partial')\epsilon/\delta(f)}\\
&= \e{(\partial'-|r-f(B_2)|)\epsilon/\delta(f)}\\ 
&\leq \e{(\partial-|r-f(B_2)|)\epsilon/\delta(f)}\\
&\leq \e{(\delta(f)-|r-f(B_2)|)\epsilon/\delta(f)}\\
&= \e{(\delta(f)\epsilon/\delta(f))-(|r-f(B_2)|\epsilon/\delta(f))}\\
&= \e{\epsilon-(|r-f(B_2)|\epsilon/\delta(f))}\\
&= \e{\epsilon} \e{-|r-f(B_2)|\epsilon/\delta(f)}
\end{align*}
\end{itemize}

Case: $f(B_1) \leq f(B_2)$.  In this case, $-(f(B_1) - f(B_2)) = f(B_2) - f(B_1) \leq \delta(f)$
Let $f(B_2) - f(B_1) = \partial$ so that 
\[ \e{-|r-f(B_2)|\epsilon/\delta(f)} = \e{-|r-(f(B_1)+\partial)|\epsilon/\delta(f)} \]
\begin{itemize}
\item Subcase: $|r-f(B_1)| \leq |r-(f(B_1)+\partial)|$. 
Let $\partial' = |r-(f(B_1)+\partial)| - |r-f(B_1)|$.  Since $\partial' \leq \partial$,
\begin{align*}
\e{-|r-f(B_1)|\epsilon/\delta(f)} 
&= \e{(-|r-f(B_1)|\epsilon - \delta(f)\epsilon + \delta(f)\epsilon)/\delta(f)}\\
&= \e{((-|r-f(B_1)|-\delta(f))\epsilon + \epsilon\delta(f))/\delta(f)}\\
&= \e{(-|r-f(B_1)|-\delta(f))\epsilon/\delta(f) + \epsilon\delta(f)/\delta(f)}\\
&= \e{(-|r-f(B_1)|-\delta(f))\epsilon/\delta(f) + \epsilon}\\
&\leq \e{\epsilon} \e{(-|r-f(B_1)|-\delta(f))\epsilon/\delta(f)}\\
&\leq \e{\epsilon} \e{(-|r-f(B_1)|-\partial')\epsilon/\delta(f)}\\
&= \e{\epsilon} \e{-(|r-f(B_1)|+\partial')\epsilon/\delta(f)}\\
&= \e{\epsilon} \e{-|r-(f(B_1)+\partial)|\epsilon/\delta(f)}\\
&= \e{\epsilon} \e{-|r-f(B_2)|\epsilon/\delta(f)}
\end{align*}

\item Subcase: $|r-f(B_1)| \geq |r-(f(B_1)+\partial)|$.  
In this case, we have that $\e{-|r-(f(B_1))|\epsilon/\delta(f)} \leq \e{-|r-f(B_1)+\partial|\epsilon/\delta(f)}$.  Thus, 
\[ \e{-|r-f(B_1)|\epsilon/\delta(f)} \leq \e{\epsilon} \e{-|r-f(B_2)|\epsilon/\delta(f)} \]
 since $\epsilon \leq 0$.
\end{itemize}
\end{proof}

The probability of $\kappa_{f,\epsilon}(B)$ being $b$ or more away from $f(B)$ decreases exponentially in $b$.
\begin{proposition}[Utility]
\label{prp:bdm-useful}
$\Pr[|\kappa_{f,\epsilon}(B) - f(B)| \geq b] \leq \frac{2p^b}{1+p}$.
\end{proposition}
\begin{proof}
\begin{align*}
\Pr[|\kappa_{f,\epsilon}&(B) - f(B)| \geq b]\\
&= 1 - \Pr[-b < \kappa_{f,\epsilon}(B) - f(B) < b]\\ 
&= 1 - \Pr[-b+1 \leq \kappa_{f,\epsilon}(B) - f(B) \leq b-1]\\ 
&= 1 - \Pr[-b+1 \leq f(B) + N(m,f(B),\exp(-\epsilon/\delta(f))) - f(B) \leq b-1]\\ 
&= 1 - \Pr[-b+1 \leq N(m,f(B),\exp(-\epsilon/\delta(f))) \leq b-1] \\
&= 1 - \sum_{n=-b+1}^{b-1} \Pr[N(m,f(B),\exp(-\epsilon/\delta(f))) = n] 
\end{align*}
If $b-1 \geq m-f(B)$ and $-b+1 \leq -m-t$, then this is $1-1 = 0$.
If $b-1 < m-f(B)$ and $-b+1 > -m-f(B)$, then this is 
\begin{align*} 
1- \sum_{n=-b+1}^{b-1} \exp(-\epsilon/\delta(f))^n &\frac{1-\exp(-\epsilon/\delta(f))}{1+\exp(-\epsilon/\delta(f))}\\ 
&= 1 - \frac{1+\exp(-\epsilon/\delta(f))-2\exp(-\epsilon/\delta(f))^{b}}{1+\exp(-\epsilon/\delta(f))}\\
&= \frac{2p^b}{1+p}
\end{align*}
If $b-1 < m-f(B)$ and $-b+1 \leq -m-t$, then this is  
\begin{align*}
1- p^{|-m-t|}\frac{1}{1+p} + \sum_{n=-m-t+1}^{b-1} &\exp(-\epsilon/\delta(f))^n \frac{1-\exp(-\epsilon/\delta(f))}{1+\exp(-\epsilon/\delta(f))}\\
&= 1- \frac{1+\exp(-\epsilon/\delta(f))-\exp(-\epsilon/\delta(f))^{b}}{1+\exp(-\epsilon/\delta(f))}\\
&= \frac{p^b}{1+p}
\end{align*}
If $b \geq m-f(B)$ and $-b > -m-t$, then this is  
\begin{align*}
1 - p^{|m-t|}\frac{1}{1+p} + \sum_{n=-b+1}^{m-t-1} &\exp(-\epsilon/\delta(f))^n \frac{1-\exp(-\epsilon/\delta(f))}{1+\exp(-\epsilon/\delta(f))}\\
&= 1- \frac{1+\exp(-\epsilon/\delta(f))-\exp(-\epsilon/\delta(f))^{b}}{1+\exp(-\epsilon/\delta(f))}\\
&= \frac{p^b}{1+p}
\end{align*}
completing the proof.
\end{proof}

\appsubsection{An Implementation}
Below is an efficient algorithm for sampling from $N(\verb|m|,\verb|t|,\verb|p|)$ for $\verb|m| > 0$, $-\verb|m| \leq \verb|t| \leq \verb|m|$, and $0 \leq \verb|p| < 1$:
\begin{verbatim}
01 sample_N(m,t,p)
02   if(flip(p/(1+p)))
03      if(flip(p^(m+t-1)))
04         return(-m-t);
05      else
06         q := (p-1)/(p^(m+t)-p);
07         for(n:=-1; n>-m-t+1; n--)
08            if(flip(q))
09               return(n);
10            q := p*q/(1-q);
11         return(-m-t+1);
12   else
13      if(flip(p^(m-t)))
14         return(m-t);
15      else
16         q := (p-1)*p^t/(p^m-p^t);
17         for(n:=0; n<m-t-1; n++)
18            if(flip(q))
19               return(n);
20            q := p*q/(1-q);
21         return(m-t-1);
\end{verbatim}

Each \verb|flip| command uses an independent Bernoulli distribution to select either true or false.  \verb|flip(|$p$\verb|)| returns true with probability $p$.

\begin{proposition}[Correctness]
\label{prp:bdm-sample-correct}
\verb|sample_N| samples from $N(m,t,p)$.
\end{proposition}
\begin{proof}
Let $q_n$ denote the value that variable $\verb|q|$ has the beginning of the $n$th iteration of the last \verb|for| loop: $q_{0} = \frac{(p-1)p^t}{p^m - p^t}$ and $q_n = p*q_{n-1}/(1-q_{n-1})$ for $0 < n < m-t-1$.  
We show by induction over $n$, that for $n$ between $0$ and $m-t-2$, 
\[ q_n = \left(1-\frac{p}{1+p}\right)^{-1}(1-p^{m-t})^{-1}\prod_{j=0}^{n-1} (1-q_j)^{-1} p^n \frac{1-p}{1+p} \]
For the base case with $n = 0$,
\begin{align*} 
q_0 &= \frac{(p-1)p^t}{p^m-p^t}\\
&= \left(1-\frac{p}{1+p}\right)^{-1}(1-p^{m-t})^{-1}\frac{1-p}{1+p}\\
&= \left(1-\frac{p}{1+p}\right)^{-1}(1-p^{m-t})^{-1}\prod_{j=0}^{-1} (1-q_j)^{-1} p^0 \frac{1-p}{1+p}\\ 
&= \left(1-\frac{p}{1+p}\right)^{-1}(1-p^{m-t})^{-1}\prod_{j=0}^{n-1} (1-q_j)^{-1} p^n \frac{1-p}{1+p} 
\end{align*}
For the inductive case, assume this is true for $n-1$.  Then,
{\small
\begin{align*}
q_n &= p*q_{n-1}/(1-q_{n-1})\\
&= p*\left( \left(1-\frac{p}{1+p}\right)^{-1}(1-p^{m-t})^{-1}\prod_{j=0}^{(n-1)-1} (1-q_j)^{-1} p^{(n-1)} \frac{1-p}{1+p} \right)*(1-q_{n-1})^{-1}\\
&= \left(1-\frac{p}{1+p}\right)^{-1}(1-p^{m-t})^{-1}\prod_{j=0}^{n-1} (1-q_j)^{-1} p^{n} \frac{1-p}{1+p}
\end{align*}
}
If the last \verb|for| loop executes, then with probability $\prod^{n-1}_{j=0} (1-q_j) q_n$ it will stop at the $n$th iteration and return $n$ for values of $n$ between $0$ to $m-t-2$ (inclusive).  Using above equation,
\begin{align*}
\prod^{n-1}_{j=0} &(1-q_j) q_n\\
&= \left(\prod^{n-1}_{j=0} (1-q_j)\right)\left(\left(1-\frac{p}{1+p}\right)^{-1}(1-p^{m-t})^{-1}\prod_{j=0}^{n-1} (1-q_j)^{-1} p^{n} \frac{1-p}{1+p}\right)\\
&= \left(1-\frac{p}{1+p}\right)^{-1}(1-p^{m-t})^{-1} p^n \frac{1-p}{1+p}
\end{align*}
Since the probability of the \verb|for| loop executing is $\left(1-\frac{p}{1+p}\right)(1-p^{m-t})$, this implies that the probability of returning $n$ such that $0 \leq n \leq m-t-2$ is $p^n \frac{1-p}{1+p} = p^{|n|} \frac{1-p}{1+p}$.

For $0 \leq n = m-t$, the probability of returning $m-t$ is $(1-\frac{p}{1+p})p^{m-t} = p^{m-t}*\frac{1}{1+p} = p^{|m-t|}*\frac{1}{1+p}$.

The probability of the \verb|for| running until completion and returning $m-t-1$ is equal to the probability that none of the other values of $n$ is returned.  That is, the probability \verb|flip(p/(1+p))| returning false less the probability of some other number between $0$ and $m-t$ being returned:
\[ (1-\frac{p}{1+p}) - p^{m-t}*\frac{1}{1+p} - \sum_{n=0}^{m-t-2} p^j (1-p)/(1+p) = p^{m-t-1}\frac{1-p}{1+p} = p^{|m-t-1|}\frac{1-p}{1+p} \]

Nearly the same reasoning shows that the negative values for noise also have the correct probabilities.
\end{proof}

Assuming that all the operations in \verb|sample_N| including \verb|flip| are constant time, \verb|sample_N| runs in expected constant time.
\begin{proposition}[Runtime Complexity]
\label{prp:bdm-sample-runtime}
\verb|sample_N| runs in $O(1)$ expected time.
\end{proposition}
\begin{proof}
The expected running time is $\sum_{n=-m-t}^{m-t}  \Pr[N(m,t,p) {=} n] * T_n$ where $T_n$ is the running time of \verb|sample_N| when it produces $n$.  The running time is constant in the case where \verb|sample_N| produces $m-t$ or $-m-t$.  The running time $T_n$ is $O(|n|)$ for $n$ such that $-m-t < n < m-t$.  Thus, ignoring constants, the expected running time is 
\begin{align}
\sum_{n=-m-t+1}^{m-t-1} p^{|n|}\frac{1-p}{1+p} n
&\leq \frac{1-p}{1+p} * 2 * \sum_{n=0}{\max(|-m-t+1|, m-t-1)} p^{n} n\\
&\leq \frac{1-p}{1+p} * 2 * \sum_{n=0}{\infty} p^{n} n\\
&= \frac{1-p}{1+p} * 2 * \frac{p}{1-p} \label{ln:geo-mean}\\
&= \frac{2p}{1+p}\\
&\leq 2
\end{align}
where line~\ref{ln:geo-mean} follows from the expected value of the geometric distribution.  (Recall that we are using $p$ to denote the failure probability unlike most references, which use $1-p$ for the failure probability.)  Thus, it is expected to run in constant time.
\end{proof}

\appsubsection{Using the Mechanism for the Sanitization Functions \textsc{COUNT} and \textsc{SUM}}
\label{app:bdm-count-sum}
We use the above privacy mechanism to implement sanitization functions similar to the ones that \textsc{pinq} provides.  Due to space constraints, we focus on two representative ones: \textsc{count} and \textsc{sum}.  Since we use a bounded discrete privacy mechanism over integers, our implementations differ from the implementations found in \textsc{pinq}.    
We force data points to be integers between $-100$ and $100$ whereas \textsc{pinq} bounds the sensitivity of functions by mapping data points to doubles between $-1$ and $1$.  (Our range may be made larger without affecting our results.)  We then use numbers outside this range to encode objects other than data points such as queries.  

Given these, we implement our \textsc{pinq}-like system as follows.  \verb|datapoint(y)| on line \verb|09| of the code of Figure~\ref{fig:code-window} would be implemented as a function with body \verb|return(-100 <= y && y <= 100)|.  \verb|emptyArray(|$x$\verb|)| must be implemented to store a value outside of $\{-100,\ldots,100\}$ so that data points can be distinguished from empty spots.  The program uses numbers larger than $100$ to indicate queries: $101$ denotes \textsc{count} and $102$ denotes \textsc{sum}.  Given $101$ or $102$, 
\verb|get_sanitization_funct(y)| returns a function that computes the count statistic or sum statistic, respectively.  \textsc{count} is computed with
\begin{verbatim}
01 count(dPts)
02   count := 0;
03   for(j:=0; j<t; j++)
04     for(k:=0; k<maxPts; k++)
05       if(-100 <= dp[j][k] <= 100)
06         count++;
07       else
08         break; 
09   s := t*maxPts/2;
10   noise:=s+sample_N(s,count-s,exp(-e/1));
11   result:=count+noise;
12   return(result);
\end{verbatim}
and \textsc{sum} with
\begin{verbatim}
01 sum(dPts)
02   sum := 0;
03   for(j:=0; j<t; j++)
04     for(k:=0; k<maxPts; k++)
05       if(-100 <= dp[j][k] <= 100)
06           sum := sum+dp[j][k];
07         else
08           break;
09   noise := 
       sample_N(t*maxPts*100,sum,exp(-e/100));
10   result := sum+noise;
11   return(result);
\end{verbatim}
\noindent where \verb|sample_N| is as defined above and \verb|e| stores the value for the privacy bound $\epsilon$.  We add and subtract \verb|t*maxPts| in the calculation of the noise in \mathsc{count} to shift the noise over to keep the value count positive.

\appsection{Automaton Model}
\label{app:autom-details}

\appsubsection{Probability of Action Sequences}

We use $\langle L, s\rangle(\vec{i})(\vec{a},s')$ to denote the probability of the automaton (starting in state $s$) producing the trace $\vec{a}$ and ending in the state $s'$ after producing the last action of $\vec{a}$ given that the available inputs are $\vec{i}$.  $\langle L, s\rangle(\vec{i})(\vec{a},s')$ is defined as follows:
\begin{align*}
\langle L, s\rangle(i\cons\vec{i})(i\cons\vec{a},s') &= \sum_{s'' \in S} \mu(s'') \langle L, s''\rangle(\vec{i})(\vec{a},s') &&\text{if } s \trans^i \mu\\
\langle L, s\rangle(\vec{i})(o\cons\vec{a},s') &= \sum_{s'' \in S} \mu(s'') \langle L, s''\rangle(\vec{i})(\vec{a},s') &&\text{if } s \trans^o \mu\\
\langle L, s\rangle(\vec{i})([\,],s) &= 1\\
\langle L, s\rangle(\vec{i})(\vec{a},s') &= 0 &&\text{\!\!\!otherwise}
\end{align*}
where $i \in I$ and $o \in O$. 
The first line in the above definition, for example, considers the case where the state $s$ transitions to a new state under the input $i$ according to the distribution $\mu$. It states the probability of starting in the state $s$, consuming the input $i$, and then performing the actions $\vec{a}$ ending in state $s'$ given that $\vec{i}$ remain available inputs. This probability is the sum of the probabilities of transitioning to a state $s''$ and then performing the actions $\vec{a}$ from $s''$, ending in state $s'$ given that $\vec{i}$ are available inputs. 

\begin{proposition}\label{prp:well-def-model}
For all automata $\langle L, s\rangle$, $\vec{a}$ in $A^*$, $s'$ in $S$, and $\vec{i}$ in $I^*$, 
$\langle L, s\rangle(\vec{i})(\vec{a},s')$
is well defined and between $0$ and $1$.
\end{proposition}
\begin{proof}
Proof by induction over the structure of $\vec{a}$.

Case: $\vec{a} = [\,]$.  $\langle L, s\rangle(\vec{i})(\vec{a},s')$ is $1$ if $s' = s$, and $\langle L, s\rangle(\vec{i})(\vec{a},s')$ is $0$ for $s' \neq s$.

Case: $\vec{a} = i\cons\vec{a}'$.  If there does not exist $\vec{i}'$ such that $\vec{i} = i\cons\vec{i}'$ and $s \trans^i \mu$, then $\langle L, s\rangle(\vec{i})(\vec{a},s') = 0$.  If there does exist such a $\vec{i}'$, then $\langle L, s\rangle(\vec{i})(\vec{a},s') = \sum_{s'' \in S} \mu(s'') \langle L, s''\rangle(\vec{i}')(\vec{a}')$.  By the inductive hypothesis, $\langle L, s''\rangle(\vec{i}')(\vec{a}',s')$ is well defined and between $0$ and $1$ for all $s''$.  Since $\mu$ is a distribution over states and the events of being in a state are mutually exclusive, $\sum_{s'' \in S} \mu(s'') = 1$.  Let $s_{\mathsf{max}} = \argmax_{s'' \in S} \langle L, s''\rangle(\vec{i}')(\vec{a}',s')$.  
\begin{align*}
\model{\langle L, s\rangle}(\vec{i})(\vec{a},s') &= \sum_{s'' \in S} \mu(s'') \langle L, s'\rangle(\vec{i}')(\vec{a}',s')\\
&\leq \sum_{s'' \in S} \mu(s'') \langle L, s_{\mathsf{max}}\rangle(\vec{i}')(\vec{a}',s')\\
&= \langle L, s_{\mathsf{max}}\rangle(\vec{i}')(\vec{a}',s')\\
&\leq 1
\end{align*}

Case: $\vec{a} = o\cons\vec{a}'$.  If there does not exist $\mu$ such that $s \trans^o \mu$, then $\langle L, s\rangle(\vec{i})(\vec{a},s') = 0$.  If there does, then
$\langle L, s\rangle(\vec{i})(\vec{a},s') = \sum_{s'' \in S} \mu(s'') \langle L, s''\rangle(\vec{i})(\vec{a}',s')$ and we can use the inductive hypothesis as above.
\end{proof}

A helpful proposition about our model follows.
\begin{proposition}\label{prp:mult-apart}
For all \textsc{plts} $L$, states $s', s'' \in S$, $\vec{i}' \in I^*$, and $\vec{h} \in H^*$,
\[ \langle L, s''\rangle(\vec{i}')(\vec{h}\cons\vec{a}', s') = \sum_{s''' \in S}  \langle L, s''\rangle([\,])(\vec{h},s''') * \langle L, s'''\rangle(\vec{i}')(\vec{a}', s') \]
\end{proposition}
\begin{proof}
Proof by induction over the structure of $\vec{h}$.  
In the case where $\vec{h} = [\,]$, $\langle L, s''\rangle([\,])(\vec{h},s''') = 1$ when $s''' = s''$ and $0$ otherwise.  Thus, 
\begin{align*}
\sum_{s''' \in S}  \langle L, s''\rangle([\,])([\,],s''') * \langle L, s'''\rangle(\vec{i}')(\vec{a}', s') &= 1 * \langle L, s''\rangle([\,]\cons\vec{i}')(\vec{a}', s')\\
&=\langle L, s''\rangle(\vec{i}')(\vec{h}\cons\vec{a}', s')
\end{align*}

Case: $\vec{h} = h\cons\vec{h}'$ for some $h$ and $\vec{h}'$.  
If $s'' \trans^h \mu'$, then
\begin{align*}
\sum_{s''' \in S}  &\langle L, s''\rangle([\,])(h\cons\vec{h}',s''') * \langle L, s'''\rangle(\vec{i}')(\vec{a}', s')\\
 &= \sum_{s''' \in S}  \left( \sum_{s'''' \in S} \mu'(s'''') \langle L, s''''\rangle([\,])(\vec{h}',s''')  \right) * \langle L, s'''\rangle(\vec{i}')(\vec{a}', s')\\
 &= \sum_{s'''' \in S} \mu'(s'''') \sum_{s''' \in S} \langle L, s''''\rangle([\,])(\vec{h}',s''') * \langle L, s'''\rangle(\vec{i}')(\vec{a}', s')\\
 &= \sum_{s'''' \in S} \mu'(s'''') \langle L, s''''\rangle(\vec{i}')(\vec{h}'\cons\vec{a}', s')\label{ln:mult-apart-ih}\\
&= \langle L, s''\rangle(\vec{i}')(h\cons\vec{h}'\cons\vec{a}', s')
\end{align*}
where third line follows from the inductive hypothesis.
If for no $\mu'$, then $s'' \trans^h \mu'$, $\langle L, s''\rangle(\vec{i}')(\vec{h}\cons\vec{a}', s') = 0 = \sum_{s''' \in S}  \langle L, s''\rangle([\,])(\vec{h},s''') * \langle L, s'''\rangle(\vec{i}')(\vec{a}', s')$ since $\langle L, s''\rangle([\,])(\vec{h},s''') = 0$ for all $s''' \in S$.
\end{proof}

\appsubsection{Extended Transitions}

We define $s \wtrans^a \nu$ so that $\nu(s')$ is the probability of reaching the $H$-disabled state $s'$ from the state $s$ where $a$ is the action performed from state $s$:
\begin{align*}
\nu(s') &= \sum_{s'' \in S} \mu(s'') \sum_{\vec{h} \in H^*}  \langle L, s''\rangle([\,])(\vec{h},s') &&\text{$s'$ is $H$-disabled}
\end{align*}
and $\nu(s') = 0$ otherwise 
where $s \trans^a \mu$.  Thus, the probability of reaching the $H$-disabled state $s'$ from $s$ by performing the action $a$ followed by a sequence of hidden actions $\vec{h}$ is calculated by considering each $s''$ that is reachable by performing the single action $a$ from $s$. For each such $s''$ we multiply the probability of ending up in $s''$ by performing an $a$ from $s$ with the 
the probability of reaching $s'$ from $s''$ by performing a sequence of hidden actions (the inner sum). The value $\nu(s')$ is then calculated by adding the probabilities corresponding to each $s''$. 
Since all $\vec{h}$ in $H^*$ contain only actions from $H$, an execution with the action sequence $\vec{h}$ cannot leave an $H$-disabled state.  Thus, $\nu(s')$ is the probability of $s'$ being the first $H$-disabled state reached.
If there is no $\mu$ such that $s \trans^a \mu$, then there is no $\nu$ such that $s \wtrans^a \nu$.

For notational convenience we extend the transition relation $\trans$ to $S_\bot$ by having no transitions to nor from $\bot$.  This implies that
\begin{align*}
\langle L, \bot\rangle(\vec{i})([\,],\bot) &= 1\\
\langle L, s\rangle(\vec{i})(\vec{a},\bot) &= 1 - \sum_{s' \in S} \langle L, s\rangle(\vec{i})(\vec{a},s)\\
\langle L, \bot\rangle(\vec{i})(\vec{a},x) &= 0 &&\text{if $\vec{a} \neq [\,]$ or $x \neq \bot$}
\end{align*}
Thus, $\Pr[\,\restrict{\model{\langle L, \bot\rangle}(\vec{i})}{E} {\supseqeq} \vec{e}\,]$ is $1$ if $\vec{e} = [\,]$ and $0$ otherwise, which matches the intuition that a nonterminating program which never interacts with the data examiner will only have the empty trace as a prefix.
\begin{proposition}\label{prp:nu-dist}
For all states $s$ and actions $a$, $s \wtrans^a \nu$ implies that $\nu$ is a distribution over $S_\bot$.
\end{proposition}
\begin{proof}
To prove that $\nu$ is a distribution over $S_\bot$, we must show that for all $x \in S_\bot$, $0 \leq \nu(x) \leq 1$ and $\sum_{x \in S_\bot} \nu(x) = 1$.
We start by proving that $\sum_{x \in S'} \nu(x) \leq 1$ by introducing a function $\eta$.

Given the set $S'$ of $H$-disabled states, let
$\eta$ be defined as follows:
\begin{align*}
\eta(n, s) &= 1 &&\text{$s \in S'$}\\
\eta(n, s) &= 0 &&\text{when $n = 0$ and $s \notin S'$}\\
\eta(n, s) &= \sum_{h \in H} \sum_{s'' \in S} \mu_h(s'') \eta(n-1,s'') &&\text{otherwise}
\end{align*}
where $s \trans^h \mu_h$ and $n$ is a natural number.  

Proof by induction over $n$ shows that
$\eta(n,s) = \sum_{\vec{h} \in H^{\leq n}} \sum_{s' \in S'} \langle L, s\rangle([\,])(\vec{h},s')$.
where $H^{\leq n} = H^n \cup H^{\leq n-1}$ for $n \leq 1$ and $H^{\leq 0} = H^0 = \{[\,]\}$.
In the base case, $n = 0$, if $s \in S'$, then $\eta(n,s) = 1 = \sum_{s' \in S'} \langle L, s\rangle([\,])([\,],s')$ since $\langle L, s\rangle([\,])([\,],s) = 1$, $\langle L, s\rangle([\,])([\,],s') = 0$ for $s \neq s'$, and $s \in S'$.  If $s \notin S'$, 
 $\eta(n,s) = 0 = \sum_{s' \in S'} \langle L, s\rangle([\,])([\,],s')$ since $\langle L, s\rangle([\,])([\,],s') = 0$ for $s \neq s'$ and $s \notin S'$ whereas $s' \in S'$.

In the inductive case, if $s \in S'$, then  
$\langle L, s\rangle([\,])(\vec{h},s') = 0$ if $\vec{h} \neq [\,]$ or $s' \neq s$ since $s$ is $H$-disabled.
Thus, $\eta(n+1,s) = 1 = \sum_{\vec{h} \in H^{\leq n+1}} \sum_{s' \in S'} \langle L, s\rangle([\,])(\vec{h},s')$ since $\langle L, s\rangle([\,])([\,],s) = 1$.
If $s \notin S'$, then
\begin{align}
\eta(n+1, s) &= \sum_{h \in H} \sum_{s'' \in S} \mu_h(s'') \eta(n,s'')\\
&= \sum_{h \in H} \sum_{s'' \in S} \mu_h(s'')  \sum_{\vec{h} \in H^{\leq n}} \sum_{s' \in S'} \langle L, s\rangle([\,])(\vec{h},s')\label{ln:nu-dist-ih}\\
&= \sum_{s' \in S'} \sum_{h \in H}  \sum_{\vec{h} \in H^{\leq n}} \sum_{s'' \in S} \mu_h(s'')  \langle L, s\rangle([\,])(\vec{h},s')\\
&= \sum_{s' \in S'} \sum_{h \in H}  \sum_{\vec{h} \in H^{\leq n}} \langle L, s\rangle([\,])(h\cons\vec{h},s')\label{ln:nu-dist-enabled}\\
&= \sum_{s' \in S'} \langle \sum_{h \in H}  \sum_{\vec{h} \in H^{\leq n}} \langle L, s\rangle([\,])(h\cons\vec{h},s') \rangle) + \langle L, s\rangle([\,])([\,],s')  \label{ln:nu-dist-add-empty}\\
&= \sum_{s' \in S'} \sum_{\vec{h} \in H^{\leq n+1}} \langle L, s\rangle([\,])(\vec{h},s')\\
&= \sum_{\vec{h} \in H^{\leq n+1}} \sum_{s' \in S'} \langle L, s\rangle([\,])(\vec{h},s')
\end{align}
Line~\ref{ln:nu-dist-ih} follows from the inductive hypothesis.
Line~\ref{ln:nu-dist-enabled} follows since $s$ is $H$-enabled.
Line~\ref{ln:nu-dist-add-empty} follows from $\langle L, s\rangle([\,])([\,],s') = 0$ since $s \notin S'$.

Induction over $n$ can also show that $0 \leq \eta(n,s) \leq 1$ since $\mu_h$ is always a distribution.

We use $\eta$ to show the following:
\begin{align*}
\sum_{s' \in S'} \nu(s') 
&= \sum_{s' \in S'} \sum_{s'' \in S} \mu(s'') \sum_{\vec{h} \in H^*} \langle L, s''\rangle([\,])(\vec{h},s')\\
&= \sum_{s'' \in S} \mu(s'') \sum_{\vec{h} \in H^*} \sum_{s' \in S'} \langle L, s''\rangle([\,])(\vec{h},s')\\
&= \sum_{s'' \in S} \mu(s'') \lim_{n \to \infty} \sum_{\vec{h} \in H^{\leq n}} \sum_{s' \in S'} \langle L, s''\rangle([\,])(\vec{h},s')\\
&= \sum_{s'' \in S} \mu(s'') \lim_{n \to \infty} \eta(n,s'')\\
&\leq \sum_{s'' \in S} \mu(s'') 1\\
&\leq 1
\end{align*}
where $s \trans^a \mu$.

For all $s' \in S$, if $s'$ is $H$-enabled, $\nu(s') = 0$.  Thus, $\sum_{s \in S} \nu(s) = \sum_{s' \in S'} \nu(s') \leq 1$.  Furthermore, for all $s \in S$, $0 \leq \nu(s)$ and $0 \leq \sum_{s \in S} \nu(s)$ since no operations that could introduce negative numbers is every used in computing $\nu(s)$.  Since $\nu(\bot) = 1 - \sum_{s \in S} \nu(s)$, $0 \leq \nu(\bot) \leq 1$ and $\sum_{x \in S_\bot} \nu(x) = 1$.  Since for all $s$, $0 \leq \nu(s)$ and $\sum_{s \in S} \nu(s) \leq 1$, it must be the case that $\nu(s) \leq 1$.
\end{proof}

Given such an automaton $\M = \langle L, s\rangle$, we define $\model{\M}$ to be a function from input sequences to a distribution over trace prefixes (finite action sequences).
\[ \Pr[\, \model{\M}(\vec{i}) {\supseqeq} \vec{a}\,] = \sum_{s' \in S} \M(\vec{i})(\vec{a},s') \]

We write $\restrict{\vec{a}}{E}$ for restricting the action sequence $\vec{a}$ to some subset $E$ of $A$.  Formally, $\restrict{[\,]}{E} = [\,]$, $\restrict{a\cons\vec{a}}{E} = a\cons\restrict{\vec{a}}{E}$ if $a \in E$, and $\restrict{a\cons\vec{a}}{E} = \restrict{\vec{a}}{E}$, otherwise.  For infinite sequences $\vec{a}$ with only a finite number of elements from $E$, $\restrict{\vec{a}}{E}$ is the finite sequence that results from $\restrict{\vec{a}'}{E}$ where $\vec{a}'$ is the finite prefix of $\vec{a}$ holding all the elements from $E$.  If $\vec{a}$ contains an infinite number of elements from $E$, then $\restrict{\vec{a}}{E}$ is the infinite sequence whose $j$th entry is the $j$th element of $E$ in $\vec{a}$.

Given an automaton $\M$, $\Pr[\restrict{\model{\M}(\vec{i})}{E} \supseqeq \vec{e}]$ is the probability of the data examiner seeing $\vec{e} \in E^{*}$ as a prefix given that the available inputs are $\vec{i}$.  To calculate $\Pr[\restrict{\model{\M}(\vec{i})}{E} \supseqeq \vec{e}]$, consider the set $\gamma(\vec{e})$ of action sequences $\vec{a}$ such that $\restrict{\vec{a}}{E} = \vec{e}$ and ends with the last element of $\vec{e}$.  That is, $\gamma(\vec{e}) = \set{\vec{a} \in A^*}{\restrict{\vec{a}}{E} = \vec{e} \land \last(\vec{a}) = \last(\vec{e})}$ with the special case that $\gamma([\,]) = \{[\,]\}$.  To calculate $\Pr[\restrict{\model{\M}(\vec{i})}{E} \supseqeq \vec{e}]$, we need not consider all $\vec{a}$ such that $\restrict{\vec{a}}{E} = \vec{e}$.  Rather, we may focus only on those in $\gamma(\vec{e})$ since every $\vec{a}$ such that $\restrict{\vec{a}}{E} = \vec{e}$ will have a prefix in $\gamma(\vec{e})$.  Since it is impossible to see two different prefixes from $\gamma(\vec{e})$ during the same execution (no element of $\gamma(\vec{e})$ is the prefix of another), they are mutually exclusive.  Thus, $\Pr[\restrict{\model{\M}(\vec{i})}{E} \supseqeq \vec{e}] = \sum_{\vec{a} \in \gamma(\vec{e})} \Pr[\model{\M}(\vec{i}) \supseqeq \vec{e}]$.

\appsubsection{Some Helpful Propositions}

We need some propositions about our model to prove the soundness of unwinding later in Appendix~\ref{app:unwinding-soundness}.

Let $H^*\cons\gamma(\vec{e}')$ stand for $\set{\vec{a} \in A^*}{\exists \vec{h} \in H^*, \exists \vec{a}'' \in \gamma(\vec{e}'), \vec{a} = \vec{h}\cons\vec{a}''}$.

We use $\gamma'(\vec{e})$ do denote those action sequences of $\gamma(\vec{e})$ that do not start with a hidden output from $H$: $\gamma'(\vec{e}) = \set{\vec{a} \in \gamma(\vec{e})}{\vec{a} = [\,] \lor \exists a \in A-H, \exists \vec{a}' \in A^*, \vec{a} = a\cons\vec{a}'}$ where $A-H$ is the set difference.

\begin{proposition}\label{prp:hidden-gamma}
For all $\vec{e} \in E^*$, if $\vec{e} \neq [\,]$, then $H^*\cons\gamma'(\vec{e}) = \gamma(\vec{e})$.
\end{proposition}
\begin{proof} 
To show that $H^*\cons\gamma'(\vec{e}) \subseteq \gamma(\vec{e})$, note that for all $\vec{a} \in H^*\cons\gamma'(\vec{e})$, there exists $\vec{h} \in H^*$ and $\vec{a}' \in \gamma'(\vec{e}) \subseteq \gamma(\vec{e})$ such that $\vec{a} = \vec{h}\cons\vec{a}'$.  Furthermore, $\restrict{\vec{h}\cons\vec{a}'}{E} = \restrict{\vec{a}'}{E} = \vec{e}'$ since $H \cap E = \emptyset$.  Since $\vec{e} \neq [\,]$, $\vec{a} \neq [\,]$ and $\mathsf{last}(\vec{h}\cons\vec{a}') = \mathsf{last}(\vec{a}') = \last(\vec{e})$.  Thus, $\vec{h}\cons\vec{a}' \in \gamma(\vec{e})$.

To show that $\gamma(\vec{e}') \subseteq H^*\cons\gamma'(\vec{e})$, for any $\vec{a} \in \gamma(\vec{e}')$, either $\vec{a} \in H^*$ or there exists $\vec{h} \in H^*$, $a \in A-H$, and $\vec{a}' \in A^*$ such that $\vec{a} = \vec{h}\cons a\cons\vec{a}'$.  The first case cannot arise since it would imply that $\vec{e} = [\,]$ since $\vec{e} = \restrict{\vec{a}}{E} = [\,]$.  For the second case, since $\vec{e} = \restrict{\vec{h}\cons a\cons\vec{a}'}{E} = \restrict{a\cons\vec{a}'}{E}$ and $\last(\vec{e}) = \mathsf{last}(\vec{h}\cons a\cons\vec{a}') = \mathsf{last}(a\cons\vec{a}')$.  Thus, $a\cons\vec{a}' \in \gamma(\vec{e})$.  Thus, $\vec{h}\cons a\cons\vec{a}' \in H^*\cons\gamma'(\vec{e})$.
\end{proof}

\begin{proposition}\label{thm:prefix-sem}
\begin{align*}
\Pr[\,\model{\langle L, s\rangle}(i\cons\vec{i}) {\supseqeq} i\cons\vec{a}\,] &= \sum_{s' \in S} \mu(s')\Pr[\,\model{\langle L, s'\rangle}(\vec{i}){\supseqeq}\vec{a}\,]\\
&\mbox{}\hspace{13ex}\text{if } s \trans^i \mu \text{ and } i \in I\\
\Pr[\,\model{\langle L, s\rangle}(\vec{i}){\supseqeq}o\cons\vec{a}\,] &= \sum_{s' \in S} \mu(s')\Pr[\,\model{\langle L, s'\rangle}(\vec{i}){\supseqeq}\vec{a}\,]\\
&\mbox{}\hspace{12ex}\text{if } s \trans^o \mu  \text{ and } o \in O \\
\Pr[\,\model{\langle L, s\rangle}(\vec{i}){\supseqeq}[\,]\,] &= 1\\
\Pr[\,\model{\langle L, s\rangle}(\vec{i}){\supseqeq}\vec{a}\,] &= 0\mbox{}\hspace{15ex}\text{otherwise}
\end{align*}
and $0 \leq \Pr[\,\model{\langle L, s\rangle}(\vec{i}){\supseqeq}\vec{a}\,] \leq 1$.
\end{proposition}
\begin{proof}
For the first equation:
\begin{align*}
\Pr[\,\model{\langle L, s\rangle}(i\cons\vec{i}) {\supseqeq} i\cons\vec{a}\,] 
&= \sum_{s'' \in S} \langle L, s\rangle(i\cons\vec{i})(i\cons\vec{a},s'')\\
&= \sum_{s'' \in S} \sum_{s' \in S} \mu(s') \langle L, s\rangle(\vec{i})(\vec{a},s'')\\
&= \sum_{s' \in S} \mu(s') \sum_{s' \in S} \langle L, s\rangle(\vec{i})(\vec{a},s'')\\
&= \sum_{s' \in S} \mu(s') \Pr[\,\langle L, s\rangle(\vec{i}) {\supseqeq} \vec{a}\,]
\end{align*}

For the second equation:
\begin{align*}
\Pr[\,\model{\langle L, s\rangle}(\vec{i}) {\supseqeq} o\cons\vec{a}\,] 
&= \sum_{s'' \in S} \langle L, s\rangle(\vec{i})(o\cons\vec{a},s'')\\
&= \sum_{s'' \in S} \sum_{s' \in S} \mu(s') \langle L, s\rangle(\vec{i})(\vec{a},s'')\\
&= \sum_{s' \in S} \mu(s') \sum_{s' \in S} \langle L, s\rangle(\vec{i})(\vec{a},s'')\\
&= \sum_{s' \in S} \mu(s') \Pr[\,\langle L, s\rangle(\vec{i}) {\supseqeq} \vec{a}\,]
\end{align*}

For the third equation:
$\Pr[\,\model{\langle L, s\rangle}(\vec{i}){\supseqeq}[\,]\,] = \sum_{s' \in S} \langle L, s\rangle(\vec{i})([\,], s') = 1$ since $\langle L, s\rangle(\vec{i})([\,], s) = 1$ and $\langle L, s\rangle(\vec{i})([\,], s') = 0$ for all $s' \neq s$.

For the forth equation:
$\Pr[\,\model{\langle L, s\rangle}(\vec{i}){\supseqeq}\vec{a}\,] = \sum_{s' \in S} \langle L, s\rangle(\vec{i})([\,], s') = 0$ since $\langle L, s\rangle(\vec{i})([\,], s') = 0$ for all $s'$.

To show that $0 \leq \Pr[\,\model{\langle L, s\rangle}(\vec{i}){\supseqeq} \vec{a}\,] \leq 1$, 
we use proof by induction over the structure of $\vec{a}$.

Case: $\vec{a} = [\,]$.  $\model{\langle L, s\rangle}(\vec{i})(\vec{a})$ is $1$.

Case: $\vec{a} = i\cons\vec{a}'$.  If there does not exist $\vec{i}'$ such that $\vec{i} = i\cons\vec{i}'$ and $s \trans^i \mu$, then $\model{\langle L, s\rangle}(\vec{i})(\vec{a}) = 0$.  If there does exist such a $\vec{i}'$, then $\model{\langle L, s\rangle}(\vec{i})(\vec{a}) = \sum_{s' \in S} \mu(s')\model{\langle L, s'\rangle}(\vec{i}')(\vec{a}')$.  By the inductive hypothesis, $\model{\langle L, s'\rangle}(\vec{i}')(\vec{a}')$ is well defined and between $0$ and $1$ for all $s'$.  Since $\mu$ is a distribution over states and the events of being in a state are mutually exclusive, $\sum_{s' \in S} \mu(s') = 1$.  Let $s_{\mathsf{max}} = \argmax_{s' \in S} \model{\langle L, s'\rangle}(\vec{i}')(\vec{a}')$.  
\begin{align*}
\model{\langle L, s\rangle}(\vec{i})(\vec{a}) &= \sum_{s' \in S} \mu(s')\model{\langle L, s'\rangle}(\vec{i}')(\vec{a}') \leq \sum_{s' \in S} \mu(s')\model{\langle L, s_{\mathsf{max}}\rangle}(\vec{i}')(\vec{a}')\\ 
&= \model{\langle L, s_{\mathsf{max}}\rangle}(\vec{i}') (\vec{a}')\\
&\leq 1
\end{align*}

Case: $\vec{a} = o\cons\vec{a}'$.  If there does not exist $\mu$ such that $s \trans^o \mu$, then $\model{\langle L, s\rangle}(\vec{i})(\vec{a}) = 0$.  If there does, then
$\model{\langle L, s\rangle}(\vec{i})(\vec{a}) = \sum_{s' \in S} \mu(s')\model{\langle L, s'\rangle}(\vec{i})(\vec{a}')$ and we can use the inductive hypothesis as above.
\end{proof}

\begin{proposition}\label{thm:nu-help}
For all $H$-disabled states $s$, $a$ in $D \cup Q \cup R$, $\vec{e}$ in $E^*$, and $\vec{i}$ in $I^*$,
if $\vec{e} \neq [\,]$, $s \trans^a \mu$, and $s \wtrans^a \nu$, then
\[ \sum_{\vec{a} \in \gamma(\vec{e})} \sum_{s' \in S} \sum_{s'' \in S} \mu(s'') * \langle L, s''\rangle(\vec{i})(\vec{a}, s') =  \sum_{x \in S_\bot} \nu(x) \Pr[\,\restrict{\model{\langle L, x\rangle}(\vec{i})}{E} {\supseqeq} \vec{e}\,] \]
\end{proposition}
\begin{proof}
\begin{align}
\sum_{\vec{a} \in \gamma(\vec{e})} &\sum_{s' \in S} \sum_{s'' \in S} \mu(s'') * \langle L, s''\rangle(\vec{i})(\vec{a}, s')\\
&= \sum_{s' \in S} \sum_{s'' \in S} \mu(s'') \sum_{\vec{a} \in \gamma(\vec{e})} \langle L, s''\rangle(\vec{i})(\vec{a}, s')\\
&= \sum_{s' \in S} \sum_{s'' \in S} \mu(s'') \sum_{\vec{a} \in H^*\cons\gamma'(\vec{e})} \langle L, s''\rangle(\vec{i})(\vec{a}, s') \label{ln:nu-help-cons-apart}\\
&= \sum_{s' \in S} \sum_{s'' \in S} \mu(s'') \sum_{\vec{h} \in H^*} \sum_{\vec{a} \in \gamma'(\vec{e})}  \langle L, s''\rangle(\vec{i})(\vec{h}\cons\vec{a}, s')\label{ln:nu-help-cons-to-cross}\\
&= \sum_{s' \in S} \sum_{s'' \in S} \mu(s'') \sum_{\vec{h} \in H^*} \sum_{\vec{a} \in \gamma(\vec{e})}  \sum_{s''' \in S}  \langle L, s''\rangle([\,])(\vec{h},s''') * \langle L, s'''\rangle(\vec{i})(\vec{a}, s') \label{ln:nu-help-mult-apart}\\
&= \sum_{s' \in S} \sum_{s'' \in S} \mu(s'') \sum_{\vec{h} \in H^*} \sum_{\vec{a} \in \gamma(\vec{e})}  \sum_{s''' \in S'}  \langle L, s''\rangle([\,])(\vec{h},s''') * \langle L, s'''\rangle(\vec{i})(\vec{a}, s') \label{ln:nu-help-drop-states}\\
&= \sum_{s''' \in S'} \left(\sum_{s'' \in S} \mu(s'') \sum_{\vec{h} \in H^*} \langle L, s''\rangle([\,])(\vec{h},s''') \right) \sum_{\vec{a} \in \gamma(\vec{e})} \sum_{s' \in S} \langle L, s'''\rangle(\vec{i})(\vec{a}, s')\\
&= \sum_{s''' \in S'} \nu(s''') \sum_{\vec{a} \in \gamma(\vec{e})} \Pr[\,\model{\langle L, s'''\rangle}(\vec{i}) {\supseqeq} \vec{a}\,]\label{ln:nu-help-nu}\\
&= \sum_{s''' \in S'} \nu(s''') \Pr[\,\restrict{\model{\langle L, s'''\rangle}(\vec{i})}{E} {\supseqeq} \vec{e}\,]\\
&= \sum_{x \in S_\bot} \nu(x) \Pr[\,\restrict{\model{\langle L, x\rangle}(\vec{i})}{E} {\supseqeq} \vec{e}\,]\label{ln:nu-help-bot-zero}
\end{align}
where $S'$ is the subset of states $S$ that are $H$-disabled.
Line~\ref{ln:nu-help-cons-apart} follows from Proposition~\ref{prp:hidden-gamma}.
Line~\ref{ln:nu-help-cons-to-cross} follows since there is a one-to-one correspondence between elements of $H^*\cons\gamma'(\vec{e})$ and $H^*\cross\gamma'(\vec{e})$ given as $\vec{a} \in H^*\cons\gamma'(\vec{e})$ corresponding to $\langle \vec{h}, \vec{a}\rangle$ where $\vec{h}$ is the largest sequence of $H^*$ such that $\vec{a} = \vec{h}\cons\vec{a}'$ for some $\vec{a}'$.
Line~\ref{ln:nu-help-mult-apart} follows from Proposition~\ref{prp:mult-apart}.
Line~\ref{ln:nu-help-drop-states} follows since $\vec{a} \in \gamma'(\vec{e})$ starting with an action not in $H$ implies that $\langle L, s'''\rangle(\vec{i})(\vec{a}, s') = 0$ for any state that is $H$-enabled.
Line~\ref{ln:nu-help-bot-zero} follows from $\Pr[\,\restrict{\model{\langle L, \bot\rangle}(\vec{i})}{E} {\supseqeq} e\cons\vec{e}\,] = 0$ since $e\cons\vec{e} \neq [\,]$ and $\nu(s''') = 0$ for any $H$-enabled state $s'''$.
\end{proof}

Informally speaking the following proposition shows how we can account for transitions on hidden actions in calculating the probability of observing a particular behavior from a given state. The first part of the proposition states that the probability of observing the sequence $\vec{e}$ starting from the state $s$ given the input sequence $d\cons\vec{i}'$ can be calculated by considering those states that are reachable from $s$ by performing the action $d$ followed by a sequence of hidden actions. For each such reachable state we take the probability of being in that state and multiply it with the probability of observing the sequence $\vec{e}$ from that state given the input sequence $\vec{i}'$. The other parts can be explained analogously.

\begin{proposition}\label{thm:nu}
For all \textsc{plts}, $s \in S$, $d \in D$, $q \in Q$, $r \in R$, $\vec{i},\vec{i}' \in I^*$, and $\vec{e},\vec{e}' \in E^*$,
\begin{align*}
\Pr[\,&\restrict{\model{\langle L, s\rangle}(d\cons\vec{i}')}{E} {\supseqeq} \vec{e}\,]\\
&= \sum_{x \in S_\bot} \nu(x) \Pr[\,\restrict{\model{\langle L, x\rangle}(\vec{i}')}{E} {\supseqeq} \vec{e}\,] &&\hspace{0ex}\text{where $s \wtrans^d \nu$}\\
\Pr[\,&\restrict{\model{\langle L, s\rangle}(q\cons\vec{i}')}{E} {\supseqeq} q\cons\vec{e}'\,] \\
&= \sum_{x \in S_\bot} \nu(x) \Pr[\,\restrict{\model{\langle L, x\rangle}(\vec{i}')}{E} {\supseqeq} \vec{e}'\,] &&\hspace{0ex}\text{where $s \wtrans^q \nu$}\\
\Pr[\,&\restrict{\model{\langle L, s\rangle}(\vec{i})}{E} {\supseqeq} r\cons\vec{e}'\,]\\
&= \sum_{x \in S_\bot} \nu(x) \Pr[\,\restrict{\model{\langle L, x\rangle}(\vec{i})}{E} {\supseqeq} \vec{e}'\,] &&\hspace{0ex}\text{where $s \wtrans^r \nu$}\\
\end{align*}
\end{proposition}
\begin{proof}
For the first equality of the proposition:
Note that if $\vec{e} = [\,]$, then 
\[ \Pr[\,\restrict{\model{\langle L, s\rangle}(d\cons\vec{i}')}{E} {\supseqeq} \vec{e}\,] = 1 = \sum_{x \in S_\bot} \nu(x) \Pr[\,\restrict{\model{\langle L, x\rangle}(\vec{i}')}{E} {\supseqeq} \vec{e}\,] \]
Otherwise, since $s \wtrans^d \nu$, we know there exists $\mu$ such that $s \trans^d \mu$.  It follows that
\begin{align}
\Pr[\,\restrict{\model{\langle L, s\rangle}(d\cons\vec{i}')}{E} {\supseqeq} \vec{e}\,] 
&= \sum_{\vec{a} \in \gamma(\vec{e})} \Pr[\,\model{\langle L, s\rangle}(d\cons\vec{i}') {\supseqeq} \vec{a}\,]\\
&= \sum_{\vec{a}' \in \gamma(\vec{e})} \Pr[\,\model{\langle L, s\rangle}(d\cons\vec{i}') {\supseqeq} d\cons\vec{a}'\,] \label{ln:nu-d-form}\\
&= \sum_{\vec{a}' \in \gamma(\vec{e})} \sum_{s' \in S} \langle L, s\rangle(d\cons\vec{i}')(d\cons\vec{a}', s')\\
&= \sum_{\vec{a}' \in \gamma(\vec{e})} \sum_{s' \in S} \sum_{s'' \in S} \mu(s'') * \langle L, s''\rangle(\vec{i}')(\vec{a}', s')\\
&= \sum_{x \in S_\bot} \nu(x) \Pr[\,\restrict{\model{\langle L, x\rangle}(\vec{i}')}{E} {\supseqeq} \vec{e}\,]\label{ln:nu-d-nu-help}
\end{align}
Line~\ref{ln:nu-d-form} follows since $s \trans^d \mu$ implies that $s$ does not transition under any outputs and $\vec{e} \neq [\,]$ implies that $\vec{a} \in \gamma(\vec{e})$ cannot be $[\,]$.  Thus, we know that the first action of $\vec{a}$ must be of the form $d\cons\vec{a}'$ for $\Pr[\,\model{\langle L, s\rangle}(d\cons\vec{i}') {\supseqeq} \vec{a}\,]$ to be non-zero.  Since $\restrict{d\cons\vec{a}'}{E} = \vec{e}$ and $d \notin E$, $\restrict{\vec{a}'}{E} = \vec{e}$.  Furthermore, $\mathsf{last}(\vec{a}') = \mathsf{last}(\vec{a}) = \mathsf{last}(\vec{e})$.  Thus, $\vec{a}' \in \gamma(\vec{e})$.
Line~\ref{ln:nu-d-nu-help} follows from Proposition~\ref{thm:nu-help}.

For the second equality of the proposition:
Note that if $\vec{e}' = [\,]$, then 
\[ \Pr[\,\restrict{\model{\langle L, s\rangle}(q\cons\vec{i}')}{E} {\supseqeq} q\cons\vec{e}'\,] = 1 = \sum_{x \in S_\bot} \nu(x) \Pr[\,\restrict{\model{\langle L, x\rangle}(\vec{i}')}{E} {\supseqeq} \vec{e}'\,] \]
Otherwise, since $s \wtrans^q \nu$, we know there exists $\mu$ such that $s \trans^q \mu$.  It follows that
\begin{align}
\Pr[\,\restrict{\model{\langle L, s\rangle}(q\cons\vec{i}')}{E} {\supseqeq} q\cons\vec{e}'\,] 
&= \sum_{\vec{a} \in \gamma(q\cons\vec{e}')} \Pr[\,\model{\langle L, s\rangle}(q\cons\vec{i}') {\supseqeq} \vec{a}\,]\\
&= \sum_{\vec{a}' \in \gamma(\vec{e}')} \Pr[\,\model{\langle L, s\rangle}(q\cons\vec{i}') {\supseqeq} q\cons\vec{a}'\,] \label{ln:nu-q-form}\\
&= \sum_{\vec{a}' \in \gamma(\vec{e}')} \sum_{s' \in S} \langle L, s\rangle(q\cons\vec{i}')(q\cons\vec{a}', s')\\
&= \sum_{\vec{a}' \in \gamma(\vec{e}')} \sum_{s' \in S} \sum_{s'' \in S} \mu(s'') * \langle L, s''\rangle(\vec{i}')(\vec{a}', s')\\
&= \sum_{x \in S_\bot} \nu(x) \Pr[\,\restrict{\model{\langle L, x\rangle}(\vec{i}')}{E} {\supseqeq} \vec{e}'\,]\label{ln:nu-q-nu-help}
\end{align}
Line~\ref{ln:nu-q-form} follows since $s \trans^q \mu$ implies that $s$ does not transition under any outputs and $\vec{e}' \neq [\,]$ implies that $\vec{a}' \in \gamma(\vec{e}')$ cannot be $[\,]$.  Thus, we know that the first action of $\vec{a}$ must be of the form $q\cons\vec{a}'$ for $\Pr[\,\model{\langle L, s\rangle}(q\cons\vec{i}') {\supseqeq} \vec{a}\,]$ to be non-zero.  Since $\restrict{q\cons\vec{a}'}{E} = q\cons\vec{e}'$, $\restrict{\vec{a}'}{E} = \vec{e}'$.  Furthermore, $\mathsf{last}(\vec{a}') = \mathsf{last}(\vec{a}) = \mathsf{last}(\vec{e}')$.  Thus, $\vec{a}' \in \gamma(\vec{e}')$.
Line~\ref{ln:nu-q-nu-help} follows from Proposition~\ref{thm:nu-help}.

For the third equality of the proposition:
Note that if $\vec{e}' = [\,]$, then 
\[ \Pr[\,\restrict{\model{\langle L, s\rangle}(\vec{i})}{E} {\supseqeq} r\cons\vec{e}'\,] = 1 = \sum_{x \in S_\bot} \nu(x) \Pr[\,\restrict{\model{\langle L, x\rangle}(\vec{i})}{E} {\supseqeq} \vec{e}'\,] \]
Otherwise, since $s \wtrans^r \nu$, we know there exists $\mu$ such that $s \trans^r \mu$.  It follows that
\begin{align}
\Pr[\,\restrict{\model{\langle L, s\rangle}(\vec{i})}{E} {\supseqeq} r\cons\vec{e}'\,] 
&= \sum_{\vec{a} \in \gamma(r\cons\vec{e})} \Pr[\,\model{\langle L, s\rangle}(\vec{i}) {\supseqeq} \vec{a}\,]\\
&= \sum_{\vec{a}' \in \gamma(\vec{e}')} \Pr[\,\model{\langle L, s\rangle}(\vec{i}) {\supseqeq} r\cons\vec{a}'\,] \label{ln:nu-r-form}\\
&= \sum_{\vec{a} \in \gamma(\vec{e}')} \sum_{s' \in S} \langle L, s\rangle(\vec{i})(r\cons\vec{a}', s')\\
&= \sum_{\vec{a} \in \gamma(\vec{e}')} \sum_{s' \in S} \sum_{s'' \in S} \mu(s'') * \langle L, s''\rangle(\vec{i})(\vec{a}', s')\\
&= \sum_{x \in S_\bot} \nu(x) \Pr[\,\restrict{\model{\langle L, x\rangle}(\vec{i})}{E} {\supseqeq} \vec{e}'\,]\label{ln:nu-r-nu-help}
\end{align}
Line~\ref{ln:nu-r-form} follows since $s \trans^r \mu$ implies that $s$ does not transition under any action other than $r$ and $\vec{e}' \neq [\,]$ implies that $\vec{a}' \in \gamma(\vec{e}')$ cannot be $[\,]$.  Thus, we know that the first action of $\vec{a}$ must be of the form $r\cons\vec{a}'$ for $\Pr[\,\model{\langle L, s\rangle}(\vec{i}) {\supseqeq} \vec{a}\,]$ to be non-zero.  Since $\restrict{r\cons\vec{a}'}{E} = r\cons\vec{e}'$, $\restrict{\vec{a}'}{E} = \vec{e}'$.  Furthermore, $\mathsf{last}(\vec{a}') = \mathsf{last}(\vec{a}) = \mathsf{last}(\vec{e}')$.  Thus, $\vec{a}' \in \gamma(\vec{e}')$.
Line~\ref{ln:nu-r-nu-help} follows from Proposition~\ref{thm:nu-help}.
\end{proof}

\appsection{Basic Properties of Differential Noninterference}
\label{app:mut-diff-proofs}

\paragraph{Sequence Differencing}
Given the input sequences $\vec{i}_1$ and $\vec{i}_2$, $\diff{\vec{i}_1}{\vec{i}_2}$  denotes the number of data points on which they differ: the minimum total number of data point insertions into $\vec{i}_1$ and $\vec{i}_2$ it takes to make them equal. Formally,
\begin{itemize}
\item $\diff{\vec{i}_1}{\vec{i}_2} = 0$ iff $\vec{i}_1 = \vec{i}_2$. 
\item For $1 \leq n$, $\diff{\vec{i}_1}{\vec{i}_2} = n$ iff there exists $d \in D$, $\vec{i}, \vec{i}'_1, \vec{i}'_2 \in I^*$, such that both of the following properties hold: 
    \begin{itemize}
    \item either $\vec{i}_1 = \vec{i}\cons d\cons\vec{i}'_1$ and $\vec{i}_2 = \vec{i}\cons\vec{i}'_2$, or $\vec{i}_1 = \vec{i}\cons\vec{i}'_1$ and $\vec{i}_2 = \vec{i}\cons d\cons\vec{i}'_2$; and
    \item $\diff{\vec{i}'_1}{\vec{i}'_2} = n-1$.
\end{itemize}
\end{itemize}

For $\diff{\vec{i}_1}{\vec{i}_2} = n$ to hold for any $n$, $\vec{i}_1$ and $\vec{i}_2$ must agree on every query from $Q$: they may only differ by $n$ data points from $D$.
Since differential privacy is defined using data sets differing on one element, in most theorems we are interested in the case where $\diff{\vec{i}_1}{\vec{i}_2} = 1$, which means that there exists $d \in D$, and $\vec{i},\vec{i}' \in I^*$ such that either $\vec{i}_1 = \vec{i}\cons d\cons\vec{i}'$ and $\vec{i}_2 = \vec{i}\cons\vec{i}'$, or $\vec{i}_2 = \vec{i}\cons d\cons\vec{i}'$ and $\vec{i}_1 = \vec{i}\cons\vec{i}'$.

For example, let $d_1$ and $d_2$ range over elements in $D$, and $q_1$ and $q_2$ range over elements in $Q$. 
\begin{itemize}
\item $\Delta([d_1,q_1,d_2], [d_1,q_1]) = 1$ \ (add $d_2$ to the end of the second sequence to get the first).
\item $\Delta([q_1,d_2,q_2], [d_1,q_1,d_2,q_2]) = 1$ \ (add $d_1$ to the front of the first to get the second).
\item $\Delta([d_1,q_1,d_2,q_2], [d_1,q_1,q_2]) = 1$ \ (add $d_2$ between $q_1$ and $q_2$ of the second to get the first).
\item $\Delta([d_1,d_2,q_1,q_2], [d_1,d_2,q_2,q_1])$ is undefined \ (the two sequences do not agree on queries).
\end{itemize}
Note that in the first example, the two sequences have a difference of one under the above definition but do not have a Hamming distance since they are of different lengths.

While the choice of using all possible subsets of the set of trace prefixes instead of a single prefix makes the power of differential noninterference more apparent, it does not actually impose a stronger requirement as shown by the next lemma.  This result simplifies reasoning about differential noninterference and is useful for proving subsequent results in this paper. 

\begin{proposition}\label{prp:strong-by-trace}
$\m$ has $\epsilon$-differential noninterference if and only if for all input sequences $\vec{i}_1$ and $\vec{i}_2$ in $I$ such that $\diff{\vec{i}_1}{\vec{i}_2} \leq 1$ and $\vec{e}$ in $E^*$,
\[ \Pr[\restrict{\m(\vec{i}_1)}{E} \supseqeq \vec{e}] \leq \e{\epsilon} * \Pr[\restrict{\m(\vec{i}_2)}{E} \supseqeq \vec{e}]\]
\end{proposition}
\begin{proof}
The only if direction follows directly from the definition by setting $S = \{\vec{e}\}$.

For the if direction, arbitrarily fix $\vec{i}_1$ and $\vec{i}_2$ such that $\diff{\vec{i}_1}{\vec{i}_2} \leq 1$ and $S \subseteq E^*$.  By assumption, for all $\vec{e}$ in $E^*$,
\[ \Pr[\restrict{\m(\vec{i}_1)}{E} \supseqeq \vec{e}] \leq \e{\epsilon} * \Pr[\restrict{\m(\vec{i}_2)}{E} \supseqeq \vec{e}]\]

Let $S'$ be $S$ with all the elements that are a longer version of another element of $S$ removed.  That is, $S' = \set{\vec{e}' \in S}{\nexists \vec{e} \in S \st \vec{e}' \supseq \vec{e}}$ where  $\vec{e}' \supseq \vec{e}$ means that $\vec{e}'$ is a strict prefix of $\vec{e}$.  Proof by induction over the length of $\vec{e}$ shows that for all $\vec{e}$ in $S$, there exists $\vec{e}'$ in $S'$ such that $\vec{e} \supseqeq \vec{e}'$.  Thus, if there exists $\vec{e}$ in $S$ such that $\restrict{\m(\vec{i}_1)}{E} \supseqeq \vec{e}$, then there exists $\vec{e}'$ in $S'$ such that $\restrict{\m(\vec{i}_1)}{E} \supseqeq \vec{e}'$.  Thus, for all $\vec{i}$, $\Pr[\restrict{\m(\vec{i})}{E} \supseqeq S] = \Pr[\restrict{\m(\vec{i})}{E} \supseqeq S']$.

For two $\vec{e}'_1$ and $\vec{e}'_2$ in $S'$ such that $\vec{e}'_1 \neq \vec{e}'_2$, $\restrict{\m(\vec{i})}{E}$ can only have one of them as a prefix since neither is a prefix of the other.  Thus, since $S$ is countable, this implies that  
\[ \Pr[\restrict{\m(\vec{i})}{E} \supseqeq S] = \sum_{\vec{e}' \in S'} \Pr[\restrict{\m(\vec{i})}{E} \supseqeq \vec{e}'] \]

Thus,
\begin{align*}
\Pr[\restrict{\m(\vec{i}_1)}{E} \supseqeq S] 
&= \Pr[\restrict{\m(\vec{i}_1)}{E} \supseqeq S']\\ 
&= \sum_{\vec{e}' \in S'} \Pr[\restrict{\m(\vec{i}_1)}{E} \supseqeq \vec{e}']\\
&\leq \sum_{\vec{e}' \in S'} \e{\epsilon} * \Pr[\restrict{\m(\vec{i}_2)}{E} \supseqeq \vec{e}']\\
&= \e{\epsilon} \sum_{\vec{e}' \in S'}  \Pr[\restrict{\m(\vec{i}_2)}{E} \supseqeq \vec{e}']\\
&= \e{\epsilon} \Pr[\restrict{\m(\vec{i}_2)}{E} \supseqeq S']\\
&= \e{\epsilon} \Pr[\restrict{\m(\vec{i}_2)}{E} \supseqeq S]
\end{align*}
\end{proof}

The next theorem is analogous to previous results about differential privacy for functions: it proves that the privacy leakage bound for a system whose inputs differ on at most $n$ data points is $n * \epsilon$ where $\epsilon$ is the leakage bound for the system if its inputs differ on one data point (see e.g., corollary of~\cite{mt07mechanism}).
\begin{proposition}
\label{prp:strong-n-composition}
If a system $\m$ has $\epsilon$-differential noninterference, then for all input sequences $\vec{i}_1$ and $\vec{i}_2$ such that $\diff{\vec{i}_1}{\vec{i}_2} \leq n$ and for all $S \subseteq E^*$,
\[ \Pr[\restrict{\model{\M}(\vec{i}_1)}{E} \supseqeq S] \leq \e{n * \epsilon} * \Pr[\restrict{\model{\M}(\vec{i}_2)}{E} \supseqeq S]\]
\end{proposition}
\begin{proof}
Proof by induction over $n$.

Base Case: $n = 0$.  In this case, $\vec{i}_1 = \vec{i}_2$ and, thus, $\Pr[\restrict{\m(\vec{i}_1)}{E} \supseqeq S] = \Pr[\restrict{\m(\vec{i}_2)}{E} \supseqeq S]$ as needed with $\e{0} = 1$.

Inductive Case: Assume for all $n' \leq n$; prove for $n+1$.  Since $\diff{\vec{i}_1}{\vec{i}_2} = n+1$, 
there must exist $\vec{i}, \vec{i}'_1, \vec{i}'_2 \in I^*$ and $d_1\in D$ such that $\vec{i}_1 = \vec{i}\cons d_1\cons\vec{i}'_1$, $\vec{i}_2 = \vec{i} \cons\vec{i}'_2$, and $\diff{\vec{i}'_1}{\vec{i}'_2} = n$.
Let $\vec{i}_3 = \vec{i}\cons i_2\cons \vec{i}'_1$.  $\diff{\vec{i}_1}{\vec{i}_3} = 1$ and $\diff{\vec{i}_3}{\vec{i}_2} = n$.  Thus, by the inductive hypothesis,
 \[ \Pr[\restrict{\m(\vec{i}_1)}{E} \supseqeq S] \leq \e{1 * \epsilon} * \Pr[\restrict{\m(\vec{i}_3)}{E} \supseqeq S]\]
and
\[ \Pr[\restrict{\m(\vec{i}_3)}{E} \supseqeq S] \leq \e{n * \epsilon} * \Pr[\restrict{\m(\vec{i}_2)}{E} \supseqeq S]\]
Thus,
\begin{align*}
\Pr[\restrict{\m(\vec{i}_1)}{E} \supseqeq S] &\leq \e{1 * \epsilon} * \left( \e{n * \epsilon} * \Pr[\restrict{\m(\vec{i}_2)}{E} \supseqeq S] \right)\\
&= \e{n+1 * \epsilon} * \Pr[\restrict{\m(\vec{i}_2)}{E} \supseqeq S]
\end{align*}
as needed.
\end{proof}

\appsection{Compositional Reasoning}
\label{app:hidden-transition-replacement}

To prove Theorem~\ref{thm:hidden-transition-replacement}, we use a definition and a proposition that helps us to track when the transition under $h^\dagger$ is being simulated by many transitions of $M_2$.  

Let $L_1 = \langle S_1, Q_1 \uplus D_1, R_1 \uplus H_1, \transone\rangle$ and $L_2 = \langle S_2, \emptyset, H_2, \transtwo\rangle$ and $M_2 = \langle L_2, s^0_2\rangle$.  Let $A_3 = Q_1 \uplus D_1 \uplus R_1 \uplus H_1 \uplus H_2 \uplus \{h^{\ddagger}\}$.  Let $h^{\dagger}$ be a distinguished internal action in $H_1$.  For simplicity, we assume that $h^\dagger$ only labels the one transition of $L_1$ that is implemented by $M_2$.  Let $h^\ddagger$ be a distinguished internal action not in $H_1$ or $H_2$.  

Let $\Psi(\vec{a})$ be a set of action sequences formed by replacing each action $h^{\dagger}$ in $\vec{a}$ with the internal action $h^{\ddagger}$ followed by any sequence $\vec{h}$ from $H_2^{+}$ and then $h^\ddagger$ again.
Formally, 
\begin{align*}
\Psi(h^\dagger\cons\vec{a}) &= h^\ddagger\cons H_2^{+} \cons h^\ddagger\cons \Psi(\vec{a})\\
\Psi(a\cons\vec{a}) &= a\cons \Psi(\vec{a}) &&\text{where $a \neq h^\dagger$}
\end{align*}
where $\cons$ is raised to work over sets in the standard way: for $X \subseteq A^*$ and $Y \subseteq A^*$, $X \cons Y = \set{\vec{a} \in A^*}{\exists \vec{a}_1 \in X, \exists \vec{a}_2 \in Y \st \vec{a} = \vec{a}_1 \cons\vec{a}_2}$ and $a \cons X = \{a\}\cons X$.

\begin{proposition} \label{prp:near-replacement} 
Let $M_1 = \langle L_1, s_0 \rangle$ and let $M_3= M_1[s^\dagger,M_2,\iota] = \langle L_3,s_0 \rangle$ where $M_2$ implements the transition of $h^\dagger$ under $\iota$. For all $\vec{a}$, $\vec{i}$, for all $s \in S$,
\[ \sum_{s' \in S_1} \langle L_1, s \rangle (\vec{i})(\vec{a},s') = \sum_{\vec{a}' \in \Psi(a)}  \sum_{s' \in S_1} \langle L_3, s \rangle (\vec{i})(\vec{a}',s') \]
\end{proposition}

\begin{proof}

We use induction over the structure of $\vec{a}$. 

Case: $\vec{a} = [\,]$.  Since $\Psi([\,]) = \{[\,]\}$, $ \langle L_1, s \rangle(\vec{i})([\,],s') = 
\langle L_3, s \rangle(\vec{i})([\,],s') = 1$. 

Case: $\vec{a} = i\cons\vec{a}'$. 
\begin{itemize}
\item Subcase: there does not exist $\vec{i}'$ such that $\vec{i} = i\cons\vec{i}'$ and $s \transone^i \mu$.
In this subcase, $\langle L_1, s \rangle(\vec{i})(\vec{a},s') = 0$. By definition of $\Psi$ we have that $\forall \vec{a}'' \in \Psi(\vec{a}) =  \Psi(i \cons \vec{a}')$ and $\vec{a}''$ is of the form $i \cons \vec{a}'''$ for some $\vec{a}'''$. Since, by definition of $\transthree$, $M_3$ has an input transition from a state only if it has an input transition from that same state in $M_1$ there does not exist $s \transthree^i \mu$. It follows that for all $\forall \vec{a}'' \in \Psi(\vec{a})$, $\langle L_3, s \rangle(\vec{i})(\vec{a}'',s') = 0$, as needed.

\item Subcase: there does exist $\vec{i}'$ such that $\vec{i} = i\cons\vec{i}'$ and $s \transone^i \mu$.
In this subcase,
\begin{align*}
\langle L_1, s \rangle(\vec{i})(\vec{a},s') &= \sum_{s'' \in S_1} \mu_1(s'') \langle L_1, s''\rangle(\vec{i}')(\vec{a}',s') &&\text{where } s \transone^i \mu_1
\end{align*}
and  
\begin{align*}
\langle L_3, s \rangle(\vec{i})(\vec{a},s') &= \sum_{s'' \in S_1 \uplus S_2} \mu_2(s'') \langle L_3, s''\rangle(\vec{i}')(\vec{a}',s') &&\text{where } s \transthree^i \mu_2.
\end{align*}

Since each $\vec{a}'' \in \Psi(\vec{a})$ is of the form $i \cons \vec{a}'''$ for some $\vec{a}'''$, we need to show that 

\[\sum_{s' \in S_1} \sum_{s'' \in S_1} \mu_1(s'') \langle L_1, s''\rangle(\vec{i}')(\vec{a}',s') = \sum_{i \cons \vec{a}''' \in \Psi(\vec{a})}  \sum_{s' \in S_1} \sum_{s'' \in S_1 \uplus S_2} \mu_2(s'') \langle L_3, s''\rangle(\vec{i}')(\vec{a}''',s').\]
We reason as follows:
\begin{align}
\sum_{i \cons \vec{a}''' \in \Psi(\vec{a})}  \sum_{s' \in S_1} \sum_{s'' \in S_1 \uplus S_2} \mu_2(s'') & \langle L_3, s''\rangle(\vec{i}')(\vec{a}''',s') \\
& = \sum_{s'' \in S_1 \uplus S_2} \mu_2(s'') \sum_{i \cons \vec{a}''' \in \Psi(\vec{a})} \sum_{s' \in S_1} \langle L_3, s''\rangle(\vec{i}')(\vec{a}''',s') \label{ln:distribute}\\
& = \sum_{s'' \in S_1} \mu_2(s'') \sum_{i\cons \vec{a}''' \in \Psi(\vec{a})}  \sum_{s' \in S_1
} \langle L_3, s''\rangle(\vec{i}')(\vec{a}''',s') \label{ln:elims2}\\
& = \sum_{s'' \in S_1} \mu_2(s'') \sum_{\vec{a}''' \in \Psi(\vec{a}')}  \sum_{s' \in S_1} \langle L_3, s''\rangle(\vec{i}')(\vec{a}''',s') \label{ln:formofseq}\\
& = \sum_{s'' \in S_1} \mu_1(s'') \sum_{\vec{a}''' \in \Psi(\vec{a}')}  \sum_{s' \in S_1} \langle L_3, s''\rangle(\vec{i}')(\vec{a}''',s') \label{ln:eqdist}\\
& = \sum_{s' \in S_1} \sum_{s'' \in S_1} \mu_1(s'') \langle L_1, s''\rangle(\vec{i}')(\vec{a}',s') \label{ln:goal}
\end{align}

Line~\ref{ln:distribute} follows from reordering summations and using distributivity of multiplication over summation. Line~\ref{ln:elims2} follows from the fact that  any $s'' \in \supp(\mu_2)$ can not be in $S_2$ since any $s'' \in \supp(\mu_2)$ is reachable via an input action. Line~\ref{ln:formofseq}
follows from, the fact that each $\vec{a}'' \in \Psi(\vec{a})$ is of the form $i\cons \vec{a}'''$ and $\vec{a}''' \in \Psi(\vec{a}')$. Line~\ref{ln:eqdist} follows from definition of $\transthree$.
We conclude in Line~\ref{ln:goal} using the inductive hypothesis $\sum_{s' \in S_1} \langle L_1, s''\rangle(\vec{i}')(\vec{a}',s') = \sum_{\vec{a}''' \in \Psi(\vec{a}')}  \sum_{s' \in S_1} \langle L_3, s''\rangle(\vec{i}')(\vec{a}''',s')$.
\end{itemize}

Case: $\vec{a} = o\cons\vec{a}'$.
\begin{itemize}

\item Subcase: $o \neq h^{\dagger}$. 
\begin{itemize}
\item Subsubcase: there does not exist $\vec{a}'$ such that $\vec{a} = o\cons\vec{a}'$ and $s \transone^a \mu$. In this subcase, $\langle L_1, s \rangle(\vec{i})(\vec{a},s') = 0$. By case definition we know $a = o \cons a'$ where $o \in R_1 \uplus H_1$ and $o \neq h^{\dagger}$. Then, $\forall \vec{a}'' \in \Psi(\vec{a}) =  \Psi(o \cons \vec{a}')$, $\vec{a}''$ is of the form $o \cons \vec{a}'''$ for some $\vec{a}'''$.
Since, by definition of $\transthree$, $M_3$ has an output transition on an action from $R_1 \uplus H_1 \setminus \{h^{\dagger}\}$ only if it has the same transition $M_1$. This gives 
$\langle L_3, s \rangle(\vec{i})(\vec{a}'',s') = 0$ for all $\forall \vec{a}'' \in \Psi(\vec{a})$, as needed. 
\item Subsubcase: there does exist $\vec{a}'$ such that $\vec{a} = o \cons\vec{a}'$ and $s \transone^o \mu$. In this subcase, the proof follows a line of reasoning analogous to the case where 
$\vec{a} = i\cons\vec{a}'$. We show that

\[\sum_{s' \in S_1} \sum_{s'' \in S_1} \mu_1(s'') \langle L_1, s''\rangle(\vec{i})(\vec{a}',s') = \sum_{o \cons \vec{a}''' \in \Psi(\vec{a})}  \sum_{s' \in S_1} \sum_{s'' \in S_1 \uplus S_2} \mu_2(s'') \langle L_3, s''\rangle(\vec{i})(\vec{a}''',s').\] 

In the step where we argue that $s'' \notin S_2$, we use the fact that all states in $S_2$ can only result from a transition on $h^{\ddagger}$ or an action from $H_2$, and that $o \in R_1 \uplus H_1 \setminus \{h^{\dagger}\}$, which is disjoint from $\{h^{\ddagger}\} \uplus H_2$. 
\end{itemize}

\item Subcase: $o = h^{\dagger}$. 
\begin{itemize}
\item Subsubcase: there does not exist $\mu$ such that $s \transone^{h^{\dagger}} \mu$,  then $\langle L_1, s\rangle(\vec{i})(\vec{a},s') = 0$ and $s \neq s^{\dagger}$. In this case all $\vec{a}'' \in \Psi(\vec{a})$ start with $h^{\ddagger} \vec{h} h^{\ddagger}$ for some vector $\vec{h}$ of actions from ${H_2}^{+}$. Since $s \neq s^{\dagger}$, according to the definition of $\transthree$ the only way for $s$ to have a transition on $h^{\ddagger}$ is if $s = \iota(s_1)$ for some $s_1 \in \supp(\mu^{\dagger})$,  where $s^{\dagger} \transone^{h^{\dagger}} \mu^{\dagger}$ and that $h^{\ddagger}$ transition is $\iota(s_1) \transthree^{h^{\ddagger}} \mathsf{Dirac}(s_1)$. By definition of $\transthree$, $s_1$ can only transition on actions present in $M_1$, which means it cannot transition on any actions from $H_2$. This makes it impossible for $h^{\ddagger}$ to be followed by a sequence $\vec{h}$ from actions $H_2$, and gives $\langle L_3, s \rangle(\vec{i})(\vec{a}'',s') = 0$ for all $\vec{a}'' \in \Psi(\vec{a})$, as needed.

\item Subsubcase: there does exist $\mu$ such that $s \transone^{h^{\dagger}} \mu$, then by the assumption that $s^{\dagger}$ is the unique state enabling $h^{\dagger}$ we know that $s = s^{\dagger}$ and by transition-determinism $s^{\dagger} \trans^{h^{\dagger}} \mu^{\dagger}$. 
We need to show that

\[\sum_{s' \in S_1} \sum_{s'' \in S_1} \mu^{\dagger}(s'') \langle L_1, s''\rangle(\vec{i})(\vec{a}',s') = \sum_{\vec{a}'' \in \Psi(\vec{a})}  \sum_{s' \in S_1} \langle L_3, s^{\dagger} \rangle(\vec{i})(\vec{a}'',s').\]

We reason as follows:
\begin{align}
\sum_{\vec{a}'' \in \Psi(\vec{a})}  & \sum_{s' \in S_1} \langle L_3, s^{\dagger} \rangle(\vec{i})(\vec{a}'',s')  \label{ln:first}\\
& =  \sum_{h^{\ddagger} \cons \vec{h} \cons h^{\ddagger}\cons \vec{a}''' \in \Psi(\vec{a})}  \sum_{s' \in S_1 } \langle L_3, s^{\dagger} \rangle(\vec{i})(h^{\ddagger} \cons \vec{h} \cons h^{\ddagger}\cons \vec{a}''',s') \\
&  = \sum_{h^{\ddagger}\cons \vec{h} \cons h^{\ddagger} \in \vec{a'''} \in \Psi(\vec{a})}  \sum_{s' \in S_1 } \langle L_3, s^0_2 \rangle(\vec{i})(\vec{h}\cons h^\ddagger\cons\vec{a}''',s') \\
&  = \sum_{h^{\ddagger}\cons \vec{h} \cons h^{\ddagger} \vec{a}''' \in \Psi(\vec{a})}  \sum_{s' \in S_1} \sum_{s''' \in S_1 \uplus S_2} \langle L_3, s^0_2 \rangle([\,])(\vec{h},s''')*\langle L_3, s'''\rangle(\vec{i})(h^{\ddagger}\cons \vec{a}''',s') \label{ln:split}\\
& = \sum_{h^{\ddagger}\cons \vec{h} \cons h^{\ddagger} \vec{a}''' \in \Psi(\vec{a})}  \sum_{s' \in S_1} \sum_{s''' \in \iota(S^{\dagger})} \langle L_3, s^0_2 \rangle([\,])(\vec{h},s''')*\langle L_3, s'''\rangle(\vec{i})(h^{\ddagger}\cons \vec{a}''',s')  \label{ln:inject}\\
&  = \sum_{h^{\ddagger}\cons \vec{h} \cons h^{\ddagger} \vec{a}''' \in \Psi(\vec{a})}  \sum_{s' \in S_1} \sum_{s''' \in \iota(S^{\dagger})} \langle L_3, s^0_2 \rangle([\,])(\vec{h},s''')* \langle L_3, s^4 \rangle(\vec{i})(\vec{a}''',s') \\
& \; \; \; \; \; \text{where} \; s''' \transthree^{h^{\ddagger}} \mathsf{Dirac}(s^4) \label{ln:dirac}\\
& = \sum_{s''' \in \iota(S^{\dagger})} (\sum_{\vec{h} \in (H_2)^+} \langle L_3, s^{\dagger} \rangle([\,]) (\vec{h},s''')) * \sum_{s' \in S_1} \sum_{\vec{a}''' \in \Psi(\vec{a}')} \langle L_3, s^4\rangle(\vec{i})(\vec{a}''',s') \\
& \; \; \; \; \; \text{where} \;  s''' \transthree^{h^{\ddagger}} \mathsf{Dirac}(s^4) \label{ln:almost} \\ 
&  = \sum_{s''' \in \iota(S^{\dagger})} \mu^{\dagger}(s''')  \sum_{s' \in S_1} \sum_{\vec{a}''' \in \Psi(\vec{a}')} \langle L_3, s^4 \rangle(\vec{i})(\vec{a}''',s') \label{ln:done}
\end{align}

For Lines~\ref{ln:first} to~\ref{ln:split} we observe that 
by definition of $\Psi$ each $\vec{a}'' \in \Psi(\vec{a})$ is of the form $h^{\ddagger}\cons\vec{h}\cons h^{\ddagger}\cons\vec{a}''' $ where $\vec{h}\neq [\,]$,
by definition of $\transthree$, $s^{\dagger} \transthree^{h^{\ddagger}} \mathsf{Dirac}(s^0_2)$ and use Proposition~\ref{prp:mult-apart}. For Line~\ref{ln:inject} let $\iota(S^{\dagger}) = \{s \in L_2 \, | \, s = \iota(s') \, \text{for some} \, s' \in \supp(\mu^{\dagger}) \}$. The equation follows 
since $\langle L_3, s^0_2 \rangle([\,])(\vec{h},s''') =0$ for all $s''' \in S_1$, 
by the assumption that for all $s \in \supp(\mu^\dagger)$, $\mu^\dagger(s) = \sum_{\vec{h} \in (H_2)^+} M_2([\,])(\vec{h},\iota(s))$ and that $\mu$ is a probability distribution. By definition of $\transthree$ we know that $\iota(s_1) \transthree^{h^{\ddagger}} \mathsf{Dirac}(s_1)$ for all $s_1 \in \supp(\mu^\dagger)$ and Line~\ref{ln:dirac} follows. 
Note that by definition of $\transthree$ and $\iota(S^{\dagger})$ where since $\iota$ is an injection we know that
for each $s''' \in \iota(S^{\dagger})$ there is a unique state $s^4 \in \supp(\mu^{\dagger})$ such that $s''' = \iota(s^4)$. 
By using distributivity of multiplication over addition and the fact that $\vec{a'''} \in \Psi(\vec{a}')$ we get Line~\ref{ln:almost}. Line~\ref{ln:done} follows from the assumption that for all 
$s \in \supp(\mu^\dagger)$, 
\[ \mu^\dagger(s) = \sum_{\vec{h} \in (H')^+} M'([\,])(\vec{h},\iota(s)). \]
To conclude this case we recall that for each $s''' \in \iota(S^{\dagger})$ there is a unique state $s^4 \in \supp(\mu^{\dagger})$ such that $s''' = \iota(s^4)$ and use the inductive hypothesis.
\end{itemize}
\end{itemize}
\end{proof}

\paragraph{Proof of Theorem~\ref{thm:hidden-transition-replacement}}
Let $M_1 = \langle L_1, s_0 \rangle$ and let $M_3= M_1[s^\dagger,M_2,\iota] = \langle L_3,s_0 \rangle$. We show that for all $\vec{i}$ in $I^*$, and $\vec{e}$ in $E^*$, $s$ in $S_1$, $\Pr[\, \restrict{\langle L_1,s \rangle(\vec{i})}{E} \supseqeq \vec{e}\,] = \Pr[\, \restrict{\langle L_3, s\rangle(\vec{i})}{E} \supseqeq \vec{e}\,]$.
By expanding the definitions of $\Pr[\, \restrict{\langle L_1,s \rangle(\vec{i})}{E} \supseqeq \vec{e}\,]$ and $\Pr[\, \restrict{\langle L_3,s\rangle(\vec{i})}{E} \supseqeq \vec{e}\,]$ we get
\begin{align*}
\Pr[\, \restrict{\langle L_1,s \rangle (\vec{i})}{E} \supseqeq \vec{e}\,] & = \sum_{\vec{a} \in \gamma_1(\vec{e})} \Pr[\, \langle L_1,s\rangle (\vec{i}) \supseqeq \vec{e}] \\
&= \sum_{\vec{a} \in \gamma_1(\vec{e})} \sum_{s' \in S_1} \langle L_1, s \rangle (\vec{i})(\vec{a},s')
\end{align*}
and
\begin{align*}
\Pr[\, \restrict{\langle L_3,s \rangle (\vec{i})}{E} \supseqeq \vec{e}\,] & = \sum_{\vec{a} \in \gamma_3(\vec{e})} \Pr[\,  \langle L_3,s\rangle (\vec{i}) \supseqeq \vec{e}] \\
&=  \sum_{\vec{a} \in \gamma_3(\vec{e})} \sum_{s' \in S_1 \uplus S_2} \langle L_3,s \rangle (\vec{i})(\vec{a},s')
\end{align*} 
where $\gamma_1$ and $\gamma_3$ are sets of sequences of actions, respectively, of $M_1$ and $M_3$ defined as follows: $\gamma_1(\vec{e}) = \set{\vec{a} \in A_1^*}{\restrict{\vec{a}}{E_1} = \vec{e} \land \last(\vec{a}) = \last(\vec{e})}$ with the special case that $\gamma_1([\,]) = \{[\,]\}$, and $\gamma_3(\vec{e}) = \set{\vec{a} \in A_3^*}{\restrict{\vec{a}}{E_3} = \vec{e} \land \last(\vec{a}) = \last(\vec{e})}$ with the special case that $\gamma_3([\,]) = \{[\,]\}$ (as justified at the end of Appendix~\ref{app:autom-details}).  Note that we use $A_i$ for the set of all actions of $M_i$ and $E_i$ for the set of observable actions of $M_i$. 

Now we show that for all $s \in S_1$, $\vec{i}$, $\vec{e}$,  
\begin{align}
\sum_{\vec{a} \in \gamma_{1}(\vec{e})} \sum_{s' \in S_1} \langle L_1, s \rangle (\vec{i})(\vec{a},s') 
&= \sum_{\vec{a} \in \gamma_{1}(\vec{e})} \sum_{\vec{a}' \in \Psi(a)}  \sum_{s' \in S_1} \langle L_3,s \rangle (\vec{i})(\vec{a'},s') \label{ln:use-near-replacement}\\
&= \sum_{\vec{a} \in \gamma_{1}(\vec{e})} \sum_{\vec{a}' \in \Psi(a)}  \sum_{s' \in S_1 \uplus S_2} \langle L_3,s \rangle (\vec{i})(\vec{a}',s') \label{ln:substitute} \\ 
& = \sum_{\vec{a} \in \Phi(\vec{e})} \sum_{s' \in S_1 \uplus S_2}  \langle L_3,s \rangle (\vec{i})(\vec{a},s') \label{ln:exclusive-two} \\
& = \sum_{\vec{a} \in \gamma_3 
(\vec{e})} \sum_{s' \in S_1 \uplus S_2} \langle L_3,s \rangle (\vec{i})(\vec{a},s') \label{ln:similarity-two}
\end{align}
where $\Phi(\vec{e}) =\bigcup_{\vec{a} \in \gamma_1(\vec{e})} \Psi(\vec{a})$, and $\Psi$ is as defined at the top of this section.

Line~\ref{ln:use-near-replacement} follows from Proposition~\ref{prp:near-replacement}.

For Line~\ref{ln:substitute}, we argue as follows: Since $M_2$ has no external actions and all 
transitions of $M_3$ on external actions end in a state in $S_1 \setminus S_2$,
for those states $s' \in S_2$, $s'$ is reachable via a hidden action only.  Thus, for any $\vec{i}$, any $\vec{a}  \in \gamma_{3}(\vec{e})$, $\langle L_3, s \rangle (\vec{i})(\vec{a},s') = 0$ since 
$\vec{a}$ ends in an observable action from $E$ by definition of $\gamma_3$.

Line~\ref{ln:exclusive-two} follows from the definition of $\Phi$ and the fact that for any pair of sequences $\vec{a_1}, \vec{a_2}$ such that $\vec{a_1} \neq \vec{a_2}$, $\Psi(\vec{a_1}) \cap \Psi(\vec{a_2}) = \emptyset$.

For Line~\ref{ln:similarity-two} we observe that any sequence $\vec{a} \in \gamma_3(\vec{e}) \setminus \Phi(\vec{e})$ must have an occurrence of the action $h^{\ddagger}$ that is neither immediately preceded by a subsequence of the form $h^{\ddagger}\cons \vec{h}$ or immediately followed by a subsequence of the form $\vec{h}\cons h^{\ddagger}$. Then, by definition of   $\transthree$, $\langle L_3,s \rangle (\vec{i})(\vec{a},s')=0$ for all sequences $\vec{a} \in \gamma_3(\vec{e}) \setminus \Phi(\vec{e})$, giving the needed equation.

\appsection{Proof of Soundness of Unwinding}
\label{app:unwinding-soundness}

\appsubsection{A Helpful Proposition}

\begin{proposition}\label{prp:nuw-help}
If $\beta$ is a bijection from $\supp(\nu_1)$ to $\supp(\nu_2)$ and for all $x'_1\in \supp(\nu_1)$, $|\ln \nu_1(x'_1) - \ln \nu_2(\beta(x'_1))| \leq \delta$, then
\begin{multline*} 
\sum_{x'_1 \in \supp(\nu_1)} \nu_1(x'_1) \e{\epsilon'-\delta} \Pr[\, \restrict{\model{\langle L, \beta(x'_1) \rangle}(\vec{i})}{E} {\supseqeq} \vec{e}'\,]\\ \leq \e{\epsilon'} \sum_{x'_2 \in S_\bot} \nu_2(x'_2)  \Pr[\, \restrict{\model{\langle L, x'_2 \rangle}(\vec{i})}{E} {\supseqeq} \vec{e}'\,]
\end{multline*}
\end{proposition}
\begin{proof}
For all $x'_1$ in $\supp(\nu_1)$,
\begin{align}
\nu_1(x'_1) \e{\epsilon'-\delta}
&= \frac{\nu_1(x'_1)}{\nu_2(\beta(x'_1))}\nu_2(\beta(x'_1)) \e{\epsilon'-\delta}\\
&= \e{\ln(\nu_1(x'_1)) - \ln(\nu_2(\beta(x'_1)))} \nu_2(\beta(x'_1)) \e{\epsilon'-\delta}\label{ln:nuw-help-eln}\\
&= \e{\epsilon' - \delta + \ln(\nu_1(x'_1)) - \ln(\nu_2(\beta(x'_1)))} \nu_2(\beta(x'_1))\\
&\leq \e{\epsilon'} \nu_2(\beta(x'_1))\label{ln:nuw-help-eps}
\end{align}
Line~\ref{ln:nuw-help-eln} follows from the fact that for every $x'_1 \in \supp(\nu_1)$, 
\[ \nu_1(x'_1)/\nu_2(\beta(x'_1)) = \e{\ln(\nu_1(x'_1))}/\e{\ln(\nu_2(\beta(x'_1)))} \]
Line~\ref{ln:nuw-help-eps} follows since for every $x'_1 \in \supp(\nu_1)$, $|\ln \nu_1(x'_1) - \ln \nu_2(\beta(x'_1))| \leq \delta$ implies that $\ln \nu_1(x'_1) - \ln \nu_2(\beta(x'_1)) \leq \delta$.

Thus,
\begin{align}
\sum_{x'_1 \in \supp(\nu_1)} &\nu_1(x'_1) \e{\epsilon'-\delta} \Pr[\, \restrict{\model{\langle L, \beta(x'_1) \rangle}(\vec{i})}{E} {\supseqeq} \vec{e}'\,]\\
&\leq \sum_{x'_1 \in \supp(\nu_1)} \e{\epsilon'} \nu_2(\beta(x'_1)) \Pr[\, \restrict{\model{\langle L, \beta(x'_1) \rangle}(\vec{i})}{E} {\supseqeq} \vec{e}'\,]\\
&= \sum_{x'_2 \in \supp(\nu_2)}  \e{\epsilon'} \nu_2(x'_2)  \Pr[\, \restrict{\model{\langle L, x'_2 \rangle}(\vec{i})}{E} {\supseqeq} \vec{e}'\,]\label{ln:nuw-help-biject}\\
&= \e{\epsilon'} \sum_{x'_2 \in S_\bot} \nu_2(x'_2)  \Pr[\, \restrict{\model{\langle L, x'_2 \rangle}(\vec{i})}{E} {\supseqeq} \vec{e}'\,]
\end{align}
Line~\ref{ln:nuw-help-biject} follows from the fact that $\beta$ is a bijection from $\supp(\nu_1)$ to $\supp(\nu_2)$.
\end{proof}

\appsubsection{Proof of Lemma~\ref{thm:near-unwinding}}
 
Below we prove that 
$\Pr[\,\restrict{\model{\langle L, x_1\rangle}(\vec{i})}{E} {\supseqeq} \vec{e}\,]  \leq \e{\epsilon} \Pr[\,\restrict{\model{\langle L, x_2\rangle}(\vec{i})}{E} {\supseqeq} \vec{e}\,]$.
Proving the reverse that $\Pr[\,\restrict{\model{\langle L, x_2\rangle}(\vec{i})}{E} {\supseqeq} \vec{e}\,]  \leq \e{\epsilon} \Pr[\,\restrict{\model{\langle L, x_1\rangle}(\vec{i})}{E} {\supseqeq} \vec{e}\,]$ is much the same reversing the roles of $x_1$ and $x_2$ and using $\beta^{-1}$ in the place of $\beta$.

Proof by induction over the structures of $\vec{e}$ and $\vec{i}$.

Case: $\vec{e} = [\,]$.  In this case,
\[ \Pr[\, \restrict{\model{\langle L, x_1\rangle}(\vec{i})}{E} \supseqeq [\,]\,] = 1 \leq \e{\epsilon}* 1 = \e{\epsilon} \Pr[\, \restrict{\model{\langle L, x_2\rangle}(\vec{i})}{E} \supseqeq [\,]\,] \]

Case: $x_1$ has no outgoing transitions and $\vec{e} \neq [\,]$.  
In this case, $\Pr[\,\restrict{\model{\langle L, x_1\rangle}(\vec{i})}{E} {\supseqeq} \vec{e}\,] = 0 \leq \e{\epsilon}  \Pr[\,\restrict{\model{\langle L, x_2\rangle}(\vec{i})}{E} {\supseqeq} \vec{e}\,]$.

Henceforth, we only consider $x_1$ with at least one out going transition.  Since $x_1$ is related to $x_2$, we know it must also have at least one out going transition.  Thus, neither $x_1$ nor $x_2$ can be $\bot$.  Thus, we use $s_1$ for $x_1$ and $s_2$ for $x_2$ for the reminder of the proof.

Case: $\vec{e} = q\cons\vec{e}'$ and $\vec{i} = [\,]$ for some $q \in Q$.  
In this case, 
\[ \Pr[\,\restrict{\model{\langle L, s_1\rangle}(\vec{i})}{E} {\supseqeq} q\cons\vec{e}'\,] 
= \sum_{\vec{a} \in \gamma(q\cons\vec{e}')} \Pr[\,\model{\langle L, s_1\rangle}(\vec{i}) {\supseqeq} \vec{a}\,] 
= \sum_{\vec{a} \in \gamma(q\cons\vec{e}')}  \sum_{s'_1 \in S} \langle L, s_1\rangle(\vec{i})(\vec{a},s'_1) \]
Since $\vec{e} \neq [\,]$, $[\,]$ is not in $\gamma(\vec{e})$.  Furthermore, all $\vec{a}$ in $\gamma(\vec{e})$ must have $q$ come before any other action of $E$.  In particular, $\vec{a}$ must have either the form $q\cons\vec{a}'$, $d\cons\vec{a}'$, or $h\cons\vec{a}'$ for some $\vec{a} \in A^*$, $d \in D$, and $h \in H$.  Since $s_1$ is $H$-disabled by being in the unwinding relation, we know that for no $h \in H$ and $\mu$ does $s_1 \trans^h \mu$.  These factors combine to mean that $\langle L, s_1\rangle([\,])(\vec{a},s'_1) = 0$ for all $s'_1 \in S$ and $\vec{a} \in \gamma(\vec{e})$.  Thus, $\Pr[\,\restrict{\model{\langle L, s_1\rangle}([\,])}{E} {\supseqeq} r\cons\vec{e}'\,] = 0$.
The same reasoning concludes that $\Pr[\,\restrict{\model{\langle L, s_2\rangle}([\,])}{E} {\supseqeq} q\cons\vec{e}'\,] = 0$ making $\Pr[\,\restrict{\model{\langle L, s_1\rangle}([\,])}{E} {\supseqeq} q\cons\vec{e}'\,] \leq \e{\epsilon} \Pr[\,\restrict{\model{\langle L, s_2\rangle}([\,])}{E} {\supseqeq} q\cons\vec{e}'\,]$ since $0 \leq \e{\epsilon} 0$

Case: $\vec{e} = r\cons\vec{e}'$ and $\vec{i} = [\,]$ for some $r \in R$.  
We consider the following subcases:
\begin{itemize}

\item Subcase: $s_1 \trans^r \mu$ for some $\mu$.
Since $s_1$ and $s_2$ are related, there exists $\mu_2$ such that $s_2 \trans^{r} \mu_2$.  
This implies that there exists $\nu_1$ and $\nu_2$ such that $s_1 \wtrans^r \nu_1$ and $s_2 \wtrans^r \nu_2$.  Since $s_1 \mathrel{\mathcal{R}^{\epsilon}} s_2$, there exists $\delta$ in $[0,\epsilon]$ 
such that $\nu_1 \mathrel{\mathcal{L}(\mathcal{R}^{\epsilon-\delta},\delta)} \nu_2$.  This implies there exists a bijection $\beta$ from $\supp(\nu_1)$ to $\supp(\nu_2)$ such that for all $x_1 \in \supp(\nu_1)$, $x_1 \mathcal{R}^{\epsilon-\delta} \beta(x_2)$ and $|\ln \nu_1(x_1) - \ln \nu_2(\beta(x_1))| \leq \delta$.  
Thus, we may apply the inductive hypothesis to $\vec{i}$ and $\vec{e}'$ to get for all $x_1$ in $\supp(\nu_1)$, $\Pr[\,\restrict{\model{\langle L, x'_1\rangle}(\vec{i})}{E} {\supseqeq} \vec{e}'\,] \leq \e{\epsilon-\delta} \Pr[\,\restrict{\model{\langle L, \beta(x'_1)\rangle}(\vec{i})}{E} {\supseqeq} \vec{e}'\,]$.
Thus,
\begin{align}
\Pr[\,&\restrict{\model{\langle L, s_1\rangle}(\vec{i})}{E} {\supseqeq} r\cons\vec{e}'\,] \\
&= \sum_{x'_1 \in S_\bot} \nu_1(x'_1) \Pr[\,\restrict{\model{\langle L, x'_1\rangle}(\vec{i})}{E} {\supseqeq} \vec{e}'\,] \label{ln:nuw-r-nu-one}\\
&= \sum_{x'_1 \in \supp(\nu_1)} \nu_1(x'_1) \Pr[\,\restrict{\model{\langle L, x'_1\rangle}(\vec{i})}{E} {\supseqeq} \vec{e}'\,] \label{ln:nuw-r-nu-drop}\\
&\leq \sum_{x'_1 \in \supp(\nu_1)} \nu_1(x'_1) \e{\epsilon-\delta} \Pr[\, \restrict{\model{\langle L, \beta(x'_1) \rangle}(\vec{i})}{E} {\supseqeq} \vec{e}'\,] \label{ln:nuw-r-ih}\\
&\leq \e{\epsilon} \sum_{x'_2 \in S_\bot} \nu_2(x'_2)  \Pr[\, \restrict{\model{\langle L, x'_2 \rangle}(\vec{i})}{E} {\supseqeq} \vec{e}'\,]\label{ln:nuw-r-biject}\\
&= \e{\epsilon} \Pr[\,\restrict{\model{\langle L, s_2\rangle}(\vec{i})}{E} {\supseqeq} r\cons\vec{e}'\,] \label{ln:nuw-r-nu-two}
\end{align}
Lines~\ref{ln:nuw-r-nu-one} and~\ref{ln:nuw-r-nu-two} follow from Proposition~\ref{thm:nu}.
Line~\ref{ln:nuw-r-ih} follows from the inductive hypothesis 
Line~\ref{ln:nuw-r-biject} follows Proposition~\ref{prp:nuw-help}.

\item Subcase: $s_1 \trans^{r'} \mu$ for some output $r' \neq r$.  
Since $s_1$ and $s_2$ are related, there exists $\mu_2$ such that $s_2 \trans^{r'} \mu_2$.  
Furthermore, for no other action $a \neq r'$ does does $s_1 \trans^a \mu'$ or $s_2 \trans^a \mu'$ for any $\mu'$.
Recall that $\Pr[\,\restrict{\model{\langle L, s_1\rangle}(\vec{i})}{E} {\supseqeq} r\cons\vec{e}'\,] 
= \sum_{\vec{a} \in \gamma(r\cons\vec{e}')} \Pr[\,\model{\langle L, s_1\rangle}(\vec{i}) {\supseqeq} \vec{a}\,] 
= \sum_{\vec{a} \in \gamma(r\cons\vec{e}')}  \sum_{s'_1 \in S} \langle L, s_1\rangle(\vec{i})(\vec{a},s'_1)$.  
For all $\vec{a} \in \gamma(r\cons\vec{e}')$, its first element from $E$ must be $r$ and, thus, it cannot start with $r'$.  However, $s$ can only transition under $r'$ and $\vec{a} \neq [\,]$, meaning there must be a transition for $\vec{a}$ to be produced.  Thus, for all such $\vec{a}$ and $s'_1$, $\langle L, s_1\rangle(\vec{i})(\vec{a},s'_1) = 0$ and $\Pr[\,\restrict{\model{\langle L, s_1\rangle}(\vec{i})}{E} {\supseqeq} r\cons\vec{e}'\,] = 0$.  Similar reasoning concludes that $\Pr[\,\restrict{\model{\langle L, s_2\rangle}(\vec{i})}{E} {\supseqeq} r\cons\vec{e}'\,] = 0$.  Thus,  
 $\Pr[\,\restrict{\model{\langle L, s_1\rangle}(\vec{i})}{E} {\supseqeq} r\cons\vec{e}'\,] = 0 \leq \e{\epsilon} * 0 = \e{\epsilon} \Pr[\,\restrict{\model{\langle L, s_2\rangle}(\vec{i})}{E} {\supseqeq} r\cons\vec{e}'\,]$ as needed.

\item Subcase: $s_1$ is an input accepting state and $\vec{i} = [\,]$.
Recall that 
\[ \Pr[\,\restrict{\model{\langle L, s_1\rangle}(\vec{i})}{E} {\supseqeq} r\cons\vec{e}'\,] = \sum_{\vec{a} \in \gamma(r\cons\vec{e}')}  \sum_{s'_1 \in S} \langle L, s_1\rangle(\vec{i})(\vec{a},s'_1) \]
  Since $\vec{a}$ cannot be $[\,]$, $\vec{i} = [\,]$, and $s_1$ is an input accepting state, this means that $\langle L, s_1\rangle(\vec{i})(\vec{a},s'_1) = 0$ for all such $\vec{a}$ and $s'_1$.  Thus, $\Pr[\,\restrict{\model{\langle L, s_1\rangle}(\vec{i})}{E} {\supseqeq} r\cons\vec{e}'\,] = 0$.

Since $s_1$ is input accepting and related to $s_2$, $s_2$ must also be input accepting.  Thus, by similar reasoning $\Pr[\,\restrict{\model{\langle L, s_1\rangle}(\vec{i})}{E} {\supseqeq} r\cons\vec{e}'\,] = 0$ and the results holds as above.

\item Subcase: $s_1$ is an input accepting state and $\vec{i} = q\cons\vec{i}$ for some $q \in Q$.
Since $\vec{e} = r\cons\vec{e}'$, no $\vec{a} \in \gamma(\vec{e})$ can have $q$ come before $r$.  Thus, much as above $\Pr[\,\restrict{\model{\langle L, s_1\rangle}(\vec{i})}{E} {\supseqeq} r\cons\vec{e}'\,] = 0 = \Pr[\,\restrict{\model{\langle L, s_1\rangle}(\vec{i})}{E} {\supseqeq} r\cons\vec{e}'\,]$.

\item Subcase: $s_1$ is an input accepting state and $\vec{i} = d\cons\vec{i}'$ for some $d \in D$.
Since $s_1$ is input accepting and related to $s_2$, $s_2$ must also be input accepting.  Thus, there exist $\nu_1$ and $\nu_2$ such that $s_1 \wtrans^d \nu_1$ and $s_2 \wtrans^d \nu_2$.  
Since $s_1 \mathrel{\mathcal{R}^{\epsilon}} s_2$, there exists $\delta$ in $[0,\epsilon]$ 
such that $\nu_1 \mathrel{\mathcal{L}(\mathcal{R}^{\epsilon-\delta},\delta)} \nu_2$.  This implies there exists a bijection $\beta$ from $\supp(\nu_1)$ to $\supp(\nu_2)$ such that for all $x_1 \in \supp(\nu_1)$, $x_1 \mathcal{R}^{\epsilon-\delta} \beta(x_2)$ and $|\ln \nu_1(x_1) - \ln \nu_2(\beta(x_1))| \leq \delta$.  
Thus, we may apply the inductive hypothesis to $\vec{i}'$ and $r\cons\vec{e}'$ to get for all $x_1$ in $\supp(\nu_1)$, $\Pr[\,\restrict{\model{\langle L, x'_1\rangle}(\vec{i}')}{E} {\supseqeq} r\cons\vec{e}'\,] \leq \e{\epsilon-\delta} \Pr[\,\restrict{\model{\langle L, \beta(x'_1)\rangle}(\vec{i}')}{E} {\supseqeq} r\cons\vec{e}'\,]$.
Thus,
\begin{align}
\Pr[\,&\restrict{\model{\langle L, s_1\rangle}(d\cons\vec{i}')}{E} {\supseqeq} r\cons\vec{e}'\,]\\
&= \sum_{x'_1 \in S_\bot} \nu_1(x'_1) \Pr[\,\restrict{\model{\langle L, x'_1\rangle}(\vec{i}')}{E} {\supseqeq} r\cons\vec{e}'\,] \label{ln:nuw-rd-nu-one}\\
&= \sum_{x'_1 \in \supp(\nu_1)} \nu_1(x'_1) \Pr[\,\restrict{\model{\langle L, x'_1\rangle}(\vec{i}')}{E} {\supseqeq} r\cons\vec{e}'\,] \label{ln:nuw-rd-nu-drop}\\
&\leq \sum_{x'_1 \in \supp(\nu_1)} \nu_1(x'_1) \e{\epsilon-\delta} \Pr[\, \restrict{\model{\langle L, \beta(x'_1) \rangle}(\vec{i}')}{E} {\supseqeq} r\cons\vec{e}'\,] \label{ln:nuw-rd-ih}\\
&\leq \e{\epsilon} \sum_{x'_2 \in S_\bot} \nu_2(x'_2)  \Pr[\, \restrict{\model{\langle L, x'_2 \rangle}(\vec{i}')}{E} {\supseqeq} r\cons\vec{e}'\,]\label{ln:nuw-rd-biject}\\
&= \e{\epsilon} \Pr[\,\restrict{\model{\langle L, s_2\rangle}(d\cons\vec{i})}{E} {\supseqeq} r\cons\vec{e}'\,] \label{ln:nuw-rd-nu-two}
\end{align}
Lines~\ref{ln:nuw-rd-nu-one} and~\ref{ln:nuw-rd-nu-two} follow from Proposition~\ref{thm:nu}.
Line~\ref{ln:nuw-rd-ih} follows from the inductive hypothesis 
Line~\ref{ln:nuw-rd-biject} follows Proposition~\ref{prp:nuw-help}.

\end{itemize}

Case: $\vec{e} = q\cons\vec{e}'$ and $\vec{i} = i\cons\vec{i}'$ for some $i$ in $I$ and $\vec{i}'$ in $I^*$.  We consider the following subcases:
\begin{itemize}

\item Subcase: $s_1$ is not an input accepting state: there exists no $\mu_1$ such that $s_1 \trans^q \mu_1$. 
Since $s_1$ and $s_2$ are related, there also cannot exist a $\mu_2$ such that $s_2 \trans^q \mu_2$.  
Since $s_1$ does have a transition and is $H$-disabled, there must exist some response $r$ such that $s_1 \trans^r \mu'_1$ and $s_2 \trans^r \mu'_2$ for some $\mu'_1$ and $\mu'_2$.  Furthermore, $s_1$ and $s_2$ transitions under no other actions.
Recall that 
\begin{align*}
\Pr[\,\restrict{\model{\langle L, s_1\rangle}(i\cons\vec{i}')}{E} {\supseqeq} q\cons\vec{e}'\,] 
&= \sum_{\vec{a} \in \gamma(q\cons\vec{e}')} \Pr[\,\model{\langle L, s_1\rangle}(i\cons\vec{i}') {\supseqeq} \vec{a}\,] \\
&= \sum_{\vec{a} \in \gamma(q\cons\vec{e}')}  \sum_{s'_1 \in S} \langle L, s_1\rangle(i\cons\vec{i}')(\vec{a},s'_1)
\end{align*}  
For all $\vec{a} \in \gamma(q\cons\vec{e}')$, its first element from $E$ must be $q$ and, thus, it cannot start with $r$.  However, $s$ can only transition under $r$ and $\vec{a} \neq [\,]$, meaning there must be a transition for $\vec{a}$ to be produced.  Thus, for all such $\vec{a}$ and $s'_1$, $\langle L, s_1\rangle(\vec{i})(\vec{a},s'_1) = 0$ and $\Pr[\,\restrict{\model{\langle L, s_1\rangle}(i\cons\vec{i}')}{E} {\supseqeq} q\cons\vec{e}'\,] = 0$.  Similar reasoning concludes that $\Pr[\,\restrict{\model{\langle L, s_2\rangle}(i\cons\vec{i}')}{E} {\supseqeq} q\cons\vec{e}'\,] = 0$.  Thus,  
\[ \Pr[\,\restrict{\model{\langle L, s_1\rangle}(i\cons\vec{i}')}{E} {\supseqeq} q\cons\vec{e}'\,] \leq \e{\epsilon} * 0 = \e{\epsilon} \Pr[\,\restrict{\model{\langle L, s_2\rangle}(q\cons\vec{i}')}{E} {\supseqeq} q\cons\vec{e}'\,] \]
 as needed.

\item Subcase: $s_1$ is an input accepting state and $i = q$ for some $\mu_1$.
Since $s_1$ is input accepting, $s_1 \trans^q \mu_1$ for some $\mu_1$.
Since $s_1$ and $s_2$ are related, there exists $\mu_2$ such that $s_2 \trans^{q} \mu_2$.  
This implies that there exists $\nu_1$ and $\nu_2$ such that $s_1 \wtrans^q \nu_1$ and $s_2 \wtrans^q \nu_2$.  Since $s_1 \mathrel{\mathcal{R}^{\epsilon}} s_2$, there exists $\delta$ in $[0,\epsilon]$ 
such that $\nu_1 \mathrel{\mathcal{L}(\mathcal{R}^{\epsilon-\delta},\delta)} \nu_2$.  This implies there exists a bijection $\beta$ from $\supp(\nu_1)$ to $\supp(\nu_2)$ such that for all $x_1 \in \supp(\nu_1)$, $x_1 \mathcal{R}^{\epsilon-\delta} \beta(x_2)$ and $|\ln \nu_1(x_1) - \ln \nu_2(\beta(x_1))| \leq \delta$.  
Thus, we may apply the inductive hypothesis to $\vec{i}'$ and $\vec{e}'$ to get for all $x_1$ in $\supp(\nu_1)$, $\Pr[\,\restrict{\model{\langle L, x'_1\rangle}(\vec{i}')}{E} {\supseqeq} \vec{e}'\,] \leq \e{\epsilon-\delta} \Pr[\,\restrict{\model{\langle L, \beta(x'_1)\rangle}(\vec{i}')}{E} {\supseqeq} \vec{e}'\,]$.
Thus,
\begin{align}
\Pr[\,&\restrict{\model{\langle L, s_1\rangle}(q\cons\vec{i}')}{E} {\supseqeq} q\cons\vec{e}'\,]\\
&= \sum_{x'_1 \in S_\bot} \nu_1(x'_1) \Pr[\,\restrict{\model{\langle L, x'_1\rangle}(\vec{i}')}{E} {\supseqeq} \vec{e}'\,] \label{ln:nuw-q-nu-one}\\
&= \sum_{x'_1 \in \supp(\nu_1)} \nu_1(x'_1) \Pr[\,\restrict{\model{\langle L, x'_1\rangle}(\vec{i}')}{E} {\supseqeq} \vec{e}'\,] \label{ln:nuw-q-nu-drop}\\
&\leq \sum_{x'_1 \in \supp(\nu_1)} \nu_1(x'_1) \e{\epsilon-\delta} \Pr[\, \restrict{\model{\langle L, \beta(x'_1) \rangle}(\vec{i}')}{E} {\supseqeq} \vec{e}'\,] \label{ln:nuw-q-ih}\\
&\leq \e{\epsilon} \sum_{x'_2 \in S_\bot} \nu_2(x'_2)  \Pr[\, \restrict{\model{\langle L, x'_2 \rangle}(\vec{i}')}{E} {\supseqeq} \vec{e}'\,]\label{ln:nuw-q-biject}\\
&= \e{\epsilon} \Pr[\,\restrict{\model{\langle L, s_2\rangle}(q\cons\vec{i}')}{E} {\supseqeq} q\cons\vec{e}'\,] \label{ln:nuw-q-nu-two}
\end{align}
Lines~\ref{ln:nuw-q-nu-one} and~\ref{ln:nuw-q-nu-two} follow from Proposition~\ref{thm:nu}.
Line~\ref{ln:nuw-q-ih} follows from the inductive hypothesis 
Line~\ref{ln:nuw-q-biject} follows Proposition~\ref{prp:nuw-help}.

\item Subcase: $s_1$ is input accepting, $i \neq q$, and $i \in Q$. 
Recall that 
\begin{align*}
\Pr[\,\restrict{\model{\langle L, s_1\rangle}(i\cons\vec{i}')}{E} {\supseqeq} q\cons\vec{e}'\,] 
&= \sum_{\vec{a} \in \gamma(q\cons\vec{e}')} \Pr[\,\model{\langle L, s_1\rangle}(i\cons\vec{i}') {\supseqeq} \vec{a}\,]\\
&= \sum_{\vec{a} \in \gamma(q\cons\vec{e}')}  \sum_{s'_1 \in S} \langle L, s_1\rangle(i\cons\vec{i}')(\vec{a},s'_1)
\end{align*}
For all $\vec{a} \in \gamma(q\cons\vec{e}')$, its first element from $E$ must be $q$ and, thus, it cannot start with $i$.  Thus, for all such $\vec{a}$ and $s'_1$, $\langle L, s_1\rangle(i\cons\vec{i}')(\vec{a},s'_1) = 0$ and $\Pr[\,\restrict{\model{\langle L, s_1\rangle}(i\cons\vec{i}')}{E} {\supseqeq} q\cons\vec{e}'\,] = 0$.  Similar reasoning allows us to conclude that $\Pr[\,\restrict{\model{\langle L, s_2\rangle}(i\cons\vec{i}')}{E} {\supseqeq} q\cons\vec{e}'\,] = 0$.  Thus,  
 $\Pr[\,\restrict{\model{\langle L, s_1\rangle}(i\cons\vec{i}')}{E} {\supseqeq} q\cons\vec{e}'\,] = 0 \leq \e{\epsilon} * 0 = \e{\epsilon} \Pr[\,\restrict{\model{\langle L, s_2\rangle}(q\cons\vec{i}')}{E} {\supseqeq} q\cons\vec{e}'\,]$ as needed.

\item Subcase: $s_1$ is input accepting, $i \neq q$, and $i \in D$. 
We use $d$ to denote $i$.
Since $s_1$ is input accepting and related to $s_2$, $s_2$ must also be input accepting.  Thus, there exist $\nu_1$ and $\nu_2$ such that $s_1 \wtrans^d \nu_1$ and $s_2 \wtrans^d \nu_2$.  
Since $s_1 \mathrel{\mathcal{R}^{\epsilon}} s_2$, there exists $\delta$ in $[0,\epsilon]$ 
such that $\nu_1 \mathrel{\mathcal{L}(\mathcal{R}^{\epsilon-\delta},\delta)} \nu_2$.  This implies there exists a bijection $\beta$ from $\supp(\nu_1)$ to $\supp(\nu_2)$ such that for all $x_1 \in \supp(\nu_1)$, $x_1 \mathcal{R}^{\epsilon-\delta} \beta(x_2)$ and $|\ln \nu_1(x_1) - \ln \nu_2(\beta(x_1))| \leq \delta$.  
Thus, we may apply the inductive hypothesis to $\vec{i}'$ and $q\cons\vec{e}'$ to get for all $x_1$ in $\supp(\nu_1)$, $\Pr[\,\restrict{\model{\langle L, x'_1\rangle}(\vec{i}')}{E} {\supseqeq} q\cons\vec{e}'\,] \leq \e{\epsilon-\delta} \Pr[\,\restrict{\model{\langle L, \beta(x'_1)\rangle}(\vec{i}')}{E} {\supseqeq} q\cons\vec{e}'\,]$.
Thus,
\begin{align}
\Pr[\,&\restrict{\model{\langle L, s_1\rangle}(d\cons\vec{i}')}{E} {\supseqeq} q\cons\vec{e}'\,]\\
&= \sum_{x'_1 \in S_\bot} \nu_1(x'_1) \Pr[\,\restrict{\model{\langle L, x'_1\rangle}(\vec{i}')}{E} {\supseqeq} q\cons\vec{e}'\,] \label{ln:nuw-qd-nu-one}\\
&= \sum_{x'_1 \in \supp(\nu_1)} \nu_1(x'_1) \Pr[\,\restrict{\model{\langle L, x'_1\rangle}(\vec{i}')}{E} {\supseqeq} q\cons\vec{e}'\,] \label{ln:nuw-qd-nu-drop}\\
&\leq \sum_{x'_1 \in \supp(\nu_1)} \nu_1(x'_1) \e{\epsilon-\delta} \Pr[\, \restrict{\model{\langle L, \beta(x'_1) \rangle}(\vec{i}')}{E} {\supseqeq} q\cons\vec{e}'\,] \label{ln:nuw-qd-ih}\\
&\leq \e{\epsilon} \sum_{x'_2 \in S_\bot} \nu_2(x'_2)  \Pr[\, \restrict{\model{\langle L, x'_2 \rangle}(\vec{i}')}{E} {\supseqeq} q\cons\vec{e}'\,]\label{ln:nuw-qd-biject}\\
&= \e{\epsilon} \Pr[\,\restrict{\model{\langle L, s_2\rangle}(d\cons\vec{i})}{E} {\supseqeq} q\cons\vec{e}'\,] \label{ln:nuw-qd-nu-two}
\end{align}
Lines~\ref{ln:nuw-qd-nu-one} and~\ref{ln:nuw-qd-nu-two} follow from Proposition~\ref{thm:nu}.
Line~\ref{ln:nuw-qd-ih} follows from the inductive hypothesis 
Line~\ref{ln:nuw-qd-biject} follows Proposition~\ref{prp:nuw-help}.

\end{itemize}

\appsubsection{Proof of Theorem~\ref{thm:unwinding}}

We use Lemma~\ref{prp:strong-by-trace} and strengthen the hypothesis to show that for all reachable states $s$ and $\vec{e}$,
\begin{align*}
\Pr[\,\restrict{\model{\langle L, s\rangle}(\vec{i}_1)}{E} {\supseqeq} \vec{e}\,] &= \sum_{\vec{a} \in \gamma(\vec{e})} \Pr[\,\model{\langle L, s\rangle}(\vec{i}_1) {\supseqeq} \vec{a}\,]\\
&\leq \e{\epsilon} \sum_{\vec{a} \in \gamma(\vec{e})} \Pr[\,\model{\langle L, s\rangle}(\vec{i}_2) {\supseqeq} \vec{a}\,] \\
&= \e{\epsilon} \Pr[\,\restrict{\model{\langle L, s\rangle}(\vec{i}_2)}{E} {\supseqeq} \vec{e}\,]
\end{align*}
Arbitrarily fix $\vec{i}_1$ and $\vec{i}_2$ such that $\diff{\vec{i}_1}{\vec{i}_2} = 1$.  
We use induction over the structures of $\vec{i}_1$, $\vec{i}_2$, and $\vec{e}$.

Case: $\vec{e} = [\,]$.  In this case, $\gamma(\vec{e}) = \{[\,]\}$ and $\Pr[\,\restrict{\model{\langle L, s\rangle}(\vec{i}_1)}{E} {\supseqeq} [\,]\,] = 1 \leq \e{\epsilon} * 1 = \e{\epsilon} \Pr[\,\restrict{\model{\langle L, s\rangle}(\vec{i}_2)}{E} {\supseqeq} [\,]\,]$ irrespective of $\vec{i}_1$ and $\vec{i}_2$.  

Only in the case where $\vec{e} = [\,]$, can $[\,]$ be in $\gamma(\vec{e})$.  Thus, we assume that $\vec{e} \neq [\,]$ in the reminder of this proof.

Case: $\vec{i}_1 = [\,]$ and $\vec{i}_2 = [\,]$.  $\Pr[\,\model{\langle L, s\rangle}([\,]) {\supseqeq} \vec{a}\,] = \Pr[\,\model{\langle L, s\rangle}([\,]) {\supseqeq} \vec{a}\,]$ for all $\vec{a} \in \gamma(\vec{e})$ for any $\vec{e}$.

Case: $\vec{i}_1 = d\cons\vec{i}'_1$ and $\vec{i}_2 = d\cons\vec{i}'_2$.  We consider three mutually exclusive subcases:
\begin{itemize}
\item Subcase: $s \trans^d \mu$.
For all $\vec{a}$ such that for no $\vec{a}'$, $\vec{a} = d\cons\vec{a}'$, $\Pr[\,\model{\langle L, s\rangle}(\vec{i}_1) {\supseqeq} \vec{a}\,] = 0 = \Pr[\,\model{\langle L, s\rangle}(\vec{i}_2) {\supseqeq} \vec{a}\,]$.  Since such $\vec{a}$ add nothing to the summations, we may ignore them and limit our attention to $ \vec{a} = d\cons\vec{a}'$ in $\gamma(\vec{e})$.  
Note that all such $\vec{a}'$ are in $\gamma(\vec{e})$ iff $d\cons\vec{a}'$ is in $\gamma(\vec{e})$.

All the states in $\supp(\mu)$ are reachable.  Thus, for each state $s'$ in $\supp(\mu)$, we may apply the inductive hypothesis on $\vec{i}'_1$, $\vec{i}'_2$, and $\vec{e}$ to get that 
\[ \sum_{\vec{a}' \in \gamma(\vec{e})} \Pr[\,\model{\langle L, s'\rangle}(\vec{i}'_1) {\supseqeq} \vec{a}'\,] \leq \e{\epsilon} \sum_{\vec{a}' \in \gamma(\vec{e})} \Pr[\,\model{\langle L, s'\rangle}(\vec{i}'_2) {\supseqeq} \vec{a}'\,]\]
Considering the sum over all such $\vec{a}$, we get
\begin{align}
\sum_{\vec{a} \in \gamma(\vec{e})} \Pr[\,&\model{\langle L, s\rangle}(\vec{i}_1) {\supseqeq} \vec{a}\,]\\
&= \sum_{d\cons\vec{a}' \in \gamma(\vec{e})} \Pr[\,\model{\langle L, s\rangle}(\vec{i}_1) {\supseqeq} d\cons\vec{a}'\,]\\
&= \sum_{\vec{a}' \in \gamma(\vec{e})} \sum_{s' \in S} \mu(s') \Pr[\, \model{\langle L, s'\rangle}(\vec{i}'_1) {\supseqeq} \vec{a}'\,] \label{ln:uw-d-prefix-sem-one}\\
&= \sum_{s' \in S} \mu(s') \sum_{\vec{a}' \in \gamma(\vec{e})} \Pr[\,\model{\langle L, s'\rangle}(\vec{i}'_1) {\supseqeq} \vec{a}'\,]\\
&= \sum_{s' \in \supp(\mu)} \mu(s') \sum_{\vec{a}' \in \gamma(\vec{e})} \Pr[\,\model{\langle L, s'\rangle}(\vec{i}'_1) {\supseqeq} \vec{a}'\,]\\
&\leq \sum_{s' \in \supp(\mu)} \mu(s') \e{\epsilon} \sum_{\vec{a}' \in \gamma(\vec{e})} \Pr[\,\model{\langle L, s'\rangle}(\vec{i}'_2) {\supseqeq} \vec{a}'\,]\label{ln:uw-d-ih}\\
&= \e{\epsilon} \sum_{\vec{a}' \in \gamma(\vec{e})}  \sum_{s' \in S} \mu(s') \Pr[\,\model{\langle L, s'\rangle}(\vec{i}'_2) {\supseqeq} \vec{a}'\,]\\
&= \e{\epsilon} \sum_{d\cons\vec{a}' \in \gamma(\vec{e})} \Pr[\,\model{\langle L, s\rangle}(\vec{i}_2) {\supseqeq} d\cons\vec{a}'\,]\label{ln:uw-d-prefix-sem-two}\\
&= \e{\epsilon} \sum_{\vec{a} \in \gamma(\vec{e})} \Pr[\,\model{\langle L, s\rangle}(\vec{i}_2) {\supseqeq} \vec{a}\,]
\end{align}
where $d\cons\vec{a}'$ in the expression $d\cons\vec{a}' \in \gamma(\vec{e})$ ranges over only those elements of $\gamma(\vec{e})$ of the form $d\cons\vec{a}'$.  That is, $\sum_{d\cons\vec{a}' \in \gamma(\vec{e})}$ is shorthand for \[ \sum_{d\cons\vec{a}' \in \set{\vec{a}'' \in \gamma(\vec{e})}{\exists \vec{a}' \in A^* \st d\cons\vec{a}' = \vec{a}''}}\]  
Note that the last line follows from the fact that  $\model{\langle L, s\rangle}(\vec{i}_2)(\vec{a}) = 0$ for all $\vec{a}$ not of the form $d \cons \vec{a}'$.
Lines~\ref{ln:uw-d-prefix-sem-one} and~\ref{ln:uw-d-prefix-sem-two} follow from Proposition~\ref{thm:prefix-sem}.
Line~\ref{ln:uw-d-ih} follows from the inductive hypothesis.

\item Subcase: $s \trans^r \mu$ for some $r \in R$.
For all $\vec{a}$ such that for no $\vec{a}'$, $\vec{a} = r\cons\vec{a}'$, $\Pr[\,\model{\langle L, s\rangle}(\vec{i}_1) {\supseqeq} \vec{a}\,] = 0 = \Pr[\,\model{\langle L, s\rangle}(\vec{i}_2) {\supseqeq} \vec{a}\,]$.  Since such $\vec{a}$ add nothing to the summations, we may ignore them and limit our attention to $r\cons\vec{a}'$ in $\gamma(\vec{e})$.  Unless $\vec{e} = r\cons\vec{e}'$ for some $\vec{e}'$, no such $r\cons\vec{a}'$ will be in $\gamma(\vec{e})$ and both summations will be zero.  Thus, we limit our attention to the case where $\vec{e} = r\cons\vec{e}'$ for some $\vec{e}'$.  In this case, we may use the inductive hypothesis on $\vec{i}_1$, $\vec{i}_2$, and $\vec{e}'$ to get that for all $s' \in \supp(\mu)$, $\sum_{\vec{a}' \in \gamma(\vec{e}')} \Pr[\,\model{\langle L, s'\rangle}(\vec{i}_1) {\supseqeq} \vec{a}'\,] \leq \e{\epsilon} \sum_{\vec{a}' \in \gamma(\vec{e}')} \Pr[\,\model{\langle L, s'\rangle}(\vec{i}_2) {\supseqeq} \vec{a}'\,]$.  Thus,
\begin{align}
\sum_{\vec{a} \in \gamma(\vec{e})} \Pr[\,&\model{\langle L, s\rangle}(\vec{i}_1) {\supseqeq} \vec{a}\,]\\
&= \sum_{r\cons\vec{a}' \in \gamma(r\cons\vec{e}')} \Pr[\,\model{\langle L, s\rangle}(\vec{i}_1) {\supseqeq} r\cons\vec{a}'\,]\\
&= \sum_{\vec{a}' \in \gamma(\vec{e}')} \sum_{s' \in S} \mu(s') \Pr[\,\model{\langle L, s'\rangle}(\vec{i}_1) {\supseqeq} \vec{a}'\,]\label{ln:uw-r-prefix-sem-one}\\
&= \sum_{s' \in \supp(\mu)} \mu(s') \sum_{\vec{a}' \in \gamma(\vec{e}')} \Pr[\,\model{\langle L, s'\rangle}(\vec{i}_1) {\supseqeq} \vec{a}'\,]\\
&\leq \sum_{s' \in \supp(\mu)} \mu(s') \e{\epsilon} \sum_{\vec{a}' \in \gamma(\vec{e}')} \Pr[\,\model{\langle L, s'\rangle}(\vec{i}_2) {\supseqeq} \vec{a}'\,]\label{ln:uw-r-ih}\\
&= \e{\epsilon} \sum_{\vec{a}' \in \gamma(\vec{e}')}  \sum_{s' \in S} \mu(s') \Pr[\,\model{\langle L, s'\rangle}(\vec{i}_2) {\supseqeq} \vec{a}'\,]\\
&= \e{\epsilon} \sum_{r\cons\vec{a}' \in \gamma(r\cons\vec{e}')} \Pr[\,\model{\langle L, s\rangle}(\vec{i}_2) {\supseqeq} r\cons\vec{a}'\,]\label{ln:uw-r-prefix-sem-two}\\
&= \e{\epsilon} \sum_{\vec{a} \in \gamma(\vec{e})} \Pr[\,\model{\langle L, s\rangle}(\vec{i}_2) {\supseqeq} \vec{a}\,]
\end{align}
Lines~\ref{ln:uw-r-prefix-sem-one} and~\ref{ln:uw-r-prefix-sem-two} follow from Proposition~\ref{thm:prefix-sem}.
Line~\ref{ln:uw-r-ih} follows from the inductive hypothesis.

\item Subcase: Otherwise.  
Since $s$ is $H$-disabled, it is not the case that $s \trans^h \mu$ for any $\mu$ or $h \in H$.
Since $\vec{a} \neq [\,]$, $\Pr[\,\model{\langle L, s\rangle}(\vec{i}_1) {\supseqeq} \vec{a}\,] = 0 = \Pr[\,\model{\langle L, s\rangle}(\vec{i}_2) {\supseqeq} \vec{a}\,]$ for all $\vec{a} \in \gamma(\vec{e})$.
\end{itemize}

Case: $\vec{i}_1 = q\cons\vec{i}'_1$ and $\vec{i}_2 = q\cons\vec{i}'_2$.  Much as above just using that $\vec{a}$ is only in $\gamma(\vec{e})$ if $\vec{e} = q\cons\vec{e}'$ for some $\vec{e}'$ and $\vec{a}' \in \gamma(\vec{e}')$.

Case: $\vec{i}_2 = d\cons\vec{i}_1$.  We consider the following subcases:
\begin{itemize}
\item Subcase: $s \trans^d \mu$.  
Since $s \trans^d \mu$, for some $\nu$, $s \wtrans^d \nu$.
Since $s$ is reachable from $s_0$, there exists an $\epsilon$-unwinding relation $\mathcal{R}^{\epsilon}$ that covers $s$ and $d$.  That is, for all $s' \in \supp(\nu)$, for all $s' \in \supp(\nu)$, $s \mathrel{\mathcal{R}^{\epsilon}} s'$ and $\nu(\bot) = 0$.
\begin{align}
\Pr[\,\restrict{\model{\langle L, s\rangle}}{E}(\vec{i}_1) {\supseqeq} \vec{e}\,] 
&\leq \e{\epsilon} \Pr[\,\restrict{\model{\langle L, s_{\mathsf{min}}\rangle}}{E}(\vec{i}_1) {\supseqeq} \vec{e}\,] \label{ln:uw-min}\\
&= \e{\epsilon} \left(\sum_{x \in S_\bot} \nu(x)\right) \Pr[\,\restrict{\model{\langle L, s_{\mathsf{min}}\rangle}(\vec{i}_1)}{E} {\supseqeq} \vec{e}\,]\label{ln:uw-nu-dist}\\
&= \e{\epsilon} \left(\sum_{s' \in S} \nu(s')\right) \Pr[\,\restrict{\model{\langle L, s_{\mathsf{min}}\rangle}(\vec{i}_1)}{E} {\supseqeq} \vec{e}\,]\label{ln:uw-drop-bot}\\
&\leq \e{\epsilon} \sum_{s' \in S} \nu(s') \Pr[\,\restrict{\model{\langle L, s'\rangle}(\vec{i}_1)}{E} {\supseqeq} \vec{e}\,]\\
&= \e{\epsilon} \sum_{x \in S_\bot} \nu(x) \Pr[\,\restrict{\model{\langle L, x\rangle}(\vec{i}_1)}{E} {\supseqeq} \vec{e}\,]\label{ln:uw-add-bot}\\
&= \e{\epsilon} \Pr[\,\restrict{\model{\langle L, s\rangle}(d\cons\vec{i}_1)}{E} {\supseqeq} \vec{e}\,]\label{ln:uw-nu}\\
&= \e{\epsilon} \Pr[\,\restrict{\model{\langle L, s\rangle}(\vec{i}_2)}{E} {\supseqeq} \vec{e}\,]
\end{align}
where $s_{\mathsf{min}}$ is the state $s' \in \supp(\nu)$ that minimizes $\Pr[\,\restrict{\model{\langle L, s'\rangle}(\vec{i}_1)}{E} {\supseqeq} \vec{e}\,]$.
Line~\ref{ln:uw-min} follows from Lemma~\ref{thm:near-unwinding}.
Line~\ref{ln:uw-nu-dist} follows from Proposition~\ref{prp:nu-dist}.
Lines~\ref{ln:uw-drop-bot} and~\ref{ln:uw-add-bot} follow from $\nu(\bot) = 0$.
Line~\ref{ln:uw-nu} follows from Proposition~\ref{thm:nu}.

\item Subcase: $s \trans^r \mu$ for some $r$.
As in the corresponding subcase in the case for $\vec{i}_1 = d\cons\vec{i}'_1$ and $\vec{i}_2 = d\cons\vec{i}'_2$, we may ignore $\vec{a}$ not of the form $\vec{a} = r\cons\vec{a}'$ and $\vec{e}$ not of the form $r\cons\vec{e}'$.  In this case, we may use the inductive hypothesis on $\vec{i}_1$, $\vec{i}_2$, and $\vec{e}'$ as before to get the required result.

\item Subcase: Otherwise.  Since $s$ does not transition under $d$ in this case and the automaton has quasi-input enabling, it does not transition under any input action.  Further, $s$ is $H$-disabled.  Thus, since $\vec{a} \neq [\,]$, $\Pr[\,\model{\langle L, s\rangle}(\vec{i}_1) {\supseqeq} \vec{a}\,] = 0$ for all $\vec{a} \in \gamma(\vec{e})$.
\end{itemize}

Case: $\vec{i}_1 = d\cons\vec{i}_2$.  We consider the following subcases.
\begin{itemize}
\item Subcase: $s \trans^d \mu$.
Since $s \trans^d \mu$, for some $\nu$, $s \wtrans^d \nu$.
Since $s$ is reachable from $s_0$, there exists an $\epsilon$-unwinding relation $\mathcal{R}^\epsilon$ that covers $s$ and $d$.  That is, for all $s' \in \supp(\nu)$, for all $s' \in \supp(\nu)$, $s \mathrel{\mathcal{R}^{\epsilon}} s'$ and $\nu(\bot) = 0$.

Thus,
\begin{align}
\Pr[\,\restrict{\model{\langle L, s\rangle}(\vec{i}_1)}{E} {\supseqeq} \vec{e}\,]
&= \Pr[\,\restrict{\model{\langle L, s\rangle}(d\cons\vec{i}_2)}{E} {\supseqeq} \vec{e}\,]\\ 
&= \sum_{x \in S_\bot} \nu(x) \Pr[\,\restrict{\model{\langle L, x\rangle}(\vec{i}_2)}{E} {\supseqeq} \vec{e}\,]\label{ln:uw-2-nu}\\ 
&= \sum_{x \in \supp(\nu)} \nu(x) \Pr[\,\restrict{\model{\langle L, x\rangle}(\vec{i}_2)}{E} {\supseqeq} \vec{e}\,]\\ 
&\leq \sum_{x \in \supp(\nu)} \nu(x) \e{\epsilon} \Pr[\,\restrict{\model{\langle L, s\rangle}(\vec{i}_2)}{E} {\supseqeq} \vec{e}\,]\label{ln:uw-2-nuw}\\ 
&= \left( \sum_{x \in \supp(\nu)} \nu(x) \right) \e{\epsilon} \Pr[\,\restrict{\model{\langle L, s\rangle}(\vec{i}_2)}{E} {\supseqeq} \vec{e}\,]\\ 
&= \e{\epsilon} \Pr[\,\restrict{\model{\langle L, s\rangle}(\vec{i}_2)}{E} {\supseqeq} \vec{e}\,]\label{ln:uw-2-nu-dist}
\end{align} 
Line~\ref{ln:uw-2-nu} follows from Proposition~\ref{thm:nu}.
Line~\ref{ln:uw-2-nuw} follows from Lemma~\ref{thm:near-unwinding}.
Line~\ref{ln:uw-2-nu-dist} follows from Proposition~\ref{prp:nu-dist}.

\item Subcase: $s \trans^r \mu$ for some $r$.  As above in the other subcases for $s \trans^r \mu$.

\item Subcase: Otherwise.
In the case where $s \trans^d \mu$ for no $\mu$ and $\vec{a} \neq [\,]$, everything is $0$, which is lower than any possible value of $\e{\epsilon} \Pr[\,\restrict{\model{\langle L, s\rangle}(\vec{i}_2)}{E} {\supseqeq} \vec{e}\,]$.
\end{itemize}

\appsection{Proof of Lemma~\ref{lem:near-system-unwinding}: $\exparamM(K)$ has an Unwinding Family}
\label{app:example-unwinding-proofs}

To prove Lemma~\ref{lem:near-system-unwinding}, arbitrarily fix a state $s$ and data point $d$.  
We use proof by induction over $j$ from $0$ to $t$ to show that for each pair of states $s_1$ and $s_2$ such that $s_1 \mathrel{\mathcal{R}_{s,d}^{2j\epsilon}} s_2$, they have the needed properties.  

In both the base or inductive cases, since $s_1 \mathrel{\mathcal{R}_{s,d}^{2j\epsilon}} s_2$, $s_2$ must have the same value for the PC as $s_1$.  Thus, they have the same set of enabled actions.  That is, there exists a $\mu_1$ such that $s_1 \trans^a \mu_1$ iff there exists a $\mu_2$ such that $s_2 \trans^a \mu_2$.  Thus, $s_1 \wtrans^a \nu_1$ iff $s_2 \wtrans^a \nu_2$.  

\paragraph{Base Case: $j=0$}  
For states with a PC of $08$, the properties follows from the related states being equal.  

For states with a PC of $16$, we can prove the needed properties using $\delta = 0$ as we must since $\mathcal{R}_{s,d}^{0}$ is a $0$-unwinding relation.  Since $j = 0$ and $s_1 \mathrel{\mathcal{R}_{s,d}^{2j\epsilon}} s_2$, $s_1$ must have the form \\ $\langle 16, \langle B'_0, \ldots, B'_{t-1}\rangle, \langle n_0, \ldots, n'_{t-1}\rangle, c', y', r', k'\rangle$.  Since $s_1$ is related to another state, it must be in $S_1^t$.  Thus, $s_1$ is reachable in $t$ queries and $c' = c+(t-1)$.  Once \verb|curSlot| is updated by line 17, it will roll over to the value of $c$.  Thus, $s_1 \wtrans^{r'} \mathsf{Dirac}(s'_1)$ where $s'_1 = \langle 08, \langle B''_0, \ldots, B''_{t-1}\rangle, \langle n''_0, \ldots, n''_{t-1}\rangle, c , y', r', k'\rangle)$ where $B''_{c} = \emptybag$, $n''_{c} = 0$, and for all $c'' \neq c$, $B''_{c''} = B'_{c''}$ and $n''_{c''} = n'_{c''}$.
Since the $c$th slot was holding the data point by which $s_1$ and $\mathsf{add}(s_1, c, d)$ differ and $s_1$ differs from $\mathsf{swap}(s_1, c, d, d')$ for each value of $d'$, $\mathsf{add}(s_1, c, d) \wtrans^{r'} \mathsf{Dirac}(s'_1)$ and $\mathsf{swap}(s_1,c,d,d') \wtrans^{r'} \mathsf{Dirac}(s'_1)$ for all $d'$. 
We use $\beta$ that maps $s'_1$ to itself and nothing else to anything.
Furthermore, for the one state $s'_1$ in $\supp(\nu_1)$, $|\ln \nu_1(s'_1) - \ln \nu_2(\beta(s'_1))| = 0 = \delta$.  Thus, $\nu_1 \mathrel{\lift{=,0}} \nu_2$ where equality is trivially a $0$-unwinding relation.

\paragraph{Inductive Case: $j > 0$}
We consider cases depending on what type of action $a$ is to show that there exists $\delta$ in $[0,2j\epsilon]$ such that $\mu_1 \mathrel{\lift{\mathcal{R}_{s,d}^{2j\epsilon-\delta},\delta}} \mu_2$:
\begin{itemize}
\item Subcase: $a \in D$.  In this case, we prove that such a $\delta$ exists using $\delta = 0$.  That is, we prove that $\nu_1 \mathrel{\lift{\mathcal{R}_{s,d}^{2j\epsilon},0}} \nu_2$.
Since $a \in D$, $s_1$ has must have the form 
\[ \langle 08, \langle B'_0, \ldots, B'_{t-1}\rangle, \langle n'_{0}, \ldots n'_{t-1}\rangle, c', y', r', k'\rangle \]
We consider subsubcases:
\begin{itemize}

\item Subsubcase: $c = c'$ and $n_{c'} < v-1$.  In this case, both states $s_1$ and $s_2$ will store the data point $a$.  For that $c'$, $\nu_1 = \mathsf{Dirac}(\langle 08, \vec{B}'', \vec{n}'', c', a, r', k'\rangle)$ where $B''_{c'} = B'_{c'} \bagunion \{a\}$, $n''_{c'} = n'_{c'} + 1$, and for all $c'' \neq c'$, $B''_{c''} = B'_{c''}$ and $n''_{c''} = n'_{c''}$.  Similarly, $\nu_2 =  \mathsf{Dirac}(\langle 08, \langle \vec{B}''', \vec{n}''', c', a, r', k'\rangle)$ where either
\begin{enumerate}
\item $B'''_{c'} = B_{c'} \bagunion \{d\} \bagunion \{a\}$, $n'''_{c'} = n'_{c'}+2$, and for all $c \neq c'' \neq c'$, $B'''_{c''} = B'_{c''}$ and $n'''_{c''} = n'_{c''}$; or

\item $B'''_{c'} = B_{c'} \bagunion \{d\} - \{d'\} \bagunion \{a\}$, $n'''_{c'} = n'_{c'}+2$, and for all $c \neq c'' \neq c'$, $B'''_{c''} = B'_{c''}$ and $n'''_{c''} = n'_{c''}$ for some $d'$.
\end{enumerate}
Thus, for $s'_1 \in \supp(\nu_1)$ and $s'_2 \in \supp(\nu_2)$, $s'_2$ is either  $\mathsf{add}(s'_1, c, d)$ or $\mathsf{swap}(s'_1, c ,d, d')$ for some $d'$.

To show that $\mu_1 \mathrel{\lift{\mathcal{R}_{s,d}^{\epsilon'},0}} \mu_2$, we use the function $\beta$ that maps $s'_1$ to the state $s'_2$ and nothing else.  Since both $\nu_1$ and $\nu_2$ are Dirac distributions, that covers all of their supports and is a bijection.  It follows from  $s'_2$ being either $\mathsf{add}(s'_1, c, d)$ of $\mathsf{swap}(s'_1, c, d, d')$ for some $d'$ that $s'_1 \mathrel{\mathcal{R}_{s,d}^{2j\epsilon}} s'_2$.  Lastly, $|\ln \nu_1(s'_1) - \ln \nu_2(s'_2)| = |\ln 1 - \ln 1| = 0 \leq \delta$

\item Subsubcase: $c \neq c'$ and $n_{c'} < v$.  Mostly, as above.

\item Subsubcase: $n_{c'} = v$.  In this case, both states $s_1$ and $s_2$ will drop the data point $a$ and not store it.
For that $c'$, $\nu_1 = \mathsf{Dirac}(s_1)$ and $\nu_2 =  \mathsf{Dirac}(s_2)$  By assumption, $s_1 \mathrel{\mathcal{R}_{s,d}^{2j\epsilon}} s_2$.  $|\ln \nu_1(s_1) - \ln \nu_2(s_2)| = |\ln 1 - \ln 1| = 0 \leq \delta$

\item Subsubcase: $c = c'$ and $n_{c'} = v-1$.  If $s_2 = \mathsf{swap}(s_1,c,d,d')$ for some $d'$, then this subsubcase is the same as the first one.  Otherwise, the $s_1$ will store the data point, but $s_2 = \mathsf{add}(s_1, c, d)$ will not since it already has $n_{c'} + 1 = v$ data points.  Thus, $\nu_2 =  \mathsf{Dirac}(s_2)$ and $\nu_1 = \mathsf{Dirac}(s'_1)$ where $s'_1 = \langle 08, \langle B''_0, \ldots, B''_{t-1}\rangle, \langle n''_1, \ldots, n''_{t-1}\rangle, c', a, r', k'\rangle)$ where $B''_{c'} = B'_{c'} \bagunion \{a\}$, $n''_{c'} = n'_{c'} + 1$,  and for all $c'' \neq c'$, $B''_{c''} = B'_{c''}$ and $n''_{c''} = n'_{c''}$.  Thus, we have that $s_2 = \mathsf{swap}(s'_1, c, d, a)$.  
Thus, $s'_1 \mathrel{\mathcal{R}_{s,d}^{2j\epsilon}} s_2$.  
We use $\beta$ that maps $s'_1$ to $s_2$ and nothing else.  Since $|\ln \nu_1(s_1) - \ln \nu_2(s_2)| = 0$, $\nu_1  \mathrel{\lift{\mathcal{R}_{s,d,d'}^{2j\epsilon}, 0}} \nu_2$.
\end{itemize}

\item Subcase: $a \in R$.  In this case, we prove that such a $\delta$ exists using $\delta = 0$.  That is, we prove that $\nu_1 \mathrel{\lift{\mathcal{R}_{s,d}^{2j\epsilon},0}} \nu_2$.

Since $a \in R$, $s_1$ must have the form $\langle 16, \langle B'_0, \ldots, B'_{t-1}\rangle, \langle n_0, \ldots, n'_{t-1}\rangle, c', y', r', k'\rangle$.  Thus, $s_1 \wtrans^a \mathsf{Dirac}(s'_1)$ where 
\[ s'_1 = \langle 08, \langle B''_0, \ldots, B''_{t-1}\rangle, \langle n''_0, \ldots, n''_{t-1}\rangle, c'+1 \mod t, y', r', k'\rangle) \]
where $B''_{c+1 \mod t} = \emptybag$, $n''_{c+1 \mod t} = 0$, and for all $c'' \neq c+1 \mod t$, $B''_{c''} = B'_{c''}$ and $n''_{c''} = n'_{c''}$.

If $s_2 = \mathsf{add}(s_1,c,d)$, then $s_2 \wtrans^r \mathsf{Dirac}(s'_2)$ where
\[ s'_2 = \langle 16, \langle B''_0, \ldots, B''_{t-1}\rangle, \langle n''_{0}, \ldots n''_{t-1}\rangle c'+1 \mod t, y', r', k'\rangle)\] 
where $B''_{c+1 \mod t} = \emptybag$, $B''_c = B'_c \bagunion \{d\}$, $n''_{c+1 \mod t} = 0$, and for all $c'' \neq c+1 \mod t$, $B''_{c''} = B'_{c''}$ and $n''_{c''} = n'_{c''}$.
Since $j > 0$, $c+(t-j) \mod t \neq c$.  Thus, the slot by which $s_1$ differs from $s_2$ will remain unchanged, and $s'_1 = \mathsf{add}(s'_2,c,d)$.

By similar reasoning, if $s_2 = \mathsf{swap}(s_1,c,d,d')$ for some $d'$,  $s'_1 = \mathsf{swap}(s'_2,c,d,d')$.  Thus, either way, $s'_1 \mathrel{\mathcal{R}_{s,d}^{2j\epsilon}} s'_2$.
$\beta$ that maps the one state of $\supp(\nu_1)$ to the one state of $\supp(\nu_2)$ shows that $\nu_1 \mathrel{\lift{\mathcal{R}_{s,d}^{2j\epsilon},0}} \nu_2$ since $|\ln \mu_1(s'_1) - \ln \mu_2(\mathsf{add}(s'_1,c,d))| = |\ln 1 - \ln 1| = 0 \leq \delta$.

\item Subcase: $a \in Q$.  In this case, we prove that such a $\delta$ exists using $\delta = 2\epsilon$.  That is, we prove that $\nu_1 \mathrel{\lift{\mathcal{R}_{s,d}^{2j\epsilon-2\epsilon},2\epsilon}} \nu_2$.
In this case, $s_1$ has the form $\langle 08, \langle B'_1, \ldots, B'_t\rangle, c', y', r', k'\rangle$.  $\nu_1$ is such that
\[ \nu_1(\langle 16, \langle B'_0, \ldots, B'_{t-1}\rangle, \langle n'_0, \ldots n'_{t-1}\rangle, c', a, r'', \kappa_a\rangle) = \Pr\left[\kappa_{a}\left(\bigbagunion_{\ell=0}^{t-1} B'_{\ell}\right) = r''\right] \] 
and $\nu_1(s'_1) = 0$ for all other states $s'_1$.
$\nu_2$ is either such that 
\begin{multline*}
\nu_2(\langle 16, \langle B'_0, \ldots, B'_c \bagunion \{d\}, \ldots B'_{t-1}\rangle, \langle n'_0, \ldots n'_c+1, \ldots n'_{t-1}\rangle, c', a, r'', \kappa_a\rangle)\\ = \Pr\left[\kappa_{a}\left(\bigbagunion_{\ell=0}^{t-1} B'_{\ell} \bagunion \{d\}\right) = r''\right]
\end{multline*}
or
\begin{multline*} 
\nu_2(\langle 16, \langle B'_0, \ldots, B'_c \bagunion \{d\}-\{d'\}, \ldots B'_{t-1}\rangle, \langle n'_0, \ldots n'_{t-1}\rangle, c', a, r'', \kappa_a\rangle)\\ = \Pr\left[\kappa_{a}\left(\bigbagunion_{\ell = 0}^{t-1} B'_{\ell} \bagunion \{d\} - \{d'\}\right) = r''\right]
\end{multline*}
for some $d'$ and $\nu_1(s'_2) = 0$ for all other states $s'_2$.
Let $\mathcal{B}$ denote which of $\bigbagunion_{\ell = 0}^{t-1} B'_{\ell} \bagunion  \{d\}$ and $\bigbagunion_{\ell = 0}^{t-1} B'_{\ell} \bagunion \{d\} - \{d'\}$ it is.
Either way $\bigbagunion_{\ell = 0}^{t-1} B'_{\ell}$ and $\mathcal{B}$ differ by at most two elements
Since $\kappa_{a}$ has $\epsilon$-differential privacy, we know that for any $r''$, 
\begin{align*}
\Pr\left[\kappa_{a}\left(\bigbagunion_{\ell=0}^{t-1} B'_{\ell}\right) = r''\right] &\leq \e{2\epsilon} * \Pr\left[\kappa_{a}(\mathcal{B}) = r'' \right]
\end{align*}
Thus,
\begin{multline}
\nu_1(\langle 16, \langle B'_0, \ldots, B'_{t-1}\rangle, c', a, r'', \kappa_a\rangle) 
\\\leq \e{2\epsilon} * \nu_2(\langle 16, \langle B'_0, \ldots, B'_c \bagunion \{d\}, \ldots B'_{t-1}\rangle, c', a, r'', \kappa_a\rangle)\label{eqn:ooa-add}
\end{multline}
and
\begin{multline}
\nu_1(\langle 16, \langle B'_0, \ldots, B'_{t-1}\rangle, c', a, r'', \kappa_a\rangle) 
\\\leq \e{2\epsilon} * \nu_2(\langle 16, \langle B'_0, \ldots, B'_c \bagunion \{d\} - \{d'\}, \ldots B'_{t-1}\rangle, c', a, r'', \kappa_a\rangle)\label{eqn:ooa-swap}
\end{multline}

Similarly, 
\begin{multline}
\nu_2(\langle 16, \langle B'_0, \ldots, B'_c \bagunion \{d\}, \ldots B'_{t-1}\rangle, c', a, r'', \kappa_a\rangle) 
\\\leq \e{2\epsilon} * \nu_1(\langle 16, \langle B'_0, \ldots, B'_{t-1}\rangle, c', a, r'', \kappa_a\rangle)\label{eqn:oao-add}
\end{multline}
and
\begin{multline}
\nu_2(\langle 16, \langle B'_0, \ldots, B'_c \bagunion \{d\} - \{d'\}, \ldots B'_{t-1}\rangle, c', a, r'', \kappa_a\rangle) 
\\\leq \e{2\epsilon} * \nu_1(\langle 16, \langle B'_0, \ldots, B'_{t-1}\rangle, c', a, r'', \kappa_a\rangle)\label{eqn:oao-swap}
\end{multline}

To show that $\nu_1 \mathrel{\lift{\mathcal{R}_{s,d}^{2j\epsilon-2\epsilon},2\epsilon}} \nu_2$, we use a function $\beta$.  In the case where $s_2 = \mathsf{add}(s_1,c,d)$, $\beta$ maps each state $s'_1$ of $\supp(\mu_1)$ to $\mathsf{add}(s'_1,c,d)$.  
To show that $\beta$ is a bijection from $\supp(\mu_1)$ to $\supp(\mu_2)$ note that $\mathsf{add}(\cdot, c, d)$ is a bijection and that Lines~\ref{eqn:ooa-add} and~\ref{eqn:oao-add} imply that $s'_1$ is in $\supp(\mu_1)$ iff $\mathsf{add}(s'_1,c,d)$ is in $\supp(\mu_2)$.  

In the case where $s_2 = \mathsf{swap}(s_1,c,d,d')$, $\beta$ maps each $s'_1$ to $\mathsf{swap}(s'_1,c,d,d')$.
To show that $\beta$ is a bijection from $\supp(\mu_1)$ to $\supp(\mu_2)$ note that $\mathsf{swap}(\cdot, c, d,d')$ is a bijection and that Lines~\ref{eqn:ooa-swap} and~\ref{eqn:oao-swap} imply that $s'_1$ is in $\supp(\mu_1)$ iff $\mathsf{swap}(s'_1,c,d,d')$ is in $\supp(\mu_2)$.  

Since $\mathcal{R}_{s,d}^{2(j-1)\epsilon} = \mathcal{R}_{s,d}^{2j\epsilon-2\epsilon}$, for all $r''$, 
$s'_1 \mathrel{\mathcal{R}_{s,d}^{2j\epsilon-2\epsilon}} \beta(s'_1)$.
Furthermore, for all $s'_1$ in $\supp(\mu_1)$, $|\ln \mu_1(s'_1) - \ln \mu_2(\beta(s'_1))| \leq \epsilon \leq 2\epsilon \leq \delta$ from Lines~\ref{eqn:ooa-add}, \ref{eqn:ooa-swap}, \ref{eqn:oao-add}, and~\ref{eqn:oao-swap}.
\end{itemize}

This completes the proof of the lemma.

Since $\mathcal{R}^{2j*\epsilon}_{s,d}$ covers $s$ and $d$ for all states $s$ and data points $d$ of the automaton $\exparamM$, Lemma~\ref{lem:near-system-unwinding} and Theorem~\ref{thm:unwinding} together prove that the automaton has $(2t*\epsilon)$-differential noninterference.

\section{The $\mathtt{isInLiftedRelation}$ Algorithm}
\label{app:lift-algo}

The reduction used by $\mathtt{isInLiftedRelation}$ is shown in Figure~\ref{fig:isInLifetedRelation}.  
\begin{figure}
\mbox{}$\mathtt{isInLiftedRelation}(S_\bot, \mathsf{R}, \delta, \nu_1, \nu_2)$\\
\mbox{}\s $V_{\mathsf{L}} := \{\}$\\
\mbox{}\s $V_{\mathsf{R}} := \{\}$\\
\mbox{}\s $E := \{\}$\\
\mbox{}\s for all $x_1 \in S_\bot$\\
\mbox{}\s\s if $\nu_1(x_1) > 0$,\\
\mbox{}\s\s\s add $x_1$ to $V_{\mathsf{L}}$\\
\mbox{}\s for all $x_2 \in S_\bot$\\
\mbox{}\s\s if $\nu_2(x_2) > 0$,\\
\mbox{}\s\s\s add $x_2$ to $V_{\mathsf{R}}$\\
\mbox{}\s for all $x_1 \in V_{\mathsf{L}}$\\
\mbox{}\s\s for all $x_2 \in V_{\mathsf{R}}$\\
\mbox{}\s\s\s if($x_1 \mathsf{R} x_2$ and $|\ln \nu_1(x_1) - \ln \nu_2(x_2)| \leq \delta)$\\
\mbox{}\s\s\s\s add edge $\langle x_1, x_2\rangle$ to $E$\\
\mbox{}\s return $\mathtt{HopcroftKarpHasPerfectMatching}(\langle V_{\mathsf{L}}, V_{\mathsf{R}}, E\rangle$)
\caption{Algorithm for checking  $\delta$-approximate lifting of relations.}
\label{fig:isInLifetedRelation}
\end{figure}
First the algorithm constructs the bipartite graph for the reduction and then uses the Hopcroft-Karp algorithm~\cite{hk73n52}.  This algorithm returns if and only if there exists a \emph{perfect matching} $M$ for the graph.  A perfect matching $M$ for a bipartite graph $\langle V_{\mathsf{L}}, V_{\mathsf{R}}, E\rangle$ is a subset of $E$ such that for every vertex $v \in V = V_{\mathsf{L}} \cup V_{\mathsf{R}}$ is incident to exactly one edge in $M$.  

(Since $\supp(\nu_1)$ and $\supp(\nu_2)$ might not be disjoint, but $V_{\mathsf{L}}$ and $V_{\mathsf{R}}$ must be disjoint, we should tag the states $x_1$ and $x_2$ differently before adding them to the sets.  However, for readability, we do not explicitly do this tagging.)

\begin{proposition}\label{prp:lift-algo-correct}
For all sets $S$, relations $\mathsf{R}$ over $S$, non-negative reals $\delta$, and distributions $\nu_1$ and $\nu_2$ over $S$,
$\mathtt{isInLiftedRelation}(S, \mathsf{R}, \delta, \nu_1, \nu_2)$ returns true iff $\nu_1 \mathrel{\lift{\mathsf{R},\delta}} \nu_2$.
\end{proposition}
\begin{proof}
By the correctness of the Hopcroft-Karp algorithm, $\mathtt{HopcroftKarpHasPerfectMatching}$ (and, thus, $\mathtt{isInLiftedRelation}$) will only return true if there exists a perfect matching $M$ for the graph.  

To prove only-if direction, assume that such an $M$ exists.  
Given a perfect matching $M$ of bipartite graph, for every $x_1 \in V_{\mathsf{L}}$ there exists a unique edge $e \in E$ such that there exists a $x_2 \in V_{\mathsf{R}}$ such that $e = \langle x_1, x_2\rangle$.  For each such $x_1$, denote the unique $x_2$ paired with it by this edge as $\beta_M(x_1)$.  $\beta_M$ is a function from $V_{\mathsf{L}}$ to $V_{\mathsf{R}}$ since for every $x_1 \in V_{\mathsf{L}}$, there exists exactly one such edge and, thus, exactly one such $x_2$, which must be in $V_{\mathsf{R}}$ since the graph is bipartite.  Furthermore, $\beta_M$ is a bijection since every $x_2$ in $V_{\mathsf{R}}$ must be incident to exactly one edge in the perfect matching $M$.

Since $V_{\mathsf{L}} = \supp(\nu_1)$ and $V_{\mathsf{R}} = \supp(\nu_2)$, $\beta_M$ is a bijection from $\supp(\nu_1)$ to $\supp(\nu_2)$.
Since $x_1$ and $\beta_M(x_1)$ are connected by an edge, $x_1 \mathrel{\mathsf{R}} \beta_M(x_1)$ and $|\ln \nu_1(x_1) - \ln \nu_2(\beta_M(x_1))| \leq \delta$.
Thus, the bijection $\beta_M$ is such that for all $x_1 \in \supp(\nu_1)$, $x_1 \mathrel{\mathsf{R}} \beta_M(x_1)$ and $|\ln \nu_1(x_1) - \ln \nu_2(\beta_M(x_1))| \leq \delta$.  This implies that $\nu_1 \mathrel{\lift{\mathsf{R},\delta}} \nu_2$.

To prove the if-direction, assume that $\nu_1 \mathrel{\lift{\mathsf{R},\delta}} \nu_2$.  Then there exists a bijection $\beta$ from $\supp(\nu_1)$ to $\supp(\nu_2)$ such that $x_1 \mathrel{\mathsf{R}} \beta(x_1)$ and $|\ln \nu_1(x_1) - \ln \nu_2(\beta(x_1))| \leq \delta$.  Let $M_\beta$ be the set such that $\langle x_1, x_2\rangle \in M_\beta$ iff $\beta(x_1) = x_2$.  $M_\beta$ is a subset of $E$ since $x_1 \in \supp(\nu_1)$, $\beta(x_1) \in \supp(\nu_2)$, $x_1 \mathrel{\mathsf{R}} \beta(x_1)$, and $|\ln \nu_1(x_1) - \ln \nu_2(\beta(x_1))| \leq \delta$ together imply that $\langle x_1, \beta(x_1)\rangle$ is in $E$.  $M_\beta$ is a perfect matching for the graph since $\beta$ is a bijection from $V_{\mathsf{L}} = \supp(\nu_1)$ to $V_{\mathsf{R}} = \supp(\nu_2)$.
\end{proof}

\begin{proposition}\label{prp:lift-algo-runtime}
$\mathtt{isInLiftedRelation}$ runs in $O(|S|^{2.5})$ time.
\end{proposition}
\begin{proof}
Given that we know that we never will attempt to add a duplicate element to any of the sets $V_{\mathsf{L}}$, $V_{\mathsf{R}}$, nor $E$, all the set operations may be done in constant time.  Thus, constructing the graph for the reduction operates in $O(|S|^2)$ time.  The Hopcroft-Karp perfect matching algorithm operates in $O(\sqrt{v} * e)$ time where $v$ is the number of vertices and $e$, the number of edges.  That is lower than $O(|S|^{2.5})$ since $e \leq v^2$ and $v = |V_{\mathsf{L}}| + |V_{\mathsf{R}}| \leq 2*|S|$.  Thus, the whole algorithm runs in $O(|S|^{2.5})$ time.
\end{proof}

\section{Proofs for the Checking Algorithm}
\label{app:checking-algo-proofs}

\paragraph{Proof of Lemma~\ref{lem:checking-algo-uw-sound}: The Soundness of $\mathtt{isUnwindFam}$}
Here $\rel$ represents the relation family $\mathcal{R}$ such that $\mathcal{R}^{\epsilon}$ is equal to $\rel[\lfloor \epsilon/\delta \rfloor]$ for $\epsilon$ such that $0 \leq \epsilon \leq t*\delta$.  If such a family is an unwinding family for transition system, then it is also one for the transition system with all the hidden states have been converted to the same one.

We prove a stronger fact that implies that $\mathcal{R}$ is an $(t*\delta)$-unwinding family for the converted transition system.  Namely, we show that the algorithm will only return true if for all $\epsilon$ from $[0,t\delta]$, for all $x_1$ and $x_2$ in $S_\bot$ such that $\langle x_1, x_2 \rangle \in \rel[\lfloor \epsilon/\delta\rfloor]$, for all $a$ in $I \cup R$, there exists $\nu_1$ such that $x_1 \wtrans^a \nu_1$ iff there exists $\nu_2$ such that $x_2 \wtrans^a \nu_2$, and when they do exist, either (1) $\nu_1 \mathrel{\lift{\rel[\lfloor \epsilon/\delta\rfloor]-0,0}} \nu_2$ or (2) $\nu_1 \mathrel{\lift{\rel[\lfloor\epsilon/\delta\rfloor-\delta],\delta}} \nu_2$.
Condition (1) is satisfied if for all $\langle x'_1, x'_2\rangle \in \rel[\lfloor \epsilon/\delta\rfloor]$, $\nu_1(x'_1) = \nu(x'_2)$.
Condition (2) is satisfied if for all $\langle x'_1, x'_2\rangle \in \rel[\lfloor \epsilon/\delta\rfloor-\delta]$, $|\ln \nu_1(x'_1) - \ln \nu(x'_2)| \leq \delta$.

The algorithm will only return true if none of the preceding \verb|return| statements return false.  Firstly, it must be the case that $|\rel| = t+1$.
Secondly, the outer most \verb|for| loop must finish executing without any of its \verb|return| statements being reached.  This will only happen if for all values of $i$ from the length of the array.  For each such value, the algorithm examines the relation $\rel[i]$, which is the relation used for all values of $\epsilon$ in $[0,t\delta]$ such that $\lfloor \epsilon/\delta\rfloor$ is equal to $i$.  Thus, by considering each value of $i$, the algorithm examines the intervals $[0,\delta)$, $[\delta, 2\delta)$, and so on up to $[(t-1)\delta, t\delta)$ and finally the point at $t\delta$.  Thus, it examines the whole range $[0,t\delta]$ as required by the above condition.

Each of these examinations consists of looking at every pair of $\langle x_1, x_2\rangle$ in the relation $\rel[i]$, and every action $a$ in $I \cup O$.  For each such action $a$ and pair, the algorithm first returns false if it is not the case that $x_1 \wtrans^a \nu_1$ iff $x_2 \wtrans^a \nu_2$ for some $\nu_1$ and $\nu_2$ since $T[x_1][a]$ is equal to $\nil$ only in the case where $x_1 \wtrans^a \nu_1$ for no $\nu_1$ (and likewise for $x_2$).

If false was not returned, the algorithm checks if it was because $\nu_1$ and $\nu_2$ both exist.  If this is not the case, the examination finishes as nothing more must be shown for this state-action pair.

In the case where $\nu_1$ and $\nu_2$ do exist, we know that $x_1$ and $x_2$ are actual states and $\bot$ since $T[\bot][a] = \nil$ for all $a$.  The examination then continues with the algorithm computing the values of $\nu_1$ and $\nu_2$ such that $x_1 \wtrans^a \nu_1$ and $x_2 \wtrans^a \nu_2$ as described above, which is well defined since $x_1$ and $x_2$ are actual states.  

Next, it checks if $\nu_1 \mathrel{\lift{\rel[i],0}} \nu_2$.  $\mathtt{isInLiftedRelation}(S_\bot, \rel[i], 0, \nu_1, \nu_2)$ will return true iff Condition (1) is satisfied.  If Condition (1) is satisfied, the examination is complete and algorithm does not return false on this execution of the loop's body.  

If Condition (1) is not satisfied, then algorithm next checks to see if Condition (2) holds.  For our restricted set of relation families, Condition (2) cannot hold if $i$ is $0$ and Condition (1) does not hold.  Thus, the next \verb|if| statement.  It uses $\mathtt{isInLiftedRelation}(S_\bot, \rel[i-1], \delta, \nu_1, \nu_2)$ to check if Condition (2) holds.  If any pair is not, the algorithm returns false.  If Condition (2) is satisfied, the examination is complete, and algorithm does not return false and this execution of the loop.

Thus, each execution of the loop will only complete without returning false if either Condition (1) or Condition (2) holds.  As the loop checks all the needed combinations of states and actions, the algorithm will only return true if the stronger fact that implements $\mathcal{R}$ is an unwinding relation is true.  

\paragraph{Proof of Lemma~\ref{lem:checking-algo-uw-time}: The Running Time of $\mathtt{isUnwindFam}$}
The conversion of all hidden actions to the same one runs in $O(|H|*|S|)$.

The outer most loop runs over the whole length of $\rel$.  The next loop is over every pair in $\rel[i]$ where $\rel[i]$ is a binary relation over states.  Thus, there are at most $|S|^2$ pairs in $\rel[i]$.  The next loop is over every action.  Thus, the body of this loop will be executed $O(t*|S|^2*|A|)$ times.

This body consists of four parts.  The first is a simple conditional taking constant time.  The second computes $\nu_1$ and $\nu_2$.  This takes $O(|H|*|S|+|S|^3)$ time.  Since the conversion of all hidden actions to the same one takes $|H|=1$, this is $O(|S|^3)$.  The third is a calls $\mathtt{isInLiftedRelation}$, which takes $O(|S|^{2.5})$ time.  The forth is a conditional and another call to $\mathtt{isInLiftedRelation}$ on $\rel[i-1]$, which takes $O(|S|^{2.5})$ time.   Thus, body is $O(|S|^3)$ time and the whole loop is $O(t*|A|*|S|^4)$.

The algorithm whole algorithm run in $O(|H|*|S| + t*|A|*|S|^4)$, which is $O(t*|A|*|S|^4)$ since $|H| \leq |A|$.

\paragraph{Proof of Theorem~\ref{thm:checking-algo-cover-sound}: The Soundness of $\mathtt{isAllCovered}$}
The algorithm will only return true if none of the preceding \verb|return| statements return false.  That is, the outer most \verb|for| loop must finish executing without any of its \verb|return| statements being reached.  This will only happen if for every reachable state $s$ and every data point $d$, either $T[s][d] = \nil$ or each of the following is true:
\begin{enumerate}
\item $\nu(\bot) = 0$ and $|\rels[s][d]| \neq t=1$;
\item for all states $s'$ such that $s' \in \supp(\nu)$, $\langle s, s'\rangle \in \rels[s][d]$; and
\item $\mathtt{isUnwindFam}(\langle S, I, O, T\rangle, \rels[s][d], \delta, t)$ returns true
\end{enumerate}
where $s \wtrans^d \nu$.  
In the case where $T[s][d] = \nil$, the trivial relation family that consists of only empty relations is a $(t*\delta)$-unwinding family for the automaton. 
In the case where $T[s][d] \neq \nil$, the three conditions above imply $\rels[s][d]$ is a $(t*\delta)$-unwinding family for the automaton by using Lemma~\ref{lem:checking-algo-uw-sound} on the last condition.  Either way, there exists a $(t*\delta)$-unwinding family that covers $s$ and $d$.  Thus, the body of the loop will return false unless there exists such an unwinding family.

As the algorithm checks every reachable $s$ for every $d$, the loop will not terminate without returning false unless the conditions of Theorem~\ref{thm:unwinding} holds.  Thus, the algorithm only returns true if the automaton has $(t*\delta)$-differential noninterference.

\paragraph{Proof of Theorem~\ref{thm:checking-algo-cover-time}: The Running Time of $\mathtt{isAllCovered}$}
Computing the reachable states can be done in time $O(|S|)$.

The outer most loop executes at most $|S|$ times.  The next loop executes at most $|D|$ times.  In the case where $T[s][d] \neq \nil$, the body takes $O(S^3)$ time to compute $\nu$, $O(|S|)$ for the inner loop, and $O(t*|A|*|S|^4)$ time for running the $\mathtt{isUnwindFam}$ algorithm (Lemma~\ref{lem:checking-algo-uw-time}).  Thus, the body takes $O(t*|A|*|S|^4)$ time and the whole algorithm takes $O(t*|D|*|A|*|S|^5)$ time.

\end{document}